\documentclass[11pt, letterpaper]{article}

\usepackage{CJKutf8}

\usepackage[whole]{bxcjkjatype}

\usepackage{algorithm}
\usepackage[noend]{algpseudocode}

\usepackage{cite} %
\usepackage[pdftex]{graphicx}
\usepackage{threeparttable}
\usepackage[margin=1in]{geometry}
\usepackage{amsmath,amssymb,amsfonts,setspace} %
\usepackage{amsthm} %
\usepackage{latexsym}
\usepackage{comment}
\usepackage{bm}
\usepackage{enumerate}
\usepackage{bookmark} 
\usepackage{type1cm} 
\usepackage{color} 
\usepackage{thm-restate}

\theoremstyle{plain}
\newtheorem{theorem}{Theorem}[section]
\newtheorem{lemma}[theorem]{Lemma}
\newtheorem{corollary}[theorem]{Corollary} 
\newtheorem{proposition}[theorem]{Proposition}

\theoremstyle{definition}
\newtheorem{definition}[theorem]{Definition}

\newcounter{cntLemmaNumber}

\newcounter{cntTheoremNumber}

\newcommand{\Dcal}{\mathcal{D}}

\newcommand{\Ical}{\mathcal{I}}
\newcommand{\Tcal}{\mathcal{T}}

\newcommand{\Qcal}{\mathcal{Q}}

\newcommand{\Odd}{\mathrm{odd}}
\newcommand{\Even}{\mathrm{even}}
\newcommand{\OL}[1]{\overline{#1}}

\newcommand{\dist}{\mathrm{dist}}

\newcommand{\Op}{\alpha} %
\newcommand{\Nop}{\beta} %

\newcommand{\Vardist}{\textsf{Dist}}
\newcommand{\Varparent}{\textsf{Par}}

\newcommand{\Varnontree}{\textsf{Nlist}}
\newcommand{\Varcontraction}{\textsf{C}}
\newcommand{\Varchecklist}{\textsf{R}}

\newcommand{\Varort}{\textsf{a}}
\newcommand{\Varvol}{\textsf{Vol}}

\newcommand{\Outset}{\ensuremath{\tilde{E}}}
\newcommand{\Minset}{\ensuremath{\tilde{E}^{\min}}}
\newcommand{\Ccset}{\ensuremath{E^{\ast}}}
\newcommand{\Minedge}{\ensuremath{e^{\ast}}}

\newcommand{\level}{\ensuremath{\mathsf{level}}}
\newcommand{\vlevel}{\ensuremath{\mathsf{vlevel}}}

\newcommand{\hlevel}{\ensuremath{\mathsf{hlevel}}}

\newcommand{\parent}{\ensuremath{\mathsf{p}}}
\newcommand{\parset}{\ensuremath{\mathsf{P}}}
\newcommand{\paredgeset}{\ensuremath{\mathsf{EP}}}

\newcommand{\freeroot}{\ensuremath{\mathsf{r}}}

\newcommand{\Augt}{\ensuremath{\mathsf{Aug}}}

\newcommand{\Idle}{\ensuremath{\textsf{idle}}}
\newcommand{\Dead}{\ensuremath{\textsf{dead}}}
\newcommand{\Active}{\ensuremath{\textsf{active}}}

\newcommand{\Mmax}{\mu}

\newcommand{\InvMmax}{\OL{\mu}}

\newcommand{\ot}{\leftarrow}

\DeclareMathOperator*{\argmin}{arg\,min}
\DeclareMathOperator*{\argmax}{arg\,max}

\allowdisplaybreaks[2]

\title{Forgetting Alternation and Blossoms:\\A New Framework for Fast Matching Augmentation and\\Its Applications to Sequential/Distributed/Streaming Computation}
\author{Taisuke Izumi\thanks{Osaka University. Emails: izumi.taisuke.ist@osaka-u.ac.jp, \{n-kitamura, yutaro.yamaguchi\}@ist.osaka-u.ac.jp}
\and Naoki Kitamura$^*$
\and Yutaro Yamaguchi$^*$
}

\date{}

\begin{document}

\maketitle
\thispagestyle{empty}

\begin{abstract}
Finding a maximum cardinality matching in a graph is one of the most fundamental problems in the field of combinatorial optimization and graph algorithms.
An algorithm proposed by Micali and Vazirani (1980) is well-known to solve the problem in $O(m\sqrt{n})$ time, where $n$ and $m$ are the numbers of vertices and of edges, respectively, which is still one of the fastest algorithms in general.
While the MV algorithm itself is not so complicated and is indeed convincing, its correctness proof is extremely challenging, which can be seen from the history: after the first algorithm paper had appeared in 1980, Vazirani has made several attempts to give a complete proof for more than 40 years.
It seems, roughly speaking, caused by the nice but highly complex structure of the shortest alternating paths in general graphs that are deeply intertwined with the so-called (nested) blossoms.

In this paper, we propose a new structure theorem on the shortest alternating paths in general graphs without taking into the details of blossoms.
The high-level idea is to forget the alternation (of matching and non-matching edges) as early as possible.
A key ingredient is a notion of alternating base trees (ABTs) introduced and utilized by the recent work of Izumi, Kitamura, and Yamaguchi (2024) to develop a nearly linear-time algorithm under one of the standard distributed computation models.
Our structure theorem refines the properties of ABTs exploited in their algorithm, and we also give simpler alternative proofs for them.
Based on our structure theorem, we propose a new algorithm under the standard sequential computation model, which is slightly slower but more implementable and much easier to confirm its correctness than the MV algorithm.

As applications of our framework, we also present new $(1 - \epsilon)$-approximation maximum matching algorithms for general graphs in the distributed and semi-streaming settings. In the CONGEST model, which is one of the standard distributed computation models, the proposed deterministic algorithm runs in $\tilde{O}(\epsilon^{-4})$ rounds. The semi-streaming algorithm uses $O(\epsilon^{-4})$ passes and $O(n \log n)$-bit memory. Both algorithms are deterministic, and substantially improve the best known upper bounds.
The algorithms are built on the top of a novel framework of amplifying approximation factors of given matchings, which is of independent interest.

\end{abstract}

\setcounter{page}{0}
\newpage
\section{Introduction}

\subsection{Background}

\emph{Matching} is often recognized as one of the central topics in graph theory and algorithms, and particularly finding a maximum cardinality matching (MCM) is a very fundamental problem in that context.
In this paper, we consider the MCM problem for general graphs. %
While the problem is known to be polynomial-time solvable, we still lack the clear understanding and exposition of efficient MCM algorithms.
Looking back to the history of sequential algorithms, the seminal blossom algorithm by Edmonds~\cite{Edmonds} opened the door to polynomial-time algorithms for this problem. 
After a few follow-up improvements~\cite{balinski1967labelling, gabow1976efficient, even1975n2}, the algorithm by Micali and Vazliani (MV algorithm) has been proposed in 1980~\cite{MV1980}, which runs in $O(m \sqrt{n})$ time (where $m$ and $n$ are the number of edges and vertices of the input graph).
This still achieves one of the best provable running time bound of deterministic computation of MCMs in general graphs, while a slightly better bound (by at most $\log n$ factor) for dense graphs is known~\cite{goldberg2004maximum}.
However, the correctness proof of this algorithm is also known to be very complex: %
Since the original paper~\cite{MV1980} lacks its full correctness proof, several follow-up arguments of ``dispelling'' its complication are attempted by one of the authors~\cite{vazirani1994theory,Vazirani12,vazirani2024theory}, and the state-of-the-art complete proof results in the paper of almost 40 pages~\cite{vazirani2024theory}. 
There are a few attempts of exploring simpler algorithms~\cite{Blum1990, GT91, gabow2017weighted}, but some of them seems incomplete (cf.~\cite{Vazirani12}) and some is via the weighted problem, which is interesting but somewhat hides the essence of tractable feature of MCM itself.

We briefly review the approach of the MV algorithm, which is the most purely graph-theoretic algorithm attaining $O(m\sqrt{n})$ time and related to our result.
The high-level strategy is to find a maximal set of disjoint \emph{shortest augmenting paths}, resorting to the seminal analysis by Hopcroft and Karp~\cite{HK73}.
More precisely, they prove that $O(\sqrt{n})$ iterations of matching augmentation by maximal sets of disjoint shortest augmenting paths finally reaches an MCM.
The MV algorithm computes such a set in $O(m)$ time for general graphs, which is the most complicated part.

The difficulty lying behind this complication is threefold:
First, shortest alternating paths in general (i.e., non-bipartite) graphs does not satisfy the property of \emph{BFS-honesty}: given a shortest alternating path, its subpath is not necessarily the shortest among the alternating paths of the same parity.
This obstacle prevents us from utilizing the standard tree-growth approach like as BFS, although it works 
for bipartite matching case.
If we do not need a shortest augmenting path, this obstacle is addressed by the approach based on \emph{blossoms}, which is originally introduced by Edmonds~\cite{Edmonds}.
In that approach, the algorithm makes a BFS-like forest rooted by free vertices (i.e., vertices without incident matching edges), and if the process encounters a blossom, the algorithm shrinks its odd cycle part into a single vertex and recurses.
This operation preserves the existence of augmenting paths, and thus finally the algorithm finds it in the shrunk graph (if any).
That path might contain a vertex associated with a shrunk odd cycle, and then the corresponding augmenting path in the original instance is obtained by expanding such vertices one by one.

The fundamental idea of Edmonds' algorithm is clean and easy to follow, but there is a critical shortcoming that it does not necessarily provide a \emph{shortest} augmenting path, mainly due to the existence of \emph{nested} blossoms.
This causes the second difficulty, because it is crucial to find shortest augmenting paths to apply the analysis by Hopcroft and Karp~\cite{HK73}.
The operation of shrinking an odd cycle obscures the structure of shortest augmenting paths, and thus the MV algorithm handles blossoms carefully by introducing a notion of blossom size (called the \emph{tenacity}).
Speaking very informally, the MV algorithm tries to manage the two (shortest) alternating paths terminating with a matching edge and a non-matching edge, which respectively referred to as an $\Even$-alternating path and an $\Odd$-alternating path.
The notion of blossom size is utilized in the MV algorithm for identifying a \emph{smallest} blossom containing each vertex $t$, which also deduces the shortest $\Even$- and $\Odd$-alternating paths to $t$: intuitively, given a vertex $t$ and a blossom $B$ containing $t$ strictly in its odd cycle, $B$ contains both $\Even$- and $\Odd$-alternating paths. Under some appropriate definitions of blossoms and their sizes, those paths are guaranteed to be the shortest if $B$ is the smallest.

However, as mentioned in~\cite{vazirani2024theory}, this approach faces a ``chicken-and-egg'' problem, because we want to identify smallest blossoms to identify the shortest alternating paths, but we need the information of the length of shortest alternating paths to each vertex to identify the smallest size blossoms.
Furthermore, the situation is more complicated when trying to take a maximal set of shortest augmenting paths.
Consider the situation that a shortest augmenting path $P$ is found.
To establish a maximal set of disjoint shortest augmenting paths, we further need to find another shortest augmenting path in the graph after the removal of $P$.
However, if a vertex $v$ (not in $P$) is not contained in any shortest augmenting paths after the removal of $P$, $v$ must be also removed for attaining $O(m)$ running time (leaving it will waste the search cost).
In other words, the algorithm needs to know quickly which vertices still have a chance of organizing a shortest augmenting path.
To overcome these issues, the correctness proof of the MV algorithm is required to introduce a bunch of extra notions and needs a heavy inductive argument~\cite{vazirani2024theory}.

\subsection{Our Results}
The key contribution of this paper is to present a new simple structure theorem for a maximal set of disjoint shortest augmenting paths, based on the approach recently proposed by Izumi, Kitamura, and Yamaguchi~\cite{IKY24}.
Different from the traditional approaches stated above (Edmonds and MV), it does not maintain an \emph{alternating forest} (with blossoms), which is a collection of trees rooted by a free vertex such that the path from roots to each vertex is a (shortest) alternating path.
Instead, it maintains yet another tree structure called \emph{alternating base tree (ABT)}\footnote{To be precise, this should be called an alternating base forest in an analogous way to an alternating forest, but throughout this paper we consider a single unified tree containing all components of such a forest by adding a super root $f$ to the original graph. Hence we use terminology tree instead of forest.}.
Intuitively, an alternating base tree is a collection of the edges in the BFS-like tree growth part of Edmonds' algorithm.
In other words, an ABT is a tree such that the edge from a vertex $u$ to its parent is the last edge of a shortest alternating path from the free root vertex.
The concept of ABTs is originally presented by Kitamura and Izumi~\cite{KI22}, and its structural property is extended to find a \emph{single} shortest 
augmenting path~\cite{IKY24}.
In this paper, we further generalize it so that one can utilize ABTs for computing a maximal set of disjoint shortest augmenting paths by extending the tree structure into the DAGs, as well as a shorter and concise proof of the structure theorem in~\cite{IKY24}.
For explaining its impact, we informally state our extended structure theorem. 

\begin{theorem}[informal]  \label{thm:structuralTheorem}
Let $U$ be the set of all free nodes in the current matching system $(G, M)$, where $G = (V(G), E(G))$ is the input graph and $M \subseteq E(G)$ is the current matching. 
Define an \emph{alternating base DAG (ABD)} $H$ as the DAG such that the vertex set is $V(G)$ and an edge from $u$ to $v$ is in $H$ if and only if $v$ is the immediate predecessor of $u$ in a shortest $U$--$u$ alternating path.\footnote{For the reader familiar with the MV algorithm, this $H$ can be explained as the DAG consisting of the props (the edges giving the minlevel of one of the endpoints), in which the procedure DDFS (double depth first search) runs in the base case (before shrinking any blossoms). The important difference is that, in the MV algorithm, after shrinking some blossom, the DAG in which DDFS runs should be updated by changing the heads of some edges to skip the shrunk blossom later, but we do not need to care about such things in this paper.}
Let $\Ccset \subseteq E(G)$ be some appropriate set of edges not contained in $H$ (which is formally defined later).\footnote{For the reader familiar with the MV algorithm, again, this $E^*$ can be explained as the set of bridges for which DDFS results in a shortest augmenting path. We give a different characterization for such edges in terms of ABTs.}
A \emph{double path} is two vertex-disjoint paths to $U$ staring from the two endpoints of an edge in $\Ccset$.
Given a maximal set $\Dcal$ of vertex-disjoint double paths in $H$, it can be efficiently transformed into a maximal set of vertex-disjoint shortest augmenting paths in $(G, M)$ by a very simple algorithm.
\end{theorem}

Following the theorem above, we actually provide yet another MCM algorithm running in $O(m \sqrt{n} \log n)$ time.
It is slightly worse than the MV algorithm, but more structured and admits a shorter and self-contained correctness proof.
A key distinguished advantage of our algorithm is higher modularity: With the aid of our structure theorem, the construction of a maximal set of disjoint shortest augmenting paths is decomposed into almost independent two phases.
The first phase is the construction of $H$. 
In this phase, we only need to know the last edges of shortest $U$--$u$ alternating paths to each vertex $u$, which can be easily identified only from the information of the lengths of shortest $\Odd$- and $\Even$-alternating paths from $U$ to each vertex (cf.~Lemma~\ref{lma:fundamental}).
Interestingly, this is only the part where we have to pay attention to alternation of paths.
In the second phase, one can completely forget alternating paths, where the task is reduced to a slight variant of the standard maximal disjoint path construction.
This simplification allows us to start the ``plain'' DDFS in the obtained DAGs without caring path alternation or blossoms (cf.~\cite[Sections 2, 5.3, and 8.3.1]{vazirani2024theory}), which mitigates the complication the MV algorithm faces.

As other applications of our framework, we also provide new $(1 - \epsilon)$-approximation maximum matching algorithms
working in the distributed and semi-streaming settings. We obtain the following theorem:
\begin{theorem} \label{thm:mainApproximate}
There exists:
\begin{itemize}
\item a $(1 - \epsilon)$-approximation maximum matching algorithm for a given graph $G$ which runs in the CONGEST model
with $O(\epsilon^{-4} \mbox{\rm\textsf{MM}}(\epsilon^{-2}n))$ rounds, where $\mbox{\rm\textsf{MM}}(N)$ means the time complexity of computing
a maximal matching in graphs on $N$ vertices;
the algorithm is deterministic except for the part of computing maximal matchings.
\item a deterministic $(1 - \epsilon)$-approximation maximum matching algorithm for a given graph $G$ which runs in the semi-streaming model
with $O(\epsilon^{-4})$ passes and $O(n \log n)$-bit memory. 
\end{itemize}

\end{theorem}

Using the known best algorithm, $\mbox{\rm\textsf{MM}}(N)$ becomes $O(\log N)$ in the randomized case, and $O(\log^{3/2} N)$
in the deterministic case~\cite{GG23}. Our first result in the CONGEST model substantially improves the known best bound of 
$\tilde{O}(\epsilon^{-10} \mbox{\rm\textsf{MM}}(n))$ rounds~\cite{MMSS25}. In the semi-streaming model, there are two perspectives
on pass complexity --- complexity depending on $n$, and independent of $n$. Our main focus is the latter one. Along that line, 
the known best complexity is $O(\epsilon^{-6})$ passes and $O(\epsilon^{-1}n \log n)$-bit space\cite{MS25}, and our second result
also attains the substantial improvement. 

As utilized in many approximation algorithms, our algorithm relies on the seminal analysis by Hopcroft and 
Karp~\cite{HK73}, which states that if a matching $M$ does not admit shortest augmenting paths 
of length less than $\lceil 2\epsilon^{-1} \rceil$, one can conclude that $M$ is a $(1 - \epsilon)$-approximate
matching. Since augmentation by a maximal set of (vertex-)disjoint shortest augmenting paths increases the length of 
shortest augmenting paths at least by two, $O(\epsilon^{-1})$ iterations of the maximal set augmentation 
yields the desired algorithm. It should be emphasized that deducing these algorithms is not straightforward even
utilizing our framework. In more details we face two technical difficulties for obtaining those: complication arising 
from alternating paths in general graphs, and simultaneous construction of multiple augmenting paths with 
guaranteeing maximality. While the first difficulty is substantially mitigated by our framework, the second one 
is still challenging in distributed and semi-streaming settings. For coping with that, we present a new common framework
of amplifying approximation factors of given matchings, which is of independent interest.

\subsection{Related Works}

\paragraph{Sequential Algorithms}
There is another line of designing fast algorithms for the MCM problem with the aid of matrix formulations.
Initiated by Mulmuley, Vazirani, and Vazirani~\cite{mulmuley1987matching}, several algebraic (randomized) algorithms~\cite{rabin1989maximum, mucha2004maximum, harvey2009algebraic} for computing the MCM size and finding 
an MCM in a general graph have been developed. The best known running time bound asymptotically coincides with the matrix multiplication time $O(n^\omega)$, where $\omega$ is known to be at most $2.371339$ \cite{alman2025more}.
Thus, for relatively dense graphs, randomization currently achieves an essential speeding-up for finding an MCM.
This approach is also extended to more general problems such as the so-called exact matching problem and the linear matroid parity problem~\cite{camerini1992random, cheung2014algebraic, sato2025exact}.

For the MCM problem for bipartite graphs, the Hopcroft--Karp algorithm~\cite{HK73} has been one of the fastest combinatorial (deterministic) algorithms for a long while (almost 50 years).
Very recently, on the top of the so-called Laplacian paradigm, $m^{1+o(1)}$-time flow algorithms have been developed~\cite{chen2022maximum, van2023deterministic}.
This immediately leads to an algorithm for finding an MCM in a bipartite graph with the same running time bound, which beats the Hopcroft--Karp bound in general.
More recently, Chuzhoy and Kahnna~\cite{chuzhoy2024faster, chuzhoy2024maximum} also proposed more direct, combinatorial algorithms, which run in $\tilde{O}(m^{1/3} n^{5/3})$ time (deterministic) and $n^{2 + o(1)}$ time (randomized); for relatively dense graphs, these algorithms outperform the Hopcroft--Karp algorithm.
These breakthroughs raise a natural question: for computation of MCMs in general graphs, is there any analogous breakthrough beating the $O(m\sqrt{n})$ bound?

\paragraph{CONGEST Algorithms}
There have been many known results for the MCM computation in distributed systems~\cite{II86,ABI86,FTR06,KMW16,LPP08,BCDELP19,Harris19,CS22,HS23,GG23,AKO18,AK20,BKS18,KI22,IKY24,FMJ22,
MMSS25}. Table~\ref{tab:MM-congest} summarizes the known upper-bound results specific to 
$(1 - \epsilon)$-approximation in the CONGEST model. In that context, the dependence
of running time on the inverse of $\epsilon$ has been the main focus. In early results, polynomial dependence
on $\epsilon^{-1}$ had been achieved only for bipartite graphs, while for general graphs only algorithms
with exponential dependence on $\epsilon^{-1}$ had been known. The existence of an algorithm for general graphs
with running time polynomially dependent on $\epsilon^{-1}$ remained a major open problem, but recently
this question has been positively solved by Fischer, Slobodan, and Uitto~\cite{FMJ22}. More recently, 
its complexity (i.e., the exponent of $\mathrm{poly}(1/\epsilon)$ factor) is improved in\cite{MMSS25}.
Our result can be placed on this line, with further improvement. To the best of our knowledge, all previous
results are based on Edmonds' blossom argument, whereas the approach of our algorithm is essentially 
different from those.

As for results not mentioned in Table~\ref{tab:MM-congest}, we mention the results on lower bounds and the algorithms
with approximation ratio other than $1 - \epsilon$. A lower bound of 
$\Omega(\sqrt{n} + D)$ rounds for $\epsilon = O(1 / \sqrt{n})$ and $D = O(\log n)$ (where $D$ is the diameter 
of the input graph) is presented in \cite{BKS18}. It also deduces the lower bound of $\Omega(1/\epsilon)$ rounds 
with respect to the dependence on $\epsilon$. Algorithms of attaining weaker guarantee on approximation factor (e.g., 
$(1/2 - \epsilon)$ or $(2/3 - \epsilon)$) have been proposed in several literatures~\cite{AKO18,LPP08,WW05}.
On the computation of exact solutions, an algorithm with $O(\mu(G))$ rounds for bipartite graphs is presented in~\cite{AKO18}, where $\mu(G)$ denotes the maximum cardinality of a matching of $G$.
In general graphs, the algorithm with the same bound is presented in~\cite{IKY24}. The algorithms in~\cite{IKY24} and its
former version~\cite{KI22} are based on the concept of alternating base trees, which we employ in this paper.

\paragraph{Semi-Streaming Algorithms}

Known upper-bound results for $(1 - \epsilon)$-approximate matching in the semi-streaming model is 
summarized in Table~\ref{tab:MM-stream}. As mentioned in the introduction, our algorithm belongs to
the class of algorithms with pass complexity independent of $n$. The situation in this class is very
close to the history in CONGEST. Up to~\cite{FMJ22}, polynomial dependence on $\epsilon^{-1}$ has been
known for bipartite graphs, and after that some improvement and generalization are proposed~\cite{MS25,HS23}.
Yet another line, admitting dependence on $\log n$ in return for moderate exponent of 
$\mathrm{poly}(\epsilon^{-1})$, is also considered in several papers~\cite{AS11,AG13,Assadi25}.
Lower bounds for the number of required passes are consider in several literatures~\cite{CKPSRSY21,Assadi21,AS23}.
The currently best bound for $(1 - \epsilon)$-approximate matching is $\Omega(\log (1 / \epsilon))$, which holds
for bipartite graphs.

\begin{table*}[t!]
\caption{Known upper bounds for the $(1 - \epsilon)$-approximate maximum cardinality/weighted matching problem 
in the CONGEST model. The parameter $\Delta$ and $W$ respectively mean the maximum degree and the maximum edge weight of the input graph $G$.}
\center
\label{tab:MM-congest}
\begin{tabular}{ c l c c }
\hline
Citation & \#rounds & Input  & Det./Rand. \\
\hline
\cite{LPP08} & $O(\epsilon^{-3}\log n)$ & bipartite & rand. \\
\cite{BKS18}& $O(2^{O(1/\epsilon)}\log \Delta  / \log \log \Delta)$ & bipartite & rand. \\

\cite{AKO18}& $O\left(\epsilon^{-1}(\log^2 \Delta+\log^{*}n)\right)$ & bipartite & rand. \\
\cite{AKO18}& $O\left(\epsilon^{-2}(\log W\Delta)  + \epsilon^{-1}(\log^2 \Delta+\log^{*}n) \right)$ 
& bipartite/weighted & rand.\\ \hline
\cite{LPP08} & $O(2^{O(1/\epsilon)}\epsilon^{-4} \log (1/\epsilon) \log n)$ & general & rand. \\
\cite{FFK21} & $O(2^{O(1/\epsilon)} \mathrm{polylog}(n))$ & general/weighted  & det. \\
\cite{FMJ22}& $O\left(\epsilon^{-63} \mathsf{MM}(n)\right)$ & general & det. \\
\cite{HS23} & $O\left(\epsilon^{-c} \mathrm{polylog}(n)\right)$ for $c >63$ & general/weighted & det. \\ 
\cite{MS25}     & $O\left(\epsilon^{-10} \log (\epsilon^{-1}) +  \epsilon^{-7} \log (\epsilon^{-1}) \mathsf{MM}(n))\right)$ & 
general & rand. \\
\textbf{This work}& \boldmath{$O\left(\epsilon^{-4} \mathsf{MM}(n)\right)$} & general & det. \\
\hline
\end{tabular}

\vspace{3mm}

\caption{Known upper bounds for the $(1 - \epsilon)$-approximate maximum cardinality/weight matching problem 
in the semi-streaming model. The notation $\tilde{O}_\epsilon$ means that it omits the factors dependent on 
$\epsilon$ due to lack of precise analyses in the original paper.}
\center
\label{tab:MM-stream}
\begin{tabular}{ c l l c c}
\hline
Citation & \#passes & Space & Input & Det./Rand. \\
\hline
\cite{EKS09} & $O(\epsilon^{-8})$ & $\tilde{O}(n)$ & bipartite & det. \\
\cite{EKMS12} & $O(\epsilon^{-8})$ & $\tilde{O}(n)$ & bipartite & det. \\
\cite{AG13} &  $O(\epsilon^{-2} \log (1/\epsilon))$ & $\tilde{O}(\epsilon^{-O(1)}n)$ & bipartite/weighted & det. \\
\cite{ALT21} & $O(\epsilon^{-2})$ & $\tilde{O}(n)$ & bipartite & det. \\
\hline
\cite{McGregor05} & $\epsilon^{-O(\epsilon^{-1})}$ & $\tilde{O}_{\epsilon}(n)$ & general & rand.  \\
\cite{AS11} & $O(\epsilon^{-1} \log n)$ & $\tilde{O}(\epsilon^{-O(1)}n)$ & general & rand. \\
\cite{AG13} &  $O(\epsilon^{-4} \log n)$ & $\tilde{O}(\epsilon^{-O(1)}n)$ & general/weighted & det. \\
\cite{Tirodkar18} & $\epsilon^{-O(\epsilon^{-1})}$ & $\tilde{O}_{\epsilon}(n)$ & general & det. \\
\cite{FMJ22} & $O(\epsilon^{-19})$ & $\tilde{O}(\epsilon^{-O(1)}n)$ & general & det. \\
\cite{HS23} & $O(\epsilon^{-39})$ & $\tilde{O}(\epsilon^{-O(1)}n)$ & general/weighted & det. \\
\cite{Assadi25} & $O(\epsilon^{-1} \log n)$ & $\tilde{O}(\epsilon^{-O(1)}n)$ & general/weighted & rand. \\
\cite{MMSS25}& $O(\epsilon^{-6})$ & $\tilde{O}(\epsilon^{-O(1)}n)$ & general & det. \\
\textbf{This work}& \boldmath{$O(\epsilon^{-4})$} & \boldmath{$\tilde{O}(n)$} & general & det. \\
\hline
\end{tabular}
\end{table*}

\section{Technical Outline}

This section provides the high-level ideas of our technical contribution. 

\subsection{Notations and Terminologies}

Throughout this paper, we only consider simple undirected graphs as the inputs.
An edge connecting two vertices $u$ and $v$ is denoted by $\{u, v\}$. We denote the vertex set and edge set of a given undirected graph $G$ by $V(G)$ and $E(G)$, respectively. 
The terminology ``path'' is used for referring to a simple path. While a path is formally defined as an alternating sequence of vertices and edges as usual, it is often regarded as a subgraph of $G$. For a path $P = v_0, e_1, v_1, e_2, \dots, e_{\ell}, v_\ell$ of $G$ and an edge 
$e = \{v_\ell, u\} \in E(G)$ satisfying $u \not\in V(P)$, we denote by 
$P \circ e$ or $P \circ u$ the path obtained by adding $e$ and $u$ to the tail of $P$. Similarly, for another path $P'$ starting at $v_\ell$ with $V(P) \cap V(P') = \{v_\ell\}$, we denote by $P \circ P'$ the path obtained by concatenating $P'$ to the tail of $P$ (without duplication of $v_\ell$). The inversion of the path $P$ 
(i.e., the path $v_\ell, e_\ell, v_{\ell - 1}, e_{\ell - 1}, \dots, e_1, v_0$) is denoted by $\overline{P}$. Given a path $P$ containing two vertices $u$ and $v$ such 
that $u$ precedes $v$, we denote by $P[u, v]$ the subpath of $P$ starting from $u$ and terminating with $v$.

For a graph $G$, a \emph{matching} $M \subseteq E(G)$ is 
a set of edges that do not share endpoints. We denote by $\Mmax(G)$ the maximum size of a matching of a graph $G$. A vertex $v$ is called \emph{free} if $v$ has no incident edge in $M$.
Given a pair $\Ical = (G, M)$, which we refer to as a \emph{matching system}, an \emph{alternating path} of $\Ical$ is a path $P = v_0, e_1, v_1, e_2, \dots, e_\ell, v_\ell$ such that exactly one of $e_{i-1}$ and $e_i$ is in $M$ for every $1 < i \le \ell$.
An \emph{augmenting path} of $\Ical$ is an alternating path connecting two different free
vertices.
Throughout this paper, we introduce a super free vertex $f$ 
that is connected with each free vertex by a length-two alternating path (where the edge incident to $f$ is a non-matching edge).
While the graph after adding $f$ only contains a single free vertex $f$, we use terminology 
``free  vertex'' for referring to free vertices in the original system. The set of all free 
vertices in this sense is denoted by $U_{\Ical}$. Since $|U_{\Ical}| \le 1$ implies that $M$ is a maximum matching, we usually assume $|U_{\Ical}| \geq 2$. 
Similarly, an ``augmenting path" is also defined as an alternating path connecting two vertices in $U_{\Ical}$ starting and ending with non-matching edges. 
We denote the length of a shortest augmenting path in $\Ical$ by  $2 \ell_{\Ical} + 1$; thus, $\ell_{\Ical}$ is the number of matching edges contained in a shortest augmenting path.

An alternating path of odd (resp.\ even) length is called an \emph{odd-alternating path} 
(resp.\ \emph{even-alternating path}).
We denote the length of the shortest odd (resp.\ even)-alternating path from $f$ to a vertex $u$ in $\Ical$ by $\dist^{\Odd}_{\Ical}(u)$ (resp.\ $\dist^{\Even}_{\Ical}(u)$). If $u$ does not admit any odd (resp.\ even)-alternating path from $f$ in $\Ical$,
we define $\dist^{\Odd}_{\Ical}(u) = \infty$ (resp.\ $\dist^{\Even}_{\Ical}(u) = \infty$). For $\theta \in \{\Odd, \Even\}$, $\OL{\theta}$ represents 
the parity different from $\theta$. The \emph{orthodox (resp.\ unorthodox) parity} of $u$ in $\Ical$ is $\theta \in \{\Odd, \Even\}$ satisfying $\dist^{\theta}_{\Ical}(u) < \dist^{\OL{\theta}}_{\Ical}(u)$ (resp.\ $\dist^{\theta}_{\Ical}(u) > \dist^{\OL{\theta}}_{\Ical}(u)$), 
denoted by $\Op_\Ical(u)$ (resp.\ $\Nop_\Ical(u) = \OL{\Op_\Ical(u)}$). Alternating paths from $f$ to $u$ of parity $\Op(u)$ (resp.\ $\Nop(u)$) are referred to as \emph{orthodox alternating paths} (resp.\ \emph{unorthodox alternating paths}).
We also mean $\dist^{\Op_\Ical(u)}_\Ical(u)$ and $\dist^{\Nop_\Ical(u)}_\Ical(u)$ by $\dist^\Op_\Ical(u)$ and $\dist^\Nop_\Ical(u)$, respectively. For each edge $e = \{u, v\} \in E$, we define $\rho_\Ical(e) = \Odd$ if $e \in M$ and $\rho_\Ical(e) = \Even$ 
otherwise. Intuitively, $\rho_{\Ical}(e)$ represents the parity of the alternating path $P$ from $f$ such that $P \circ e$ can become an alternating path. 

For any vertex $u$, we define $\paredgeset_{\Ical}(u)$
as the set of the edges $e$ incident to $u$ such that there exists a shortest orthodox alternating path from $f$ to $u$
terminating with edge $e$. We also define $\parset_{\Ical}(u) = \{u' \mid \{u', u\} \in \paredgeset_{\Ical}(u)\}$, and 
$\paredgeset_{\Ical} = \bigcup_{u \in V(G)} \paredgeset_{\Ical}(u)$. 
An edge $e = \{y, z\}$ not contained in $\paredgeset_\Ical$ is called \emph{critical}.

For the above notations, we often omit the subscript $\Ical$ if it is obvious in the context.

\subsection{Refinement of the Izumi--Kitamura--Yamaguchi Theorem}
Our new theorem is bulit on the top of the structure theorem by Izumi, Kitamura, and 
Yamaguchi~\cite{IKY24} (referred to as the IKY theorem) based on the framework of \emph{alternating base 
trees}~\cite{KI22}. To explain whole outline, we first introduce the concept and implication of the IKY theorem, 
as well as a few necessary notions and terminologies. An alternating base tree $T$ of 
a matching system $\Ical$ is defined as follows:

\begin{definition}[Alternating Base Tree]
An \emph{alternating base tree (ABT)} of $\Ical = (G, M)$ is a spanning tree $T$ of $G$
rooted by $f$ that is obtained by choosing any vertex in $\parset(u)$ as the parent of each vertex $u \in 
V(G) \setminus \{f\}$, denoted by $\parent_T(u)$. Note that this construction necessarily deduces a tree (cf.\ Lemma~\ref{lma:fundamental} (P1)).
For each vertex $v \in V(G)$, the subtree of $T$ rooted by $v$ is denoted by $T(v)$.
\end{definition}

Note that ABTs do \emph{not} guarantee that a tree path from $f$ to another vertex is an alternating path (see the vertex 
$v$ in Fig.~\ref{fig:example_ABT}). 

A trivial but important observation holding for any alternating base tree $T$ is that if there exists an 
unorthodox alternating path $P$ from $f$ to $t$, it must enter $T(t)$ using an edge not contained in $T$ 
(referred to as an \emph{incoming edge}\footnote{It is originally called an outgoing edge in \cite{IKY24}.} of $T(t)$). Considering the situation of finding such $P$, a happy case is
that $P$ crosses between the outside and inside of $T(t)$ exactly once. Then one can decompose the task
of finding $P$ into two independent sub-tasks of finding the prefix of $P$ at the outside of $T(t)$ 
and suffix of $P$ at the inside of $T(t)$. The technical idea of \cite{IKY24} is to provide a systematic 
way of identifying the incoming edge $e$ of $T(t)$ such that some shortest unorthodox 
alternating path $P$ from $f$ to $t$ containing $e$ satisfies this happy case. More precisely, 
it introduces the notion of \emph{edge level}, defined as follows: 

\begin{figure*}[tb]
    \centering
    \includegraphics[width=0.4\columnwidth]{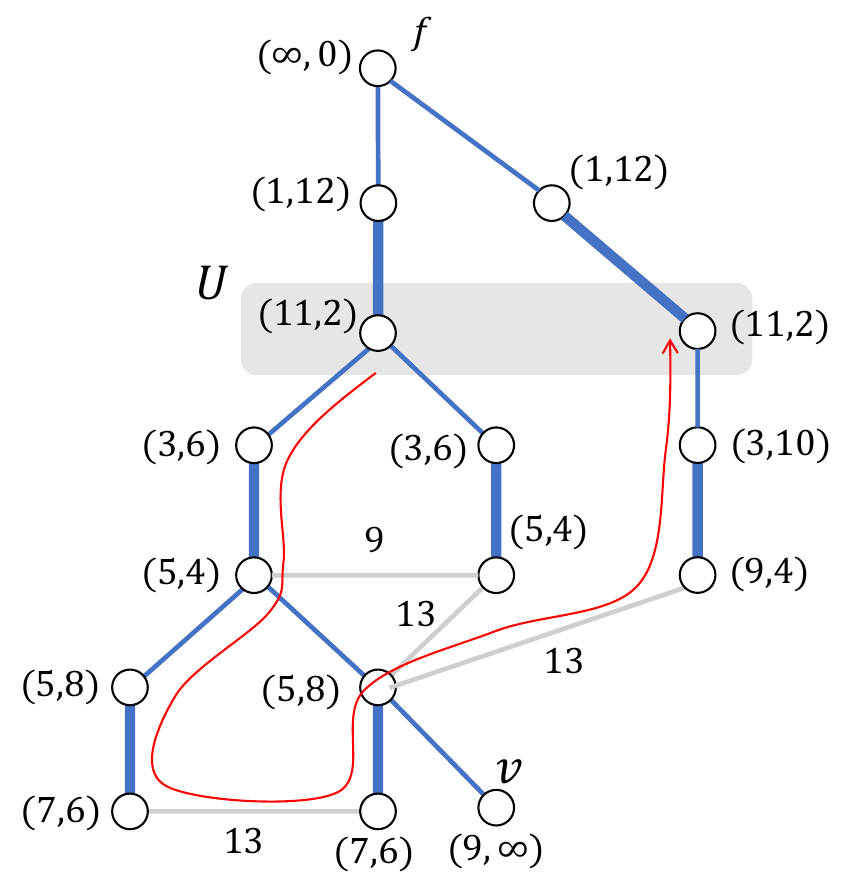}
    \caption{An example of ABT and edge volume. The tree edges are colored with blue. The pair of values assigned to each vertex means
    the values of $\dist^{\Odd}(\cdot)$ and $\dist^{\Even}(\cdot)$. The value assigned to each non-tree edge is its volume. The path indiceted by a red arrow is a shortest augmenting path.}
    \label{fig:example_ABT}
\end{figure*}

\begin{definition}[Edge volume, height, and level]
\sloppy{
For any edge $e = \{y, z\}$ of $G$, we define the \emph{volume} and \emph{height} of $e$ as 
$\vlevel(e)=\dist^{\rho(e)}(y)+\dist^{\rho(e)}(z) + 1$ and 
$\hlevel(e)=\max\{\dist^{\rho(e)}(y),\, \dist^{\rho(e)}(z)\}$, respectively.
We then define the \emph{level} of $e$ as $\level(e) = n' \cdot \vlevel(e) + \hlevel(e)$, where $n'$ is an
arbitrary predefined value larger than $|V(G)|$.\footnote{This is slightly modified from the definition in \cite{IKY24}, but 
essentially the same. The intention is to introduce the lexicographic ordering over the edges 
that $e_1$ is smaller than $e_2$ if and only if either $\vlevel(e_1)<\vlevel(e_2)$ or $(\vlevel(e_1)=\vlevel(e_2)) \wedge 
(\hlevel(e_1)<\hlevel(e_2))$ holds.}
}
\end{definition}
The IKY theorem shows that picking up any \emph{minimum incoming edge 
(MIE)} of $T(t)$, which is a minimum-level one among the incoming edges of $T(t)$, always yields the 
happy case as stated above. Given an ABT $T$ and an alternating path $P$, we define the level (resp.\ volume) 
of $P$ with respect to $T$ as the maximum level (resp.\ volume) of all non-tree edges in $P$, which is denoted 
by $\level_T(P)$ (resp.\ $\vlevel_T(P)$). If $P$ does not contain any non-tree edge, both are defined to be zero. 
The formal statement of the IKY theorem is below:

\begin{theorem}[Izumi, Kitamura, and Yamaguchi\cite{IKY24}, rephrased]
\label{thm:IKY24}
Let $T$ be any ABT, $t$ be any vertex in $T$, and $\Minedge = \{y, z\}$ be 
any minimum incoming edge of $T(t)$, where $z$ is the vertex contained in $V(T(t))$. 
For any $\theta \in \{\Odd, \Even\}$ such that 
$\dist^{\theta}(t) < \infty$ holds, there exists a shortest $\theta$-alternating path $P$ from $f$ to $t$ satisfying the following conditions:
\begin{itemize}
    \item If $\theta = \Op(t)$, $\vlevel_T(P) < \vlevel(\Minedge)$.
    \item Otherwise, $P$ has the form of $P[f, y] \circ \{y, z\} \circ P[z, t]$, and satisfies
    $\level_T(P[f, y]) < \level(\Minedge)$ and $\level_T(P[z, t]) < \level(\Minedge)$; 
    in other words, $\level_T(P) = \level(\Minedge)$ and $\Minedge$ is the unique edge attaining the maximum.
\end{itemize}
\end{theorem}
Note that in the case of $\theta = \Nop(t)$, $P[f, y]$ and $P[z, t]$ are 
respectively contained in the outside/inside of $T(t)$ because their levels are less than $\level(\Minedge)$.
While the original proof of the theorem requires a long and complicated argument, we present a much 
simpler and concise proof of it, which is also a part of our contribution.

Since the proof of Theorem~\ref{thm:IKY24} is constructive, it also deduces a simple recursive 
algorithm of computing a shortest orthodox/unorthodox path from $f$ to a given vertex $t$ (see Alg.~\ref{alg:pathConstruction}), 
but it requires the prior information on $\dist^{\Op}(u)$, $\dist^{\Nop}(u)$, and an MIE of $T(u)$ for all 
$u \in V(G)$ (note that $\paredgeset(u)$ for all $u$ (and thus an ABT $T$ itself) is easily computed from these 
information). To address this matter, we present a simple algorithm of computing those information based on 
our structure theorem, which runs in $O(m \log n)$ time (cf.\ Sections~\ref{subsec:compDist} and \ref{subsec:compDistDetails}).

\subsection{Extension to a Maximal Set of Disjoint Shortest Augmenting Paths}
\label{subsec:maximalAugPath}

Let $T$ be any ABT, and recall that $\ell_\Ical$ denotes the number of matching edges contained in a shortest augmenting path of $\Ical = (G, M)$.
We denote by $F_T$ the forest obtained by removing $f$ and its neighbors (equivalently, 
the set of subtrees $T(u)$ for all $u \in U_\Ical$) and by $\freeroot_T(v)$ the root of the tree in $F_T$ 
containing $v \in V(G)$.
Considering simultaneous construction of multiple shortest augmenting paths following the IKY theorem, it
has a shortcoming that the prefix of the path $P[f, y]$ is ``not localized'', which means that $P[f, y]$
may cross many components in $F_T$. This can be an obstacle to establish a collection of disjoint augmenting
paths. The key observation toward Theorem~\ref{thm:structuralTheorem} is that 
such a case never occurs if $\Minedge$ taken in Theorem~\ref{thm:IKY24} is critical (i.e., 
$\Minedge \not\in \paredgeset_\Ical$). In such a case, the augmenting path found by the approach of 
Theorem~\ref{thm:IKY24} necessarily lies in the two components in $F_T$ bridged by $\Minedge$.
Even better, if $\Minedge$ is critical, it is not necessary to take a minimum-level one, but suffices
to take a minimum-volume one, which is equal to $2\ell_\Ical + 5$ by a simple argument 
(cf.~Lemma~\ref{lma:wastingEdge}).

This observation straightforwardly yields the double-path argument 
in Theorem~\ref{thm:structuralTheorem}. We take all the critical edges of volume $2\ell_\Ical + 5$ which 
can cross two subtrees in $F_T$ for some $T$ as the set $\Ccset$, and consider the
construction of a double path $(P, Q)$ for $e = \{y, z\} \in \Ccset$ in the ABD $H$ (for the formal definitions of double paths and ABDs, see Section~\ref{subsec:DPtoSAP}).
Since all the edges in the double path are contained in $\paredgeset_\Ical$, one can construct an ABT $T$ \emph{respecting} 
$(P, Q)$, i.e., an ABT containing $P$ and $Q$ as its paths. Then $e$ is a non-tree edge of $T$ 
bridging two components $T(\freeroot_T(y))$ and $T(\freeroot_T(z))$ in $F_T$ ($\freeroot_T(y) \neq \freeroot_T(z)$).
The application of the augmenting-path construction algorithm by Theorem~\ref{thm:IKY24}
yields the transformation of a double path $(P, Q)$ into a shortest augmenting path from 
$\freeroot_T(y)$ to $\freeroot_T(z)$ lying 
within $T(\freeroot_T(y))$ and $T(\freeroot_T(z))$ (cf.\ Algorithm~\ref{alg:AugPathConstruction}).
Importantly, this construction can be applied simultaneously to 
the set $\Dcal$ of vertex-disjoint double paths: One can construct an ABT $T$ respecting 
all the double paths in $\Dcal$ (which is obviously possible because $\Dcal$ is vertex-disjoint),
and transform $\Dcal$ into a set of shortest augmenting paths simultaneously. They all lie at 
the insides of distinct subtrees, and thus are disjoint. Theorem~\ref{thm:structuralTheorem} 
is proved by showing that this transformation preserves the maximality.

\subsection{An Overview of a New MCM Algorithm}

\subsubsection{Whole Structure}

We present a new MCM algorithm based on our structure theorem, which runs in 
$O(m\sqrt{n} \log n)$ time. As the MV algorithm, it iterates the construction of a maximal 
set of disjoint shortest augmenting paths. The seminal analysis by Hopcroft and Karp~\cite{HK73} guarantees 
that $O(\sqrt{n})$ iterations suffices to reach an MCM.
Hence our aim is constructing a maximal set of disjoint shortest augmenting paths in $O(m \log n)$ time.
By Theorem~\ref{thm:structuralTheorem}, it suffices to construct a maximal set of disjoint double paths. 
\begin{enumerate}
    \item Compute $\dist^{\theta}(u)$ for all $u \in V(G)$ and $\theta \in \{\Even, \Odd\}$, and identify
    the edge set $\paredgeset_{\Ical}$ (and $H$), which includes the identifiaction of value $2\ell + 5$ (by
    Corollary~\ref{corol:IKY24}). 
    \item Following the information on $H$, compute a maximal set 
    $\Dcal$ of disjoint double paths. If $\Dcal$ is empty, the algorithm terminates. 
    \item Construct an ABT $T$ respecting all double paths in $\Dcal$, and find any MIE of $T(u)$ for all 
    $u \in V(H)$. 
    \item Augment the current matching by the set $\Augt(\Dcal)$ of disjoint shortest augmenting paths obtained from $\Dcal$.
\end{enumerate}
We have already explained that the last step is implemented with $O(m)$ time. Hence we see
the technical outline of other steps.

\subsubsection{Step 1: Computing Orthodox/Unorthodox Distance}
\label{subsec:compDist}

The hardest part in computing distance information is the computation of $\dist^{\Nop}(u)$. Our algorithm 
computes it through an observation derived from Theorem~\ref{thm:IKY24}, which states that for any $v \in V(G)$ 
admitting unorthodox alternating path from $f$ to $v$, any ABT $T$, and an MIE $\Minedge$ of $T(v)$, $\dist^{\Nop}(v) = \vlevel(\Minedge) - \dist^{\Op}(v)$ holds (cf.\ Corollary \ref{corol:IKY24}). That is, the identification of an 
MIE of $T(v)$ implies the computation of $\dist^{\Nop}(v)$. Our algorithm grows up an ABT $T$ in 
the round-by-round manner. At round $r$, it identifies all the pairs $(v, \theta)$ satisfying $\dist^{\theta}(v) = r$ 
and all the edges $e$ of height $r$. The main technical challenge lies at the point that we need to conduct it concurrently with the growth of $T$, because the value of $\dist^{\Nop}(v)$ is also necessary to identify the distances of its children for computing $\vlevel(\Minedge)$. A straightforward idea of implementing this approach is 
as follows: when the volume of a non-tree edge $e$ is determined, the algorithm checks all the vertices $t$
such that $e$ can be an MIE of $T(t)$
by upward traversal of $T$ from the two endpoints of 
$e$ to $f$. However, this idea is not efficient because the upward search might check a vertex $t'$ 
which has already decided $\dist^{\Nop}(t')$ again. To avoid such a wasting search, our algorithm adopts a 
simple strategy of contracting an edge $\{t, \parent_T(t)\}$, i.e., shrinking the two vertices into a single vertex, when $\dist^{\Nop}(t)$ is decided. Not only 
for saving search cost, this strategy also provides a simpler mechanism for identifying MIEs: 
By repeating tree edge contraction, any MIE $\Minedge$ of $T(t)$ eventually becomes an edge incident to 
$t$ earlier than round $\dist^{\Nop}(t)$ (cf.\ Lemma~\ref{lma:distComp}). Hence $\dist^{\Nop}(t)$ is 
correctly decided only by checking the minimum-level edge incident to $t$ when the set of incident edges and 
their volumes are modified (where self-loops incident to $t$ are excluded, because they are the edges 
incident to a vertex in $T(t)$
but not incoming edges of $T(t)$).
This contraction strategy is implemented by two data structures:
The contraction operation naturally induces a partition of $V(G)$, where each subset is a vertex set shrunk into the same vertex.
It is managed by the standard union-find data structure.
The set of incident non-tree edges for each (shrunk) vertex is managed by any meldable heap supporting \textsc{find-min}, \textsc{merge}, \textsc{add}, and \textsc{delete-min} with $O(\log n)$-amortized time. 

\paragraph{Note on CONGEST/semi-streaming Implementation}

Our distance computation algorithm is naturally transformed into the ones for CONGEST
and semi-streaming settings. Basically, the algorithm performs the two tasks of 
expanding $T$ and aggregation of the volume information of MIEs along $T$, which are
conducted in the round-by-round manner. Hence the aforementioned models can easily emulate one round of the sequential implementation with a bundle of rounds/passes (cf.\ Section~\ref{subsec:subroutine}). It provides CONGEST/semi-streaming implementations 
of bounded-length distance computation, which outputs all orthodox/unorthodox distances up to $2\ell + 5$. The running time of the CONGEST implementation is $O(\ell^2)$ rounds, and the number of passes taken by the semi-streaming 
implementation is $O(\ell)$. Note that in the CONGEST model we need an extra $O(\ell)$
factor due to the simulation of shrunk vertices, each of which correspond to 
a subgraph of diameter $O(\ell)$. 

\subsubsection{Step 2: Computing a Maximal Set of Disjoint Double Paths}
\label{subsec:DDFS}
This part is almost the same as construction of a maximal set of disjoint shortest augmenting paths by the \emph{double DFS (DDFS)} in \cite{MV1980,vazirani2024theory}.
The only difference is that they start it in DAGs involving nested blossoms already shrunk during the computation of $\dist^\theta(v)$ (which are called the oddlevel/evenlevel of $v$ in \cite{MV1980,vazirani2024theory}), but we first do it simply in ABDs (and then, in the later step, transform the obtained double paths into shortest augmenting paths in the original graph).
Our algorithm iteratively grows a set of disjoint double paths until no further disjoint double path is found.
More precisely, for each critical edge $e = \{y, z\}$ of volume $2\ell + 5$ (which is not necessarily crossable), we run the DDFS %
on $H$ with starting vertices $y$ and $z$. 
Here we explain its properties our algorithm utilizes, without stating the implementation details: 
It outputs a double path $D = (P, Q, y, z)$ or a \emph{bottleneck} $v$, and an \emph{omissible set} $W \subseteq V(H)$, which will be removed if $D$ is found and shrunk into $v$ otherwise.
Intuitively, $W$ is a set of vertices $w$ such that any $w$--$U$ path necessarily intersects some vertex already removed, $V(D)$, or $v$.
Thus, if $D$ is found, $V(D) \subseteq W$ and we just remove $W$.
Otherwise, $W$ contains the vertices reachable from $y$ or $z$ except for $v$ and we shrink $W$ into $v$.
For the resulting graph $H$ after deletion/shrinking, the next try of DDFS is performed for another critical edge $e'$ of volume $2\ell + 5$ (if at least one of the endpoints has been removed or both of them have been shrunk into the same vertex, it is simply skipped). 
The running time of the DDFS algorithm is $O(|E_W|)$, where $E_W$ is the set of edges incident to a vertex in $W$.
Since the vertices in $W$ do not join the subsequent runs of DDFS, whole running time (up to the construction of a maximal set of disjoint double paths) is $O(m)$.

\subsubsection{Step 3: Identification of MIEs}

The strategy of Step 3 is almost the same as Step 1: Let $T$ be any given ABT, and each vertex $v \in V(T)$ manages the
set of incident non-tree edges by an efficient meldable heap. The contraction of each tree edge from the leaf side 
obviously admits the identification of an MIE for each subtree $T(v)$ ($v \in V(T)$). That is, the minimum element 
in the heap at any leaf vertex $v$ is an MIE of $T(v)$. This step is implemented within $O(m \log n)$ time.

\subsection{Overview of Approximate MCM}

This section sketches how our framework is utilized for leading an efficient algorithm for $(1 - \epsilon)$-approximation of a maximum matching in distributed and semi-streaming settings.
For simplicity, we assume that $\epsilon$ is represented as $\epsilon = 2^{-x}$ for some integer $x > 0$.

Following the Hopcroft--Karp analysis, our $(1 - \epsilon)$-approximation algorithm aims 
to obtaining a matching not admitting an augmenting path of length shorter than $2\epsilon^{-1}$,
which is accomplished by iteratively computing a maximal set of shortest augmenting paths. However, 
as mentioned in the introduction, it is not so easy to take a maximal set of disjoint paths in 
distributed or streaming settings. While many known algorithms~\cite{McGregor05,FMJ22,MMSS25} in this context
try to find a substantially large number of short augmenting paths not limited to shortest ones for 
circumventing maximality requirement, our framework helps only for finding shortest augmenting paths, 
and thus one cannot adopt such an approach. Instead, our algorithm outputs 
a \emph{hitting set} $B \subseteq V(G)$, which a vertex subset any shortest augmenting path in $(G, M)$ intersects,
in addition to a set $\Qcal$ of disjoint shortest augmenting paths.
More precisely, one can construct an algorithm $\textsc{AugAndHit}_{\ell}(G, M)$ of 
outputting a set $\Qcal$ of shortest augmenting paths of length at most $2\ell + 1$ and 
a hitting set $B$ of size at most $16|\Qcal|(\ell+1)^2 + |M| / (4(\ell+1))$ (if the length of 
shortest augmenting paths is larger than $2\ell + 1$, it returns empty $\Qcal$). This 
algorithm is implemented in the CONGEST model with $O(\ell^2 \mbox{\rm\textsf{MM}}(n))$ 
rounds, where $\mbox{\rm\textsf{MM}}(n)$ is the running time of computing a maximal 
matching of graphs on $n$ vertices, and implemented
in the semi-streaming model with $O(n\log n)$-bit space and $O(\ell^2)$ paths.
The formal description is given in Lemma~\ref{lma:AugAndHit}, and its proof and implementation details are shown in Section~\ref{subsec:AugAndHit}.

Using $\textsc{AugAndHit}_{\ell}(G, M)$, we construct an algorithm $\textsc{Amplifier}_{\alpha}$, which takes a parameter $\alpha > 0$ (represented as a power of two) and any matching $M$ such that $|M| = (1 - \alpha')\Mmax(G)$ holds for some $\alpha' \leq \alpha$, and outputs a matching $M'$ of $G$ whose approximation factor is at least 
$\min\{(1 - \alpha' + \Theta(\alpha^2)), (1- \alpha/2)\}$ (Lemma~\ref{lma:amplifier}).
Iteratively applying $\textsc{Amplifier}_{\alpha}$ $O(\alpha^{-1})$ times, we obtain a $(1 - \alpha/2)$-approximate matching.
Starting with any maximal matching (which is a $1/2$-approximate matching) and calling that iteration process $O(\log \epsilon^{-1})$ times with $\alpha = 1/2, 1/4, 1/8, \dots 1/\epsilon^{-1}$, we finally obtain a $(1 - \epsilon)$-approximate matching.

We sketch the outline of $\textsc{Amplifier}_\alpha$ (see Section~\ref{subsec:Amplifier} for the details).
The execution of $\textsc{Amplifier}_\alpha$ consists of repetition of $K = 4\alpha^{-1}$ \emph{phases} with modification of the input graph.
Let $(G_i, M_i)$ be the matching system at the beginning of the $i$-th phase.
The behavior of the $i$-th phase consists of the following three steps:
\begin{enumerate}
    \item \emph{Path finding}: Call $\textsc{AugAndHit}_K(G_i, M_i)$ and find a set of shortest augmenting paths $\Qcal_i$ and a hitting set $B_i$. 
    \item \emph{Length stretch}: For each vertex $v \in B_i$, if $v$ is not free, the algorithm subdivides the matching edge incident to $v$ into a length-three alternating path such that the middle edge is a non-matching edge. Otherwise, we rename $v$ into $v'$, and add a length-two alternating path to $v'$ where the edge incident to $v'$ is a matching edge and the endpoint of the path other than $v'$ becomes the new $v$. The resultant graph and matching is denoted by $(G_{i+1}, M'_i)$.
    \item \emph{Augmentation}: Obtain a matching $M_{i+1}$ in $G_{i+1}$ by augmenting $M'_i$ along augmenting paths corresponding to $\Qcal_i$.
\end{enumerate}
After the $K$ phases, the algorithm finally executes the \emph{recovery} step.
In this step, the algorithm iteratively restores each subdivided (or extended) part, and finally outputs an augmented matching $\hat{M}_1$ in the original graph $G_1$.

Since almost all parts of $\textsc{Amplifier}_{\alpha}$ is to execute $\textsc{AugAndHit}_{\ell}$, its 
CONGEST/semi-streaming implementations are straightforwardly deduced from Lemma~\ref{lma:AugAndHit}.
By a simple complexity analysis, we obtain Theorem~\ref{thm:mainApproximate}.

\section{Proof of Theorem~\ref{thm:structuralTheorem}}
\label{sec:details}

\subsection{A Shorter Proof for the IKY Theorem}

The goal of this section is to provide a simpler proof of Theorem~\ref{thm:IKY24}. 
We first introduce further additional notations. We define $\Outset_T(t)$ as the set of all incoming edges of $T(t)$, 
and $\Minset_T(t)$ as the set of all MIEs 
of $T(t)$. When referring to two endpoints of $e \in \Outset_T(t)$ by symbols $y$ and $z$ (and their variants with 
upper/lower scripts), we always assume $z \in V(T(t))$. To refer to the level, volume, height of MIEs for $T(t)$, 
we choose an arbitrary edge in $\Minset_T(t)$ as a \emph{canonical MIE} of $T(t)$, which is denoted by 
$\Minedge_T(t)$.

In the following argument, $T$ is fixed as any ABT of $\Ical$. In addition, when considering some edge 
$e \in \Outset_T(t)$, we implicitly assume that $t$ is a vertex such that $\Outset_T(t) \neq \emptyset$ (and 
thus $t \neq f$). For the proof, we first present a few auxiliary lemmas.

\begin{lemma} \label{lma:fundamental}
For any $u \in V(T) \setminus \{f\}$, the following properties hold:
\sloppy{
\begin{itemize}
    \item[(P1)] For any $u' \in \parset(u)$, $\dist^{\Op}(u) > \dist^{\Op}(u')$ holds, which also implies that the construction in the definition of alternating base trees certainly provides a spanning tree. It also implies 
    $\dist^{\Op}(u) < \dist^{\Op}(u'')$ for any descendant $u''$ of $u$ in any ABT $T$.
    \item[(P2)] For any neighbor $u'$ of $u$, $\dist^{\OL{\Op(u)}}(u') + 1 = 
    \dist^{\Op}(u)$ holds if and only if $u' \in \parset(u)$ holds.
\end{itemize}
}
\end{lemma}

\begin{proof}
\noindent
\textbf{Proof of (P1)}: By the definition of $\paredgeset(u)$, there exists an alternating path $P$ from $f$ to $u$ of length $\dist^{\Op(u)}(u)$ whose last two vertices are $u'$ and $u$. Hence $P[f, u']$ is an alternating path from $f$ to $u'$ of length less than $\dist^{\Op(u)}(u)$. It implies $\dist^{\Op(u')}(u') < \dist^{\Op(u)}(u)$.

\textbf{Proof of (P2)}: We first prove the direction $\Rightarrow$. Let $P$ be a shortest $\OL{\Op(u)}$-alternating path from $f$ to $u'$. For any vertex $v \in V(P)$, $\dist^{\Op}(v) \leq |P| < \dist^{\Op}(u)$ holds.
Thus, $P$ does not contain $u$, and $P \circ \{u', u\}$ is an $\Op(u)$-alternating path from $f$ to $u$ of length $\dist^{\OL{\Op(u)}}(u') + 1 = \dist^{\Op}(u)$, i.e., it is a shortest orthodox alternating path and thus $u' \in \parset(u)$ holds.

Next, we prove the direction $\Leftarrow$. Let $P$ be any $\Op(u)$-alternating path from $f$ to $u$ terminating with edge $\{u', u\}$. 
Then the prefix $P[f, u']$ must be a shortest $\OL{\Op(u)}$-alternating path from $f$ to $u'$ as follows.
Suppose to the contrary that there exists a shorter $\OL{\Op(u)}$-alternating path $P'$ from $f$ to $u'$.
Then, $P'$ does not contain $u$ as observed in the proof of $\Rightarrow$, and hence one can obtain a shorter $\Op(u)$-alternating path from $f$ to $u$ than $P$ by replacing $P[f, u']$ with $P'$, a contradiction.
Thus, $\dist^{\OL{\Op(u)}}(u') + 1 = \dist^{\Op}(u)$.
\end{proof}

\begin{lemma} \label{lma:fundamental2}
For any edge $e = \{u, v\}$, the following properties hold:
\begin{itemize}
    \item[(P1)] $\dist^{\rho(e)}(u) + 1 \geq \dist^{\Op}(v)$ holds.
    \item[(P2)] If $e \not\in \paredgeset(v)$, $\dist^{\rho(e)}(u) + 1 > \dist^{\Op}(v)$ holds.
\end{itemize}

\end{lemma}

\begin{proof}
\textbf{Proof of (P1)}: Let $P$ be any shortest $\rho(e)$-alternating path from $f$ to $u$. 
Since $P \circ \{u, v\}$ itself or its subpath $P'$ up to the first appearance of $v$ 
is an alternating path to $v$, $\dist^{\rho(e)}(u) + 1 \geq \dist^{\Op}(v)$ holds.

\textbf{Proof of (P2)}: If the equality holds, $\Op(v) = \OL{\rho(e)}$ must be satisfied, which implies $e \in \paredgeset(v)$ by Lemma~\ref{lma:fundamental} (P2). Hence the equality does not hold in the case of $e \not\in \paredgeset(v)$.
\end{proof}

\begin{lemma} \label{lma:volumeBound}
For any $t \in V(T)$, $\dist^{\Op}(t) + \dist^{\Nop}(t) \geq \vlevel(\Minedge_T(t))$ holds.
\end{lemma}

\begin{proof} 
Let $P$ be any shortest unorthodox alternating path from $f$ to $t$. 
By the definition of $\paredgeset$, $P$ does not contain any edges in $\paredgeset(t)$, and thus $P$ must use an edge in $\Outset_T(t) \setminus \paredgeset(t)$ to enter $T(t)$.
Let $e = \{y, z\}$ be the last edge in $(\Outset_T(t) \setminus \paredgeset(t)) \cap E(P)$ so that the suffix $P[z, t]$ lies in $T(t)$. 
We also define $Q$ as any shortest (orthodox) alternating path from $f$ to $t$ (Fig.~\ref{fig:insideIsShorter}). By Lemma~\ref{lma:fundamental} (P1), $Q$ does not intersect 
any vertices in $V(T(t)) \setminus \{t\}$.
Since $P[z, t]$ lies in $T(t)$, the concatenation $Q \circ \OL{P}[t, z]$ is a $\rho(e)$-alternating path from $f$ to $z$ (note that this holds even 
if the length of $\OL{P}[t, z]$ is zero). In addition, $P[f, y]$ is a $\rho(e)$-alternating 
path from $f$ to $y$.
Thus,
\begin{align*}
\dist^{\Op}(t) + \dist^{\Nop}(t) 
&= |Q| + |P| \\
&= |Q \circ \OL{P}[t, z]| + |P[f, y]| + |\{y, z\}| \\
&\geq \dist^{\rho(e)}(z) + \dist^{\rho(e)}(y) + 1 \\
&= \vlevel(e) \\
&\geq \vlevel(\Minedge_T(t)).
\end{align*}
The lemma is proved. 
\end{proof}

\begin{lemma} \label{lma:mieNotParedge}
For any $t \in V(T)$ with $\dist^\Nop(t) < \infty$, there exists an MIE $\Minedge$ of $T(t)$ such that $\Minedge \not\in 
\paredgeset(t)$ holds.
\end{lemma}

\begin{proof}
Let $e = \{y, t\}$ be any edge in $\paredgeset(t)$. By the definition, $\rho(e) = \Nop(t)$ holds. 
Let $P$ be any shortest $\beta(t)$-alternating path from $f$ to $t$, and as with the previous proof, let $e' = \{y', z'\}$ be the last edge in $(\Outset_T(t) \setminus \paredgeset(t)) \cap E(P)$ so that the suffix 
$P[z', t]$ lies in $T(t)$.
Let $Q$ be a shortest $\alpha(t)$-alternating path from $f$ to $t$ terminating with $e$.
Since $Q[f, y]$ is a shortest $\Nop(t)$-alternating path from $f$ to $y$, we have $\vlevel(e) = |P| + |Q|$. 
In addition, $|P| > |Q|$ obviously holds, and thus $\hlevel(e) = |P|$. Since $Q \circ \OL{P}[t, z']$ is a 
$\rho(e')$-alternating path from $f$ to $z'$ and $P[f, y']$ is a $\rho(e')$-alternating path from $f$ to $y'$, 
we also have $\vlevel(e') \leq |P| + |Q| = \vlevel(e)$. By Lemma~\ref{lma:fundamental} (P1), we have 
$\dist^{\Op}(z') \leq \dist^{\Op}(t)$ (note that we might have the case of $z' = t$). 
It implies $|P[f, z']| \geq |Q|$, and thus $|Q \circ \OL{P}[t, z']| \leq |P|$. Since 
$|P[f, y']| \leq |P|$ obviously holds, it concludes $\hlevel(e') \le \max\{|P[f, y']|,\, |Q \circ \OL{P}[t, z']|\} \leq |P| = \hlevel(e)$.
Then, if $e$ is an MIE of $T(t)$, $e'$ is also an MIE of $T(t)$. The lemma is proved.
\end{proof}

\begin{lemma} \label{lma:oddOrthodoxProp}
Let $t \in V(T)$ be a vertex such that $\alpha(t) = \Odd$, $t'$ be a child of $t$ such that $\{t, t'\}$ is not a matching edge, 
and $e = \{y, z\}$ be any edge incident (not necessarily incoming) to a vertex in $V(T(t'))$ (Fig.~\ref{fig:oddOrthodoxProp}). 
Then $\vlevel(e) > \vlevel(\Minedge_T(t))$ (and thus $\level(e) > \level(\Minedge_T(t))$) holds.
\end{lemma}

\begin{proof}
Let $P$ be any shortest orthodox alternating path from $f$ to $t'$ terminating the edge $\{t, t'\}$. Since 
$P[f, t]$ is an $\Even$-alternating path (i.e., unorthodox alternating path) from 
$f$ to $t$, we have
\begin{align}
\vlevel(e) &= (\dist^{\rho(e)}(y) + 1) + \dist^{\rho(e)}(z) \nonumber \\
&\geq 2\dist^{\Op}(z) & \text{(Lemma~\ref{lma:fundamental2} (P1) and $\dist^{\rho(e)}(z) \geq \dist^{\Op}(z)$)} \nonumber \\
&\geq 2\dist^{\Op}(t') & \text{(Lemma~\ref{lma:fundamental} (P1))} \nonumber \\
&> 2 \dist^{\Nop}(t) & \text{($|P| > |P[f, t]|$)} \nonumber \\
&> \dist^{\Op}(t) + \dist^{\Nop}(t) \nonumber \\ 
&\geq \vlevel(\Minedge_T(t)). & \text{(Lemma~\ref{lma:volumeBound})} \nonumber %
\end{align}
The lemma is proved.
\end{proof}

\begin{figure*}[tb]
\begin{minipage}{0.35\columnwidth}
    \centering
    \includegraphics[width=0.8\columnwidth]{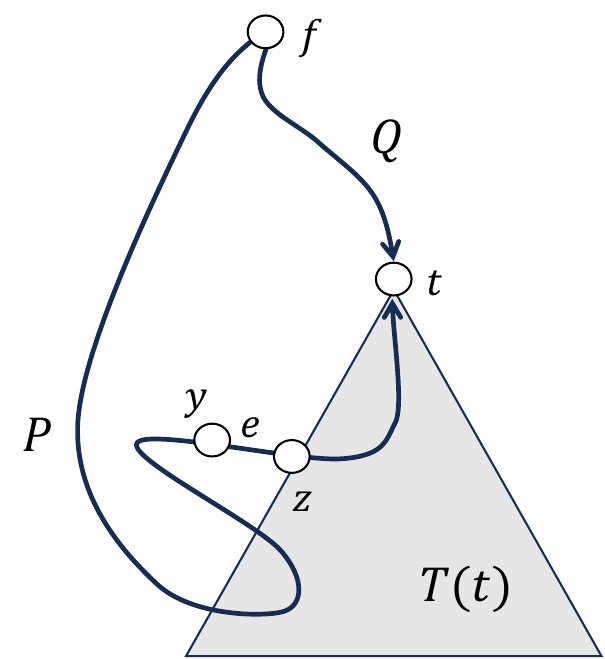}
    \caption{Proof of Lemma~\ref{lma:volumeBound}.}
    \label{fig:insideIsShorter}
\end{minipage}
\begin{minipage}{0.63\columnwidth}
    \centering
    \includegraphics[width=0.8\columnwidth]{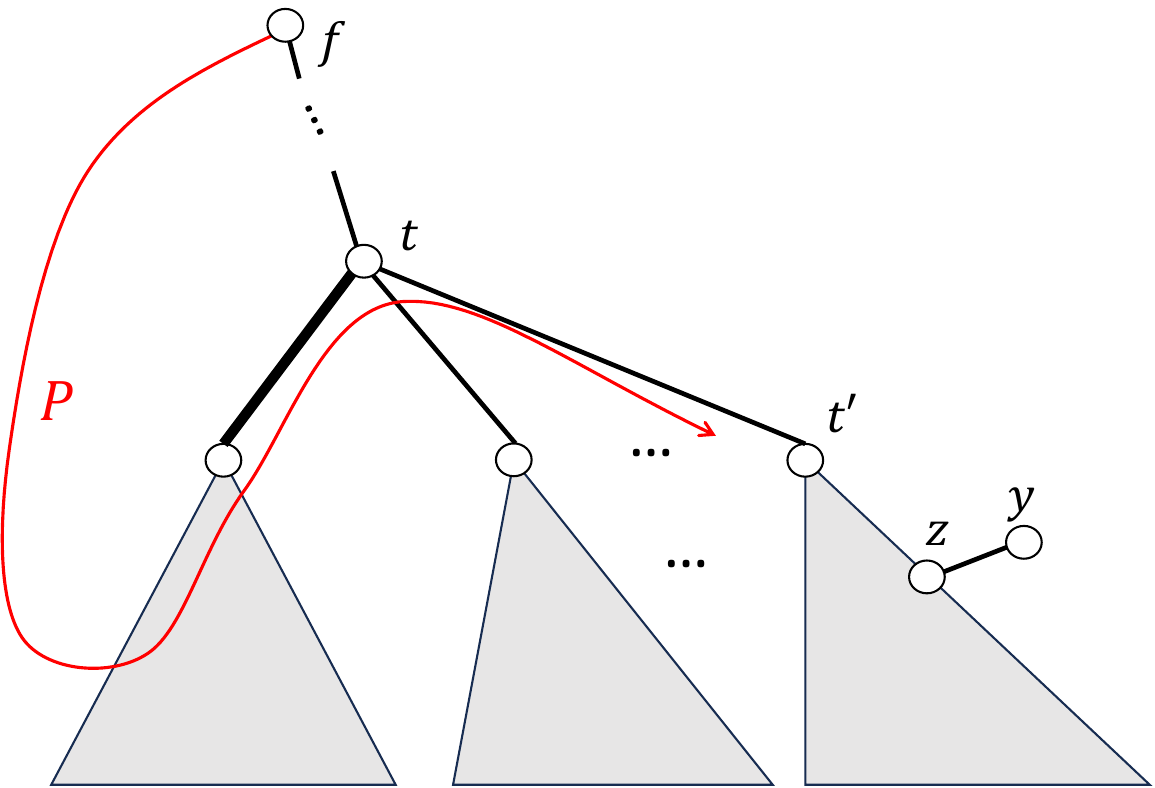}
    \caption{Proof of Lemma~\ref{lma:oddOrthodoxProp}}
 \label{fig:oddOrthodoxProp}
\end{minipage}
\end{figure*}

\begin{lemma} \label{lma:heightBound}
For any $t \in V(T)$ with $\dist^\beta(t) < \infty$, $\dist^{\Nop}(t) > \hlevel(\Minedge_T(t))$ holds. 
\end{lemma}

\begin{proof}
Let $\Minedge = \{y, z\}$ be an MIE of $T(t)$ not contained in $\paredgeset(t)$ (the existence of such an edge
is guaranteed by Lemma~\ref{lma:mieNotParedge}).
We then have 
\begin{align*}
\dist^{\Nop}(t) &\geq \dist^{\rho(\Minedge)}(y) + \dist^{\rho(\Minedge)}(z) + 1 - \dist^{\Op}(t) 
& \text{(Lemma~\ref{lma:volumeBound})} \\
&\geq \dist^{\rho(\Minedge)}(y) + \dist^{\Op}(z) + 1 - \dist^{\Op}(t) &  \\ 
&> \dist^{\rho(\Minedge)}(y). & \text{($z \in V(T(t))$ and Lemma~\ref{lma:fundamental} (P1))}
\end{align*}
We also have 
\begin{align*}
\dist^{\Nop}(t) &\geq \dist^{\rho(\Minedge)}(z) + \dist^{\rho(\Minedge)}(y) + 1 - \dist^{\Op}(t) 
& \text{(Lemma~\ref{lma:volumeBound})} \\
&\geq \dist^{\rho(\Minedge)}(z) + \dist^{\Op}(z) - \dist^{\Op}(t) & \text{(Lemma~\ref{lma:fundamental2} (P1))}\\
&\geq \dist^{\rho(\Minedge)}(z). & \text{($z \in V(T(t))$ and Lemma~\ref{lma:fundamental} (P1))}
\end{align*}
If $z \neq t$ holds, the last inequality becomes strict (i.e., changes to ``$>$'') by 
Lemma~\ref{lma:fundamental} (P1). Otherwise, 
by Lemma~\ref{lma:fundamental2} (P2) and the condition 
$\{y, z\} = \{y, t\} \not\in \paredgeset(t)$,
the second inequality becomes strict. In any case, one can conclude 
$\dist^{\Nop}(t) > \dist^{\rho(\Minedge)}(z)$.
Hence we obtain $\dist^{\beta}(t) > \max\{\dist^{\rho(\Minedge)}(y),\, 
\dist^{\rho(\Minedge)}(z)\} = \hlevel(\Minedge) = \hlevel(\Minedge_T(t))$.
\end{proof}

\begin{lemma} \label{lma:rootIncidence}
For $t \in V(T)$ with $\dist^\beta(t) < \infty$, let $e = \{y, t\}$ be any non-tree edge incident to $t$.
Then $\level(e) > \level(\Minedge_T(t))$ holds
unless both $e \not\in \paredgeset(t)$ and $\rho(e) = \Op(t)$ are satisfied.
Particularly, if $e \not\in \paredgeset(t)$ and $\rho(e) = \Nop(t)$, $\vlevel(e) > \vlevel(\Minedge_T(t))$
holds.
\end{lemma}

\begin{proof}
First, consider the case of $e \not\in \paredgeset(t)$ and $\rho(e) = \Nop(t)$.
Then we have
\begin{align}
\vlevel(e) &= (\dist^{\rho(e)}(y) + 1) + \dist^{\rho(e)}(t)    \nonumber \\
&> \dist^{\Op}(t) + \dist^{\Nop}(t) & \text{(Lemma~\ref{lma:fundamental2} (P2) and 
$e \not\in \paredgeset(t)$)} 
\nonumber \\
&\geq \vlevel(\Minedge_T(t)). & \text{(Lemma~\ref{lma:volumeBound})} \nonumber 
\end{align}
In the case of $e \in \paredgeset(t)$, the parity of any alternating path from $f$ 
terminating with $e$ is $\Op(t)$, which implies $\rho(e) = \Nop(t)$. Applying Lemma~\ref{lma:fundamental2} (P1) 
instead of (P2), we can show $\vlevel(e) \geq \vlevel(e')$ in the same way as above. Hence it suffices to show 
$\hlevel(e) > \hlevel(\Minedge_T(t))$. By the condition of $e \in \paredgeset$ and Lemma~\ref{lma:fundamental} (P2), 
we have $\dist^{\rho(e)}(y) + 1 = \dist^{\Nop(t)}(y) + 1 = \dist^{\Op(t)}(t) < \dist^{\Nop(t)}(t) = \dist^{\rho(e)}(t)$
(recall $\rho(e) = \Nop(t)$). Hence we obtain $\hlevel(e) = \max\{\dist^{\rho(e)}(y),\, \dist^{\rho(e)}(t)\} = 
\dist^{\Nop(t)}(t)$. By Lemma~\ref{lma:heightBound}, we conclude $\hlevel(e) > \hlevel(\Minedge_T(t))$. 
\end{proof}

Now we are ready to prove Theorem~\ref{thm:IKY24}.

\begin{proof}[Proof of Theorem~\ref{thm:IKY24}]
Let $\Minedge = \{y, z\}$. The proof follows the induction on $\dist^{\theta}(t)$. 

\textbf{Basis:}
Consider the case when $\dist^{\theta}(t) = 1$.
In this case, $t$ is a neighbor of $f$ and $\theta = \Op(t)$.
The shortest alternating path from $f$ to $t$ consists of a single tree edge $\{f, t\}$, whose level is zero by definition.
Hence the theorem obviously holds.

\textbf{Induction step:}
Suppose as the induction hypothesis that the theorem holds for any $(t', \theta')$ such that $\dist^{\theta'}(t') < k$, and consider $(t, \theta)$ such that $\dist^{\theta}(t) = k$. 

\textbf{(Case 1)} When $\theta = \Op(t)$:  
Let $Q$ be the maximal upward alternating path from $t$ in $T$, and $v$ be the upper endpoint of $Q$ (Fig.~\ref{fig:IKY24-1}). 
The maximality of $Q$ implies that both $\{v, \parent(v)\}$ and the last edge $\{v, v'\}$ of $Q$ are non-matching edges.
By the definition of $\parset$ and Lemma~\ref{lma:fundamental} (P2), $\OL{Q}$ must be the suffix of a shortest alternating path $P$ from $f$ to $t$.
Since the length of $Q$ is positive, 
$|P[f, v]| < |P| = k$ holds. Then one can apply the theorem for $(v, \Even)$ by the induction hypothesis, which results 
in a shortest $\Even$-alternating path $X$ from $f$ to $v$ of level at most $\level(e')$, where $e'$ is any edge in $\Minset_T(v)$.
By applying Lemma~\ref{lma:oddOrthodoxProp}, we obtain $\vlevel(e') < \vlevel(\Minedge)$, and thus the volume of $X \circ Q$ is less than $\vlevel(\Minedge)$.

\textbf{(Case 2)} When $\theta = \Nop(t)$: By Lemma~\ref{lma:heightBound}, we have 
$\dist^{\Nop}(t) > \hlevel(\Minedge) \ge \dist^{\rho(\Minedge)}(y)$. We apply 
the theorem inductively to $(y, \rho(\Minedge))$. It provides a shortest $\rho(\Minedge)$-alternating path 
$Y$ from $f$ to $y$ of level at most $\level(\Minedge_T(y))$. If $\Minedge$ is not an MIE of $T(y)$, 
the level of $Y$ is smaller than $\level(\Minedge)$. Otherwise, we have $\rho(\Minedge) = \Op(y)$ by 
Lemma~\ref{lma:rootIncidence}. Then, by the condition of this theorem, we have $\vlevel_T(Y) < \vlevel(\Minedge)$.
Thus, in any case, we have $\level_T(Y) < \level(\Minedge)$.
Similarly, one can also obtain
a shortest $\rho(\Minedge)$-alternating path $Z$ from $f$ to $z$ such that $\level_T(Z) < \level(\Minedge)$ holds.
An illustrative example is given in Fig~\ref{fig:IKY24-2}. Since $Y$ and $Z$ have levels less than $\level(\Minedge)$, 
they contain no edge in $\Outset_T(t)$. Hence, $Y$ lies at the outside of $T(t)$, and $Z$ enters $T(t)$ exactly once, using the tree edge $\{\parent_T(t), t\}$. Then $Z[t, z]$ lies at the inside of $T(t)$, i.e., $Y$ and 
$Z[t, z]$ are completely disjoint. It implies that $X = Y \circ \Minedge \circ \OL{Z}[z, t]$ is an unorthodox alternating path from $f$ to $t$ of 
length at most $|Y| + 1 + |Z| - \dist^{\Op}(t) = \vlevel(\Minedge) - \dist^{\Op}(t)$. It is upper bounded by $\dist^{\Nop}(t)$ by Lemma~\ref{lma:volumeBound}, 
and thus the shortest. Since $\level_T(Y) < \level(\Minedge)$ and $\level_T(Z) < \level(\Minedge)$ hold, this path satisfies 
the condition of the theorem.
\end{proof}

As a byproduct, we obtain the following lemma. (Recall that the length of a shortest augmenting path is $2\ell + 1$.)

\begin{lemma} \label{lma:wastingEdge}
Any shortest augmenting path $X$ contains no edge of volume larger than $2\ell + 5$, i.e., 
$\vlevel(X) \leq 2\ell + 5$, and does not contain a vertex $v$ such that $\dist^{\Op}(v) > \ell + 2$.
\end{lemma}

\begin{proof}
Let $u_0$ and $u_1$ be the endpoints of $X$, and $\hat{X}$ be the odd cycle obtained 
from $X$ by adding the two length-two paths from $f$ to $u_0$ and from $f$ to $u_1$.
Let $\{y, z\}$ be an edge of volume larger than $2\ell + 5$, and suppose for contradiction that 
$\hat{X}$ is of the form $Y \circ \{y, z\} \circ Z$.  Since $|Y| \geq \dist^{\rho(\{y,z\})}(y)$ and 
$|Z| \geq \dist^{\rho(\{y,z\})}(z)$ hold, $|\hat{X}| \geq 
\dist^{\rho(\{y,z\})}(y) + \dist^{\rho(\{y,z\})}(z) + 1 = \vlevel(\{y,z\}) > 2\ell + 5$, 
i.e., $|X| > 2\ell + 1$. It is a contradiction. 

Suppose for contradiction that $X$ contains a vertex $v$ of $\dist^{\Op}(v) > \ell + 2$.
Let $v'$ be the successor of $v$ in $X$, where we assume that $v$ and $v'$ appears in the order of $u_0, v, v', u_1$ 
without loss of generality. If $\dist^{\rho(\{v, v'\})}(v') > \ell + 1$, the volume of $\{v, v'\}$ is larger
than $2\ell + 5$. Otherwise, $\OL{X}[u_1, v]$ is an alternating path from $u_1$ to $v$, and hence $\dist^{\alpha}(v) \le \dist^{\rho(\{v, v'\})}(v') + 1 \le \ell + 2$. In any case,
we obtain a contradiction.
\end{proof}

This also implies the following corollary, which is used later.

\begin{corollary} \label{corol:IKY24}
Let $T$ be any ABT.
For any vertex $v$ satisfying $\Outset_T(v) \neq \emptyset$, $\vlevel(\Minedge_T(v)) = \dist^{\Op}(v) + \dist^{\Nop}(v)$ holds.
If $v$ is a free vertex, then $\vlevel(\Minedge_T(v)) \ge 2\ell + 5$.
Particularly, $v$ is an endpoint of a shortest augmenting path if and only if $\vlevel(\Minedge_T(v)) = 2\ell + 5$.
\end{corollary}

\begin{figure*}[tb]
\begin{minipage}{0.48\columnwidth}
    \centering
    \includegraphics[width=0.8\columnwidth]{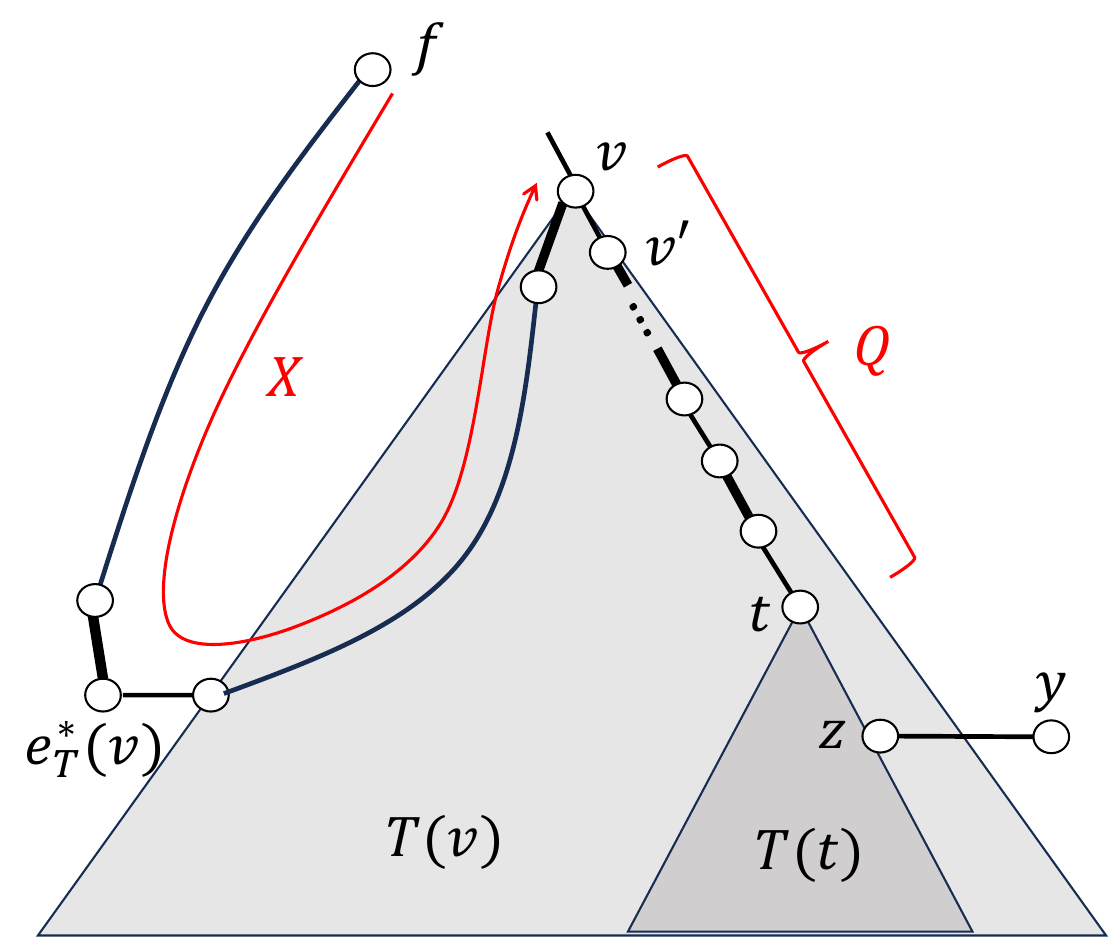}
    \caption{Proof of Theorem~\ref{thm:IKY24} (S1).}
    \label{fig:IKY24-1}
\end{minipage}
\begin{minipage}{0.55\columnwidth}
    \centering
    \includegraphics[width=0.8\columnwidth]{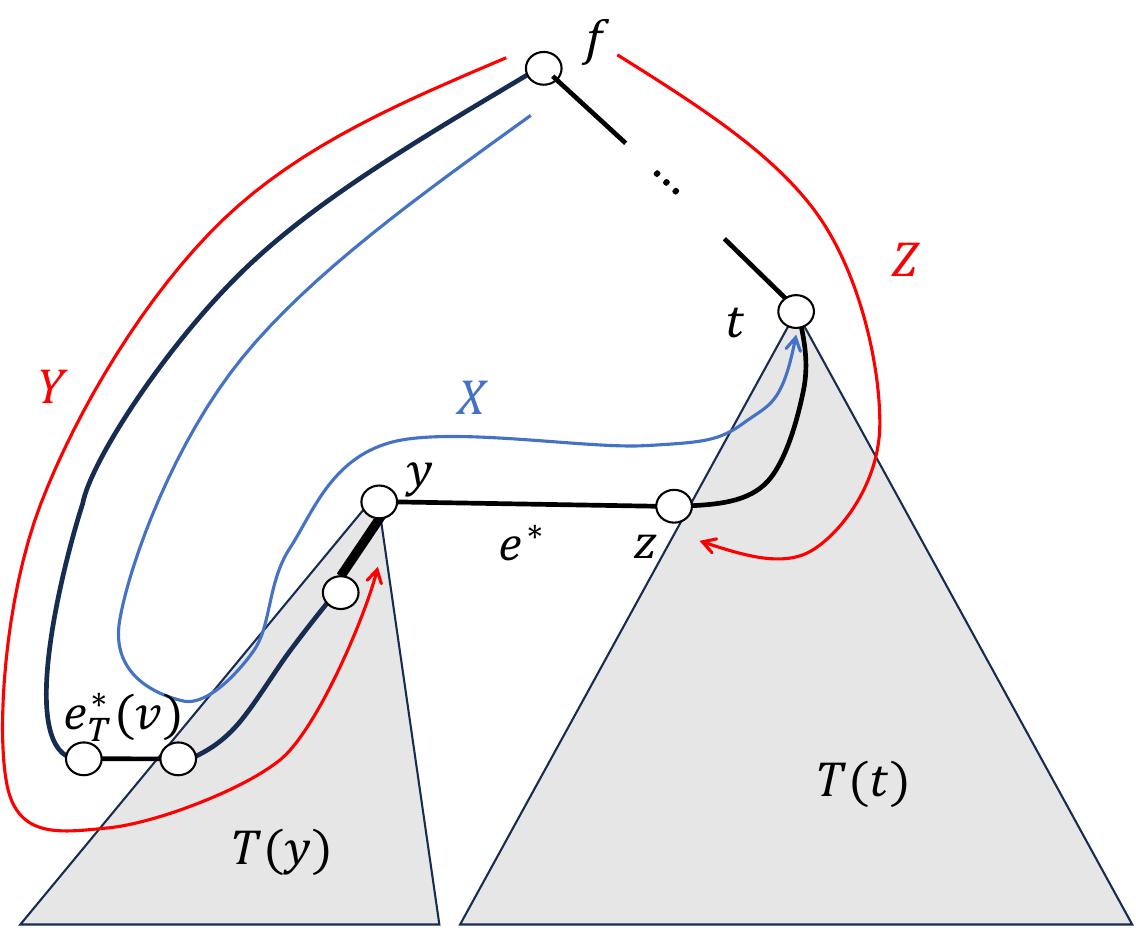}
    \caption{Proof of Theorem~\ref{thm:IKY24} (S2)}
 \label{fig:IKY24-2}
\end{minipage}
\end{figure*}

Note that the proof of Theorem~\ref{thm:IKY24} is constructive: provided that $T$ is given and an MIE 
$\Minedge_T(u) \in \Minset_T(u)$ is known for each $u \in V(G) \setminus \{f\}$, one can construct a shortest 
$\theta$-alternating path $P$ from $f$ to any vertex $t$ (with $\dist^\theta(t) < \infty$) of volume at most $\vlevel(\Minedge_T(t))$. The pseudocode of the algorithm 
is given in Algorithm~\ref{alg:pathConstruction} 
(this algorithm is essentially the same as the one presented in~\cite{IKY24}).
The procedure \textsc{PathConstruction}($T, s, t, \theta$) is called for $s$ and its descendant $t$ in $T$,
and returns the suffix $P[s, t]$ of a shortest $\theta$-alternating path $P$ from $f$ to $t$ shown in Theorem~\ref{thm:IKY24}.

\begin{algorithm}[ht] 
\caption{Procedure \textsc{PathConstruction}}
\label{alg:pathConstruction}
{\setlength{\baselineskip}{14pt}
\begin{algorithmic}[1]
\Procedure{PathConstruction}{$T, s, t, \theta$}
\If {$s = t$}
  \Return $s$
\EndIf
\If {$\theta = \Op(t)$} 
  \State \Return  \Call{PathConstruction}{$T, s, \parent_T(t), \OL{\theta}$}${} \circ \{\parent_T(t), t\}$.
\Else
  \State Let $\Minedge_T(t) = \{y, z\}$.
  \State $Y' \ot$ \Call{PathConstruction}{$T, s, y, \rho(e)$}
  \State $Z' \ot$ \Call{PathConstruction}{$T, t, z, \rho(e)$}
  \State \Return $Y' \circ \{y, z\} \circ \OL{Z'}$
\EndIf
\EndProcedure
\end{algorithmic}
}
\end{algorithm}

\begin{lemma} \label{lma:pathConstruction}
Let $T$ be any ABT, $s$ and $t$ be two vertices such that $t$ is a descendant of $s$ and admits $\theta$-alternating path from $f$. The procedure \textsc{PathConstruction}$(T, s, t, \theta)$ returns the suffix $P[s, t]$ of a shortest $\theta$-alternating path $P$ from $f$ to $t$ of level at most $\level(\Minedge_T(t))$. 
The algorithm runs in $O(|P[s, t]|)$ time. %
\end{lemma}

\begin{proof}
The correctness follows from the proof of Theorem~\ref{thm:IKY24}. The running time is easily analyzed by the straightforward 
inductive argument on $|P[s, t]|$.  
\end{proof}

\subsection{From a Double Path to a Shortest Augmenting Path}\label{subsec:DPtoSAP}
\begin{definition}[crossable edge]\label{def:crossable}
An edge $e = \{y, z\}$ is called \emph{crossable} if $\vlevel(e) = 2\ell + 5$ and there exists an ABT $T$ such that $\freeroot_T(y) \neq \freeroot_T(z)$ holds; then we say that $T$ \emph{respects} $e$.
We denote the set of all crossable edges that are critical (i.e., not in $\paredgeset$) by $\Ccset$.
\end{definition}

We formalize the \emph{alternating base DAG} $H$ for $\Ical$ and the notion of 
\emph{double paths} which are mentioned in Theorem~\ref{thm:structuralTheorem}. 

\begin{definition}[Alternating Base DAG]
An \emph{alternating base DAG} (ABD) $H$ for $\Ical$ is defined as follows: 
\begin{itemize}
    \item The vertex set $V(H)$ is $V(G)$ minus $f$ and its neighbors. 
    \item The edge set $E(H)$ is the set of directed edges in $\paredgeset$ with upward orientation. More precisely, $E(H) = \{ (u, v) \mid u \in V(H) \setminus U,\ v \in \parset(u) \}$. It implies that each vertex in $U$ becomes a sink.
\end{itemize}
\end{definition}

\begin{definition}[double path]
A \emph{double path} of $(H, \Ccset)$ is a 4-tuple 
$(P, Q, y, z)$, where $\{y, z\} \in \Ccset$ and $P$ and $Q$ are two vertex-disjoint paths in $H$ respectively starting from $y$ and $z$ and terminating with two distinct vertices in $U$.
\end{definition}

While the entries $y$ and $z$ are redundant, we include them in the definition for
clarification of the edge in $\Ccset$ associated with it; e.g., we say that an ABT $T$ \emph{respects} a double path $(P, Q, y, z)$ if it respects the crossable edge $e = \{y, z\}$.

We present the algorithm \textsc{DoubleToAug} transforming a given double path $(P, Q, y, z)$ into 
a shortest augmenting path. The pseudocode of the algorithm is 
shown in Algorithm~\ref{alg:AugPathConstruction}. Let $u$ and $u'$ be the last free vertices 
of $P$ and $Q$. It first constructs an ABT $T'$ respecting $(P, Q, y, z)$ by 
choosing the parents of vertices in $V(P) \cup V(Q)$ along $P$ and $Q$, which results in the ABT $T'$
where $\{y, z\}$ bridges $T(u)$ and $T(u')$ (i.e., $\freeroot_T(y) = u$ and $\freeroot_T(z) = u'$). 
By applying \textsc{PathConstruction}, we obtain two shortest $\rho(e)$-alternating paths $Y$ 
from $\freeroot_T(y)$ to $y$ and $Z$ from $\freeroot_T(z)$ to $z$. The algorithm outputs the path 
$Q = Y \circ \{y, z\} \circ \OL{Z}$. It is guaranteed that the constructed path is localized within 
$T(\freeroot_T(y))$ and $T(\freeroot_T(z))$.

\begin{algorithm}[ht] 
\caption{Procedure \textsc{DoubleToAug}}
\label{alg:AugPathConstruction}
{\setlength{\baselineskip}{14pt}
\begin{algorithmic}[1]
\Procedure{DoubleToAug}{$P, Q, y, z$}
\State Construct any ABT $T$ respecting $(P, Q, y, z)$.
\State $Y \ot$ \Call{PathConstruction}{$T, \freeroot_T(y), y, \rho(\{y, z\})$}
\State $Z \ot$ \Call{PathConstruction}{$T, \freeroot_T(z), z, \rho(\{y, z\})$}
\State \Return $Y \circ \{y, z\} \circ \OL{Z}$
\EndProcedure
\end{algorithmic}
}
\end{algorithm}

\begin{theorem} \label{thm:augPathConst}
Let $(P, Q, y, z)$ be any double path. Then, \textsc{DoubleToAug}$(P, Q, y, z)$ 
outputs a shortest augmenting path $X$ from $\freeroot_T(y)$ to $\freeroot_T(z)$ satisfying 
$V(X) \subseteq V(T(\freeroot_T(y))) \cup V(T(\freeroot_T(z)))$, where $T$ is the ABT constructed
by the algorithm.
\end{theorem}

\begin{proof}
Let $e = \{y, z\}$ for short. By Corollary~\ref{corol:IKY24} and the fact of $\vlevel(e) = 2\ell + 5$, to prove $V(X) \subseteq V(T(\freeroot_T(y))) \cup V(T(\freeroot_T(z)))$,
it suffices to show that the output path $Y$ (resp.\ $Z$) satisfies 
$\vlevel_T(Y) < \vlevel(e)$ (resp.\ $\vlevel_T(Z) < \vlevel(e)$). By symmetry, we prove only 
$\vlevel_T(Y) < \vlevel(e)$. By Lemma~\ref{lma:rootIncidence}
and the fact of $\{y, z\} \not\in \paredgeset$, if $\rho(\{y, z\}) = \Nop(y)$, we have 
$\vlevel(e) > \vlevel(\Minedge_T(y))$. Then $\textsc{PathConstruction}$ outputs the path $Y$ of volume 
$\vlevel(\Minedge_T(y)) < \vlevel(e)$. Otherwise, $Y$ is an $\Op(y)$-alternating path, and thus 
Theorem~\ref{thm:IKY24} guarantees $\vlevel_T(Y) < \vlevel(e)$. 

Since $Y$ and $Z$ are 
the shortest, $|X| = |Y| + 1 + |Z| = (\dist^{\rho(e)}(y) - 2) + 1 + (\dist^{\rho(e)}(z) - 2) = 
\vlevel(e) - 4 = 2\ell + 1$, i.e., $X$ is a shortest augmenting path.
\end{proof}

\subsection{Constructing a Maximal Set of Disjoint Shortest Augmenting Paths}
\label{sec:maximalSet}

As explained in Section~\ref{subsec:maximalAugPath}, the procedure \textsc{DoubleToAug} can be safely 
applied simultaneously to a set $\Dcal$ of vertex-disjoint double paths, which provides a set 
of disjoint shortest augmenting paths (referred to as $\mathsf{Aug}(\Dcal)$). In this section, we present that 
such a transformation preserves the maximality, i.e., $\mathsf{Aug}(\Dcal)$ is a maximal 
set of disjoint shortest augmenting paths.

\begin{lemma} \label{lma:uniqueCrossableEdge1}
Let $X$ be any shortest augmenting path from $u_1$ to $u_2$ $(u_1, u_2 \in U)$.
If $X$ contains two edges $e_1$ and $e_2$ of volume $2\ell + 5$, at least one of them 
belongs to $\paredgeset$. 
\end{lemma}

\begin{proof}
\noindent
Let $e_1 = \{y_1, z_1\}$ and $e_2 = \{y_2, z_2\}$.  
Without loss of generality, we assume that the order of these endpoints along $X$ is $y_1, z_1, y_2, z_2$.
Let $v_1$ be the vertex in $V(X[z_1, y_2])$ such that $\dist^{\Op}(v_1)$ is the minimum, 
and $v_0$ and $v_2$ be its immediate predecessor and successor in $X$, respectively (see Fig.~\ref{fig:augToDownup}). 
Obviously, $\rho(\{v_0, v_1\}) \neq \rho(\{v_1, v_2\})$ holds. By symmetry, we assume 
$\Op(v_1) = \rho(\{v_0, v_1\})$ without loss of generality. Let $W$ be any shortest 
orthodox path from $f$ to $v_1$. 
By the choice of $v_1$, any vertex $v' \in V(X[z_1, v_1])$ satisfies $\dist^{\Op}(v') \geq \dist^{\Op}(v_1) = |W|$, 
and any vertex $w' \in V(W) \setminus \{v_1\}$ satisfies $\dist^{\Op}(w') < |W|$. Hence $V(W) \cap V(\OL{X}[v_1, z_1]) = \{v_1\}$ holds and $W \circ \OL{X}[v_1, z_1]$ is a $\rho(e_1)$-alternating path from $f$ to $z_1$, i.e., $|W| + |\OL{X}[v_1, z_1]| \geq 
\dist^{\rho(e_1)}(z_1)$. Then we have
\begin{align*}
|\OL{X}[u_2, v_1]|
&= 2\ell + 1 - |\OL{X}[v_1, z_1]| - |X[u_1, z_1]| \\
&= \vlevel(e_1) - 4 - |\OL{X}[v_1, z_1]| - |X[u_1, z_1]| & \text{($\vlevel(e_1) = 2\ell + 5$)} \\
&= (\dist^{\rho(e_1)}(y_1) + \dist^{\rho(e_1)}(z_1) + 1 - 4) \\
&  \hspace*{2cm} - |\OL{X}[v_1, z_1]| - (|X[u_1, y_1]| + 1) \\
&\leq \dist^{\rho(e_1)}(z_1) - |\OL{X}[v_1, z_1]| - 2 & \text{($|X[u_1, y_1]| + 2 \geq \dist^{\rho(e_1)}(y_1)$)} \\
&\leq |W| - 2.
\end{align*}
Letting $Q$ be the length-two alternating path from $f$ to $u_2$,
the inequality above implies that $|Q \circ \OL{X}[u_2, v_1]|$ is a shortest orthodox alternating path from $f$ to $v_1$, and thus $\{v_1, v_2\} \in \paredgeset(v_1)$ holds. It also implies $\dist^{\Op}(v_2) < \dist^{\Op}(v_1)$ by 
Lemma~\ref{lma:fundamental} (P1). By the choice of $v_1$, we have 
$v_2 \not\in V(X[z_1, y_2])$. That is, $\{v_1, v_2\} = \{y_2, z_2\}$ and thus 
$\{y_2, z_2\} \in \paredgeset(y_2)$ holds. 
\end{proof}

\begin{figure*}[tb]
    \centering
    \includegraphics[width=0.7\columnwidth]{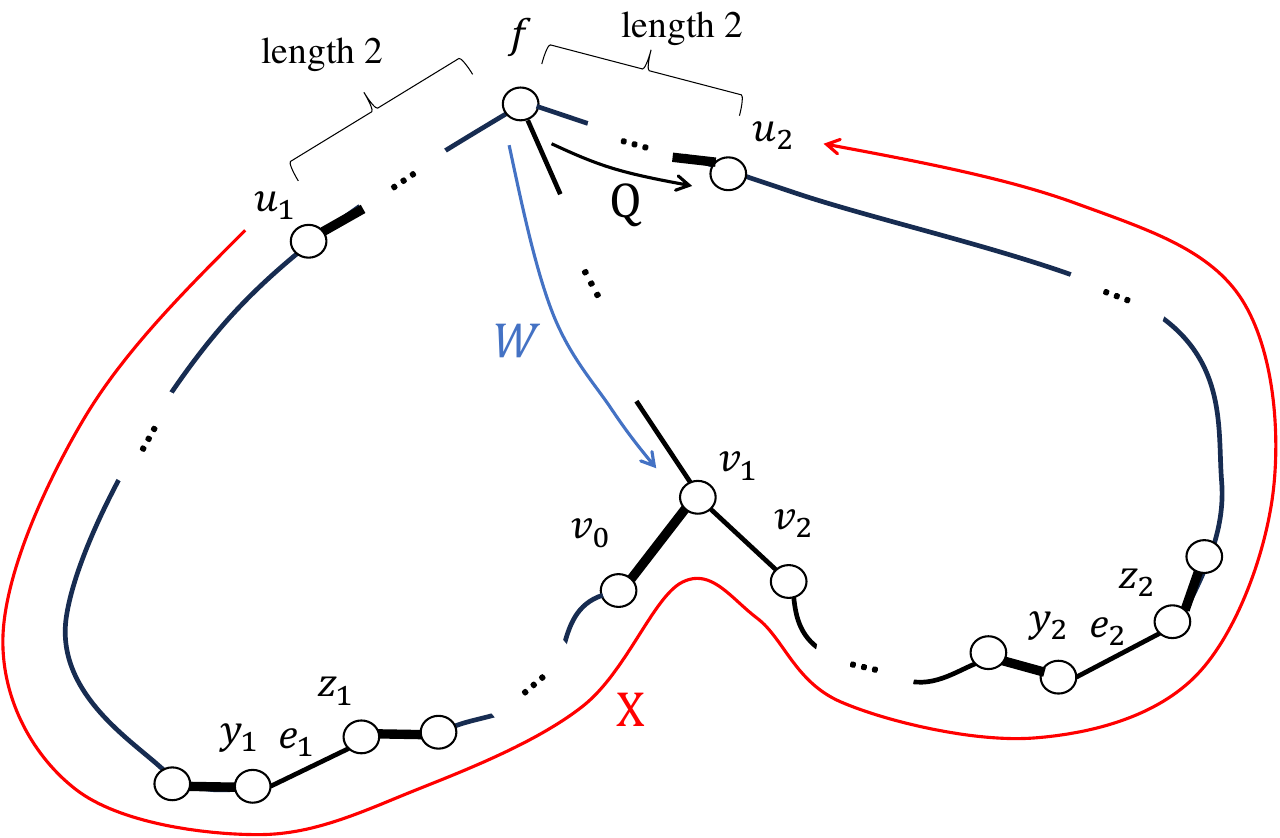}
    \caption{Proof of Lemma~\ref{lma:uniqueCrossableEdge1}.}
    \label{fig:augToDownup}
\end{figure*}

\begin{lemma} \label{lma:uniqueCrossableEdge2}
Any shortest augmenting path $Q$ contains exactly one edge in $\Ccset$ (i.e., crossable and critical).
\end{lemma}

\begin{proof}
By the definition of $\Ccset$ and Lemma~\ref{lma:uniqueCrossableEdge1}, $Q$ does not 
contain two edges in $\Ccset$.
Hence it suffices to show that $Q$ contains at least one edge in $\Ccset$.
Let $Q = v_0, e_1, v_1, e_2, v_2, \dots, e_k, v_k$ (where $k = |Q|$). 
For any $i$, the parities of $e_i$ and $e_{i+1}$ are different, and thus at most one of $e_i$ and 
$e_{i+1}$ can belong to $\paredgeset(v_i)$. Hence one can construct an ABT $T'$ containing all the edges in $E(Q) \cap \paredgeset$. In that $T'$, 
there exists a non-tree edge in $E(Q)$ bridging $T'(v_0)$ and $T'(v_k)$. It is crossable and critical by Lemma~\ref{lma:wastingEdge} and Corollary~\ref{corol:IKY24}. %
\end{proof}

This lemma also provides the following corollary (the construction of $T'$ follows from the proof of 
Lemma~\ref{lma:uniqueCrossableEdge2}).

\begin{corollary} \label{corol:uniqueCrossableEdge2}
Let $Q$ be any shortest augmenting path from $u_0$ to $u_1$.
Then one can construct an ABT $T'$ from any ABT $T$ by the following rule: For any $\{v, v'\} \in E(Q)$, if $(v, v') \in \paredgeset(v)$ is an edge in $Q$, $v$ changes its parent to $v'$. In the constructed ABT $T'$, the (unique) edge $e \in E(Q) \cap \Ccset$ bridges $T'(u_0)$ and $T'(u_1)$ (i.e., $T'$ respects $e$).
\end{corollary}

\begin{lemma} \label{lma:crossableIncidentEdge}
Let $v$ be a vertex with a crossable edge in $\paredgeset(v)$.
Then, for any ABT $T$, $\vlevel(\Minedge_T(v)) = 2\ell + 5$ holds. 
\end{lemma}

\begin{proof}
Let $\Minedge_T(v) = \{y, z\}$.
By Corollary~\ref{corol:IKY24}, it suffices to show that there exists an ABT $T'$ such that $\freeroot_T(y) \neq \freeroot_T(z)$.
Let $S_y$ and $S_z$ be the set of all vertices reachable from $y$ or $z$ in $H$, and $H'$ be the subgraph of $H$ induced by
$S_y \cup S_z$. Note that $H'$ necessarily contains $v$. If $H'$ admits a double path 
$D = (Y, Z, y, z)$, one can obtain $T'$ such that $\{y, z\}$ bridges two components in $F_{T'}$, and the lemma is proved.
Otherwise, $H'$ has an articulation point $x$ that is reachable from both $y$ and $z$.
If $x$ locates between $z$ and $v$ (including $v$ itself), $\Minedge_T(v)$ cannot become an incoming edge of $T(v)$.
Hence $x$ is not $v$ and is reachable from $v$.
However, it contradicts that $v$ has a crossable edge $e \in \paredgeset(v)$ because both endpoints of $e$ always belongs to $T'(x)$ for any ABT $T'$.
The lemma is proved.
\end{proof}

Now we are ready to prove (the formal version of) Theorem~\ref{thm:structuralTheorem}.

\let\temp\thetheorem
\renewcommand{\thetheorem}{\ref*{thm:structuralTheorem}}
\begin{theorem}[formal]
For any maximal set $\Dcal$ of vertex-disjoint double paths, 
$\Augt(\Dcal)$ is a maximal set of disjoint shortest augmenting paths.
\end{theorem}
\let\thetheorem\temp
\addtocounter{theorem}{-1}

\begin{proof}
Suppose for contradiction that $\Augt(\Dcal)$ is not maximal. Then, there exists a shortest augmenting path $P$ such that $\{P\} \cup \Augt(\Dcal)$ is still a set of 
vertex-disjoint shortest augmenting paths. Let $u_0$ and $u_1$ be the first and last vertices of $P$. 
We consider the following two cases:

\textbf{(Case 1)} When $P$ does not intersect any double paths in $\Dcal$:
By Lemma~\ref{lma:uniqueCrossableEdge2}, $P$ contains a unique edge in $E(P) \cap \Ccset$,
which is denoted by $\{y, z\}$. Let $T$ be any ABT respecting 
$\Dcal$, and $T'$ be any ABT containing all the edges of volume $2\ell + 5$ in 
$\paredgeset \cap E(P)$, which is constructed from $T$ by Corollary~\ref{corol:uniqueCrossableEdge2}. 
The paths $Y$ from $y$ to $u_0$ and $Z$ from $z$ to $u_1$ in $T'$ forms a double path 
$D = (Y, Z, y, z)$ in $H$. By Corollary~\ref{corol:uniqueCrossableEdge2}, the construction of $T'$ 
changes only the parents of the vertices in $P$. Hence $T'$ still respects $\Dcal$. That is, 
$\{D\} \cup \Dcal$ is a set of disjoint double paths, but it contradicts the maximality of 
$\Dcal$. 

\textbf{(Case 2)} When $P$ intersects a double path in $\Dcal$: Letting $e$ be the unique edge in $E(P) \cap \Ccset$,
we represent $P$ as $P = P_1 \circ e \circ P_2$, where $P_1$ and $P_2$ do not contain any edge in $\Ccset$.
By symmetry, we assume that $P_1$ intersects a double path in $\Dcal$, and define 
$P_1$ as $u_0 = v_0, e_1, v_1, e_2, \dots, e_{k}, v_{k}$.
Let $T$ be any ABT respecting $\Dcal$.
We obtain an ABT $T'$ from $T$, which still respects $\Dcal$ and 
respects $P_1$ ``as much as possible'': for any $v_i \in V(P_1) \setminus \{v_0, v_{k}\}$, 
if $v_i$ is not a vertex in $V(\Dcal)$ and $v_{i+1} \in \parset(v_i)$ (resp.\ $v_{i-1} \in \parset(v_i)$) holds, $v_i$ chooses $v_{i+1}$ (resp.\ $v_{i-1}$) as 
its parent. Let $\{v_j, v_{j+1}\}$ be the first edge in $\Outset_{T'}(u_0) \cap E(P_1)$, so that
$P_1[u_0, v_j]$ is contained in $T'(u_0)$.
Since $\freeroot_{T'}(v_{j+1}) \neq \freeroot_{T'}(v_j) = u_0$, $\{v_j, v_{j+1}\}$ is a crossable edge (by Lemma~\ref{lma:wastingEdge} and Corollary~\ref{corol:IKY24}).
As $\{v_j, v_{j+1}\} \not\in \Ccset$, we have $\{v_j, v_{j+1}\} \in \paredgeset$ by definition.
Then, by the construction of $T'$, there exists a double path $D = (Y', Z', y', z') \in \Dcal$ hitting $v_j$ or $v_{j+1}$.
Without loss of generality, we assume that $Y'$ is the path hitting $v_j$ or $v_{j+1}$. 

Suppose that $Y'$ hits $v_{j}$. Since $T'$ respects $D$, $Y'$ is a path in $T'$ as well as in $T$, and thus 
$\freeroot_T(y') = \freeroot_{T'}(y') = u_0$ holds. It contradicts that $P$ is disjoint from $\Augt(\Dcal)$.
Otherwise, $Y'$ hits $v_{j+1}$.
If $v_{j+1} \in \parset(v_j)$, we should have $\parent_{T'}(v_j) = v_{j+1}$ (as $Y'$ does not hit $v_j$), contradicting $\freeroot_{T'}(v_j) \neq \freeroot_{T'}(v_{j+1})$.
Thus we obtain $v_{j} \in \parset(v_{j+1})$ and then $\vlevel(\Minedge_{T'}(v_{j+1})) = 2\ell + 5$ by Lemma~\ref{lma:crossableIncidentEdge}.
In the construction by \textsf{DoubleToAug}, 
$Y'$ is transformed into an alternating path $Q$ from $y'$ to $\freeroot_{T'}(y')$ of volume less than 
$2\ell + 5$. Hence $Q$ leaves $T'(v_{j+1})$ through the edge $\{\parent_{T'}(v_{j+1}), v_{j+1}\}$, i.e., $Q$ must contain $v_{j+1}$. It contradicts that $P$ is disjoint from $\Augt(\Dcal)$, again. 
\end{proof}

\subsection{Auxiliary Lemmas}
In addition to the main theorem, we further present a few auxiliary lemmas, which will be used in Section~\ref{subsec:AugAndHit}.

\begin{lemma} \label{lma:hittingset}
Let $P$ be any shortest augmenting path from $u_1$ to $u_2$ with edge $\{y, z\} \in \Ccset$, where $u_1, y, z, u_2$ appear in this order along $P$, %
and $X$ be any subgraph of $H$ such that each connected component in $X$ is a path terminating with a vertex in $U$ or an in-tree rooted by a 
vertex in $U$. If $V(P) \cap V(X) = \emptyset$, then $H$ admits two disjoint paths $Q_y$ from $y$ to $u_1$ and $Q_z$ from $z$ to $u_2$ which are also disjoint from $V(X)$. 
\end{lemma}

\begin{proof}
The proof almost follows the argument of (Case 1) in the proof of Thorem~\ref{thm:structuralTheorem}. 
Since the connected components in $X$ are disjoint, one can construct an ABT $T$ which respects every component in $X$ (where
``$T$ respects a tree $T'$'' means that $T$ contains $T'$ as a subgraph). Let $T^{\ast}$ be any ABT 
containing all the edges of volume $2\ell + 5$ in $\paredgeset \cap E(P)$, which is constructed from $T$ by 
Corollary~\ref{corol:uniqueCrossableEdge2}. Due to the disjointness of $P$ and $X$, $T^{\ast}$ still respects $X$, and admits the paths 
$Q_y$ from $y$ to $u_1$ and $Q_z$ from $z$ to $u_2$ along $T^{\ast}$. Since any component in $X$ contains one vertex in $U$, the 
subtrees $T^{\ast}(u_1)$ and $T^{\ast}(u_2)$ does not contain any vertex in $V(X)$. That is, $Q_y$ and $Q_z$ satisfy the condition 
of the lemma.
\end{proof}

We define a \emph{partial double path} to $v$ as a tuple $(P, Q, y, z)$ such that $P$ and $Q$ are disjoint, $\{y, z\} \in \Ccset$, and
$P$ and $Q$ are respectively a $y$--$U$ path and a $z$--$v$ path. We have the following lemma.

\begin{lemma} \label{lma:partialreconstruction}
Let $T$ be any ABT, $u \in U$ be any free vertex, and $v$ be any vertex in $V(T(u))$. If $v$ has an edge $(v, v') \in \paredgeset$ 
such that $v' \not \in V(T(u))$ and there exists a partial double path $(P, Q, y, z)$ to $v$ satisfying $V(P) \subseteq V(T(u))$ and 
$V(Q) \subseteq V(T(u))$, there exists a shortest $\Nop(v)$-alternating path from $f$ to $v$ containing $u$ and $\Nop(v) = \Even$. 
\end{lemma}

\begin{proof}
One can split $T(u)$ into two subtrees $T_1$ and $T_2$ such that $V(T_1) \cup V(T_2) = V(T(u))$, $V(P) \subseteq V(T_1)$, and $V(Q) \subseteq V(T_2)$. Due to the existence of $(v, v')$, we obtain the ABT $T'$ by moving $T_2$ 
to the subtree containing $v'$. By applying $\textsc{PathConstruction}(T', u, y, \rho(\{y, z\}))$ and 
$\textsc{PathConstruction}(T', v, z, \rho(\{y,z\}))$, we obtain a shortest $\Nop(v)$-alternating path from $f$ to $v$ lying in $V(T_1) \cup V(T_2)$. Since we assume $|\paredgeset(v)| \geq 2$, the edge from $v$ to 
its parent is a non-matching edge. Hence $\Nop(v) = \Even$ holds. 
The lemma is proved.
\end{proof}

We also present a slightly generalized version of Corollary~\ref{corol:uniqueCrossableEdge2}.

\begin{lemma} \label{lma:uniqueCrossableEdge3}
Let $T$ be any ABT, $u \in U$ be any free vertex, and $v$ be any vertex in $V(T(u))$.
If $v$ has an edge $(v, v') \in \paredgeset$ such that $v' \not \in V(T(u))$ and there exists 
a shortest $\Nop(v)$-alternating path $X$ from $u$ to $v$ such that $V(X) \subseteq V(T(u))$, then there exists a partial double path to $v$ lying in $T(u)$. %
\end{lemma}

\begin{proof}
Following the approach of Corollary~\ref{corol:uniqueCrossableEdge2}, one can construct an ABT $T'$ where all edges in $E(X) \cap \paredgeset$ of volume $2\ell + 5$ are contained in $T'$.
Let $e = \{y, z\}$ be last incoming edge of $T'(v)$ along $X$, which is not an edge in $\paredgeset$. 
Due to the existence of $(v, v')$, for any edge $\{y', z'\}$ crossing between $V(T'(u)) \setminus V(T'(v))$ and $T'(v)$, there exists an ABT $T''$ such that $\freeroot_{T''}(y') \neq \freeroot_{T''}(z')$. It implies that the volume of $\{y', z'\}$ is at least $2\ell + 5$ by Corollary~\ref{corol:IKY24}.
Since $X$ cannot contain any edge of volume larger than $2\ell + 5$ by Lemma~\ref{lma:wastingEdge}, the volume of $\{y, z\}$ is $2\ell + 5$, i.e., $e$ is crossable, and thus $e \in \Ccset$ holds.
Then the paths $P$ from $y$ to $u$ and $Q$ from $z$ to $v$ along $T'$ form a partial double path to $v$.
In the construction of $T'$ from $T$, a vertex $x$ in $X$ changes its parent to $x'$ only when $(x, x')$ is an edge in $X$.
As $V(X) \subseteq V(T(u))$, we have $V(T'(u)) = V(T(u))$, and hence both $P$ and $Q$ lie in $T(u)$.
\end{proof}

\section{The Details of Our MCM Algorithm}

\subsection{Implementation of Step 1}
\label{subsec:compDistDetails}

The pseudocode of our distance computation algorithm is presented in Algorithm~\ref{alg:constABT}. 
In the algorithm, each vertex $v$ manages variables $\Varparent[v]$, 
$\Vardist^{\Even}[v]$, and $\Vardist^{\Odd}[v]$, which are respectively 
initialized with $\infty$. We also introduce two auxiliary variables $\Varort[v]$ and $\Varvol[e]$ indicating the computed orthodox 
parity of $u$ and edge volume of $e$ respectively. Formally, they are defined as $\Varort[u] = \argmin_{\theta \in \{\Odd, \Even\}} \Vardist^{\theta}[u]$ and $\Varvol[\{y,z\}] = \Vardist^{\rho(\{y,z\})}[y] + \Vardist^{\rho(\{y,z\})}[z] + 1$.
The values of these variables are computed from $\Varparent[\cdot]$, $\Vardist^{\Even}[\cdot]$, and 
$\Vardist^{\Odd}[\cdot]$, and thus they are not explicitly stored in Algorithm~\ref{alg:constABT}.
Intuitively, the algorithm (implicitly) follows rounds $r = 0, 1, 2, \dots $ and keeps 
the following invariant at the beginning of each round $r$: 
\begin{enumerate}
\item[(I1)] For each $\theta \in \{\Odd, \Even\}$ and $v \in V(G)$, $\Vardist^{\theta}[v] = \dist^{\theta}(v)$ if $\dist^{\theta}(v) < r$.
\item[(I2)] For each $v \in V(G)$, $\Varparent[v] = v'$ for some $\{v, v'\} \in E(G)$ with $\rho(\{v, v'\}) = \OL{\Op(v)}$ and $\dist^{\OL{\Op(v)}}(v') = \dist^{\Op}(v) - 1$ if $0 < \dist^{\Op}(v) \le r$.
\end{enumerate}
We denote the graph induced by the edges $\{v, \Varparent[v]\}$ for all vertices with $\Vardist^{\Op(v)}(v) \leq r$ 
by $T^r$. Supposing (I2), it is a subtree of some ABT $T$.
Note that $T^r$ contains the vertex $v$ such that $\dist^{\Op}(v) = r$ holds, but $\Vardist^{\Op}[v]$ 
does not store $r$ yet at the beginning of round $r$. The algorithm also 
manages the tree $\tilde{T}^r$ obtained from $T^r$ by contracting all the edges $\{v, \Varparent[v]\}$ 
such that $\dist^{\Nop}(v) < r$ holds. Throughout the algorithm, that tree is implicitly managed 
by the union-find data structure $\Varcontraction$ storing the family of the vertex sets shrunk into 
the same single vertex. Obviously, each set of vertices shrunk into a single vertex $x$ induces a connected 
subtree of $T^r$. For simplicity, we assume that $\textsc{find}(v)$ operation of $\Varcontraction$ 
returns the root vertex of the contracted subtree containing $v$. Each vertex $v$ further manages 
a priority queue $\Varnontree[v]$, which stores the set of the non-tree edges $e$ incident to $v$ in 
$\tilde{T}^r$ such that its volume $\vlevel(e)$ is fixed. The priority of each edge follows the value of 
$\level$ (smaller value has higher priority). Each $\Varnontree[v]$ is implemented by any meldable heap 
of supporting \textsc{find-min}, \textsc{merge}, \textsc{add}, and \textsc{delete-min} with $O(\log n)$-amortized 
time.

\begin{algorithm}[t!] 
\caption{Procedure \textsc{ComputeDist}}
\label{alg:constABT}
{\setlength{\baselineskip}{14pt}
\begin{algorithmic}[1]
\State $\Varchecklist.\textsc{add}(f, 0, \Even)$ \Comment{Initialization}
\While {$\Varchecklist \neq \emptyset$} \Comment{Rounds $r = 0, 1, 2, \dots$}
  \State $(v, r, \theta) \ot \Varchecklist.\textsc{delete-min}()$
  \If {$\Vardist^{\theta}[v] = \infty$} \Comment{Deciding the distance of parity $\theta$}
    \State $\Vardist^{\theta}[v] \ot r$    
    \ForAll {$e = \{x, v\}$ incident to $v$}
      \If {$\Varparent[x] = \infty$ and $\rho(e) = \theta$} \Comment{Join $x$ into the tree}
        \State $\Varparent[x] \ot v$; $\Varchecklist.\textsc{add}((x, r + 1, \OL{\theta}))$ 
      \ElsIf {$\max\{\Vardist^{\rho(e)}[v], \Vardist^{\rho(e)}[x]\} = r$} \Comment{$\level(\{x, v\})$ is decided}
        \State $(x', v') \ot (\Varcontraction.\textsc{find}(x), \Varcontraction.\textsc{find}(v))$
        \State $\Varnontree[v'].\textsc{add}((x, v))$
        \State $\Varnontree[x'].\textsc{add}((v, x))$
        \If {$\Vardist^{\OL{\Varort[v']}}[v'] = \infty$} $\Varchecklist.\textsc{add}(v', \Varvol[\{x, v\}] - \Vardist^{\Varort[v']}[v'] - 1, \OL{\Varort[v']})$ \EndIf
        \If {$\Vardist^{\OL{\Varort[x']}}[x'] = \infty$} $\Varchecklist.\textsc{add}(x', \Varvol[\{x, v\}] - \Vardist^{\Varort[x']}[x'] - 1, \OL{\Varort[x']})$ \EndIf
      \EndIf
    \EndFor
    \If {$\theta = \OL{\Varort[v]}$} \Comment{Contraction}
      \State $\Varcontraction.\textsc{union}(v, \Varparent[v])$
      \State $v' \ot \Varcontraction.\textsc{find}(v)$
      \State $\Varnontree[v'].\textsc{merge}(\Varnontree[v])$ 
      \State$(y, z) \ot \Varnontree[v'].\textsc{find-min}()$
      \While {\Varcontraction.\textsc{find}($y$) = \Varcontraction.\textsc{find}($z$)} \Comment{Delete self-loops}
        \State $\Varnontree[v']$.\textsc{delete-min}() 
        \State $(y, z) \ot \Varnontree[v'].\textsc{find-min}()$  
      \EndWhile
      \If {$\Vardist^{\OL{\Varort[v']}}[v'] = \infty$} $\Varchecklist.\textsc{add}(v', \Varvol(\{y,z\}) - \Vardist^{\Varort[v']}[v'] - 1, \OL{\Varort[v']})$ \EndIf
    \EndIf
  \EndIf
\EndWhile
\end{algorithmic}
}
\end{algorithm}

The algorithm also manages the global priority queue $\Varchecklist$, which manages the set of the vertices 
processed in the subsequent execution. Each entry in $\Varchecklist$ is a triple $(v, r, \theta)$ for 
$v \in V(G)$, $r \in \mathbb{N}$, and $\theta \in \{\Even, \Odd\}$, which means that $\Vardist^{\theta}[v]$ 
must be decided with $r$ if it is still undecided at round $r$. Note that $(v, r, \theta) \in \Varchecklist$ does not
necessarily imply that $\dist^{\theta}(v) = r$ because some 
other entry $(v, r', \theta)$ of $r' < r$ might exists (in such a case, $(v, r, \theta)$ becomes a wasting entry). Intuitively, $(v, r, \theta)$ is added to $\Varchecklist$ when (1) a neighbor $x$ of $v$ joined into $T^{r-1}$ as a child of $v$ at round $r - 1$, or (2) $v$ obtains a new non-tree edge $e = \{y, z\} \in \Varnontree[v]$ in $\tilde{T}^{\hat{r}}$ ($T^{\hat{r}}$ after contracting all the edges $\{v', \Varparent[v']\}$ with $\dist^\Nop(v') < \hat{r}$) for some $\hat{r}$ such that $\theta = \OL{\Varort[v]}$ and $r = \Varvol[e] - \Vardist^{\Varort[v]}[v] - 1$ holds.
The second case occurs when $\Varvol[e]$ of some non-tree edge $e$ incident to $v$ is determined or $\Varnontree[v]$ is updated by contraction.
The beginning of round $r$ is defined as the first timing when no tuple $(\cdot, r', \cdot)$ of $r' < r$ is contained in $\Varchecklist$.
Note that $r$ in triples $(\cdot, r, \cdot)$ added to $\Varchecklist$ is actually nondecreasing throughout the algorithm, but we only need a weaker property shown in Lemmas~\ref{lma:orthodxfix}--\ref{lma:unorthodoxfix} because once $\Vardist^\theta[v]$ is fixed (which is correct by the lemmas) it will never change.

At each step, the algorithm dequeues one entry $(v, r, \theta)$ from $\Varchecklist$ (line 3).
If $\Vardist^{\theta}[v]$ has already been decided, the entry is simply skipped (line 4).
Otherwise, decides $\Vardist^{\theta}[v] = r$ (line 5).
Then it checks each neighbor $x$ of $v$. If $\rho(\{x, v\}) = \theta$ and 
$x$ is not involved into the tree $T^r$ (i.e., $\Varparent[x] = \infty$), set $\Varparent[x] = v$ and add $(x, r+1, \OL{\theta})$ to $\Varchecklist$ (lines 7--8), which corresponds to determining the orthodox distance $\dist^{\Op}(x)$.
Otherwise, $\{x, v\}$ becomes a non-tree edge.
If $\max\{\Vardist^{\rho(e)}[x],\, \Vardist^{\rho(e)}[v]\} = r$ holds, $\Varvol[\{x, v\}]$ is fixed at round $r$. For $v'$ and $x'$ into which $v$ and $x$ are shrunk, respectively, the algorithm adds $\{x, v\}$ to $\Varnontree[v']$ and $\Varnontree[x']$ (lines 10--12).
Since it means that $x'$ and $v'$ obtain a new non-tree incident edge whose volume is fixed, unless $\Vardist^{\OL{\Varort[v']}}[v']$ (resp.\ $\Vardist^{\OL{\Varort[x']}}[x']$) has already been decided, $(v', \Varvol[\{x, v\}] - \Vardist^{\Varort[v']}[v'] - 1, \OL{\Varort[v']})$ (resp.\ $(x', \Varvol[\{x, v\}] - \Vardist^{\Varort[x']}[x'] - 1, \OL{\Varort[x']})$) is added to $\Varchecklist$ (lines 13--14).
Finally, if $\theta = \OL{\Varort[v]}$ holds (i.e., $\theta$ is the unorthodox parity of $v$), the algorithm contracts the edge $\{v, \Varparent[v]\}$ (lines 16--18).
Then delete all the self-loops incident to the root $v'$ of the contracted subtree (lines 19--22), which correspond to the edges incident to a
descendant of $v'$ but not incoming edges of $T^r(v')$.
By the contraction, $v'$ obtains a new incident edge. The entry corresponding to the minimum-level incident edge $\{y, z\}$ of $v'$ is added to $\Varchecklist$ unless $\Vardist^{\OL{\Varort[v']}}[v']$ has already been decided (line 23).

We proceed to the correctness proof of the algorithm \textsc{ComputeDist}. 
For $\theta \in \{\Odd, \Even\}$ and $v \in V(G)$ such that $\dist^{\theta}(v) = r$ holds, 
we say that $(v, \theta)$ is \emph{precisely fixed} if $(v, r', \theta)$ with $r' < r$ is
never added to $\Varchecklist$ and $(v, r, \theta)$ is added to $\Varchecklist$, where $\Varparent(v) = v'$ for some $\{v, v'\} \in E(G)$ with $\rho(\{v, v'\}) = \OL{\Op(v)}$ and $\dist^{\OL{\Op(v)}}(v') = \dist^{\Op}(v) - 1$ in addition if $\theta = \Op(v)$,
until the beginning of round $r$. We also say that \emph{everything is precisely fixed} at round 
$r$ if $(v, \theta)$ for all $v \in V(G)$ and $\theta \in \{\Odd, \Even\}$ with $\dist^\theta(v) \leq r$ 
are precisely fixed. We first present two fundamental propositions, which are obvious from the algorithm.

\begin{proposition} \label{prop:preciselyfixed}
Suppose that everything is precisely fixed at round $r$.
Then for any $(v, \theta)$ such that $\dist^{\theta}(v) \leq r$, $\mbox{\rm\Vardist}^{\theta}(v) = \dist^{\theta}(v)$ holds at the beginning of round $r + 1$. 
\end{proposition}

\begin{proposition} \label{prop:treeexistence}
Suppose that everything is precisely fixed at round $r$. Then there exists an ABT $T$ containing $T^r$ as its subgraph.
\end{proposition}

The following lemmas are key technical properties.
\begin{lemma} \label{lma:distComp}
Let $T$ be any ABT, $t$ be any vertex in $V(T)$ such that 
$\dist^{\Nop}(t) < \infty$, and $\Minedge = \{y, z\}$ be any MIE of $T(t)$. Then the following two properties
hold:
\begin{itemize}
    \item[(P1)] For the path $P$ from $z$ to $t$ along $T$, any vertex $v \in V(P) \setminus \{t\}$ satisfies $\dist^{\Nop}(v) < \dist^{\Nop}(t)$.
    \item[(P2)] $\hlevel(\Minedge) < \dist^{\Nop}(t)$.
\end{itemize}    
\end{lemma}

\begin{proof}
The property (P1) is proved as follows:
\begin{align*}
\dist^{\Nop}(v) &= \vlevel(\Minedge) - \dist^{\Op}(v) & \text{(Corollary~\ref{corol:IKY24})} \\
&< \vlevel(\Minedge) - \dist^{\Op}(t) & \text{(Lemma~\ref{lma:fundamental} (P1))} \\
&= \dist^{\Nop}(t). & \text{(Corollary~\ref{corol:IKY24})}
\end{align*}
For (P2), $\hlevel(\Minedge) < \dist^{\Nop}(t)$ holds by Lemma~\ref{lma:heightBound}.
\end{proof}

\begin{lemma} \label{lma:orthodxfix}
Let $t$ be any vertex, and $r = \dist^{\Op}(t)$. 
Suppose that everything is precisely fixed at round $r - 1$.
Then $(t, \Op(t))$ is precisely fixed.
\end{lemma}

\begin{proof}
We first show that $(t, r, \Op(t))$ is added to $\Varchecklist$ at round $r - 1$.
By the assumption of the lemma, we have $\Varparent(t) = t'$ for some $\{t, t'\} \in E(G)$ with $\rho(\{t, t'\}) = \OL{\Op(t)}$ and $\dist^{\OL{\Op(t)}}(t') = r - 1$.
This can occur only in line 8 when a triple $(t', r - 1, \OL{\Op(t)}) \in \Varchecklist$ is processed at round $r - 1$.
Thus, $(t, r, \Op(t))$ is added to $\Varchecklist$ at that time. 

Next, we show that $(t, r', \Op(t))$ for $r' < r$ is never added to $\Varchecklist$.
If $(t, r', \Op(t))$ is added to $\Varchecklist$, it can occur only in line 8. 
Let $(t', r'', \theta')$ be the triple such that $(t, r', \Op(t))$ is added when processing it.
Then $\theta' = \OL{\Op(t)}$, $r' = r'' + 1$, $\{t, t'\} \in E(G)$, and $\rho(\{t, t'\}) = \OL{\Op(t)}$. By the assumption of the lemma and
Proposition~\ref{prop:preciselyfixed}, we have $\Vardist^{\OL{\Op(t)}}(t') = \dist^{\OL{\Op(t)}}(t') = r''$. 
Then $t'$ is a parent of $t$ in $T^{r}$, and thus we obtain $\dist^{\Op(t)}(t) = \dist^{\OL{\Op(t)}}(t')  + 1 = r' 
< r$. It is a contradiction.
\end{proof}

\begin{lemma} \label{lma:unorthodoxfix}
Let $t$ be any vertex, and $r = \dist^{\Nop}(t) < \infty$. 
Suppose that everything is precisely fixed at round $r - 1$.
Then $(t, \Nop(t))$ is precisely fixed.
\end{lemma}

\begin{proof}
We first show that no $(t, x, \Nop(t))$ of $x < r$ is added to $\Varchecklist$ at round $\hat{r} < r$. 
By the definition of the algorithm, $(t, x, \Nop(t))$ is added to $\Varchecklist$ only when an incoming edge $e'$
of $T^{\hat{r}}(t)$ (and thus that of $T(t)$) with volume $\Varvol[e'] = x + \Vardist^{\Op(t)}(t) + 1$ becomes incident to 
$t$ in $\tilde{T}^{\hat{r}}$. By assumption that everything is precisely fixed at round $\hat{r} \leq r - 1$, 
if such an edge exists and $(t, x, \beta(t))$ is added in round $\hat{r}$, then $\Vardist^{\rho(e')}(y') = \dist^{\rho(e')}(y')$, $\Vardist^{\rho(e')}(z') = \dist^{\rho(e')}(z')$, and $\Vardist^{\Op(t)}(t) = \dist^{\Op}(t)$ hold.
Then we obtain $\vlevel(e') = \Varvol[e'] < r + \Vardist^{\Op(t)}(t) + 1 = \vlevel(\Minedge)$. It is a contradiction. 

Next, we show that $(t, r, \Nop(t))$ is added to $\Varchecklist$ at round $r-1$ or earlier. 
Let $T$ be any ABT containing $T^r$ as a subgraph (such a $T$ necessarily exists by 
Proposition~\ref{prop:treeexistence}), and $\Minedge = \{y, z\}$ be any MIE of $T(t)$. 
Let $P$ be the path from $z$ to $t$ in $T$, and %
let $$r' = \max\left\{\max_{v \in V(P) \setminus \{t\}} \dist^{\Nop}(v),\, \dist^{\rho(\Minedge)}(y),\, \dist^{\rho(\Minedge)}(z) \right\}.$$ %
By Lemma~\ref{lma:distComp} (P1) and (P2), $r' \leq r - 1$ holds. 
By Proposition~\ref{prop:preciselyfixed} and the assumption of the lemma, it implies 
$\Varvol[\Minedge] = \vlevel(\Minedge)$ and $\Vardist^{\Op(t)}(t) = \dist^{\Op}(t)$ at the beginning of round $r$. By 
Corollary~\ref{corol:IKY24}, we have $\Varvol[\Minedge] - \Vardist^{\Op(t)}(t) - 1 = \dist^{\Nop}(t) = r$.
By Lemma~\ref{lma:distComp} (P1) and the assumption of the lemma, every vertex $v \in V(P)\setminus \{t\}$ satisfies $\Vardist^{\Nop(v)}(v) = \dist^{\Nop}(v) \leq r'$. Hence they are contracted at round $r'-1$ or earlier, and thus $\Minedge$ is an edge incident to $t$ until the beginning of round $r'$.
If $r' = \max_{v \in V(P) \setminus \{t\}} \dist^{\Nop}(v)$ holds, $(t, r, \Nop(t))$ is added to $\Varchecklist$ at line 23 when contracting the last vertex in $V(P) \setminus \{t\}$.
If $r' = \dist^{\rho(\Minedge)}(y)$ (resp.\ $r' = \dist^{\rho(\Minedge)}(z)$) holds, it is added at line 13 or 14 when $(y, r', \rho(\Minedge))$ (resp.\ $(z, r', \rho(\Minedge))$) is processed.
\end{proof}

\begin{theorem} \label{thm:constABT}
The algorithm \textsc{ComputeDist} satisfies $\mbox{\rm\Vardist}^{\theta}(v) = \dist^{\theta}(v)$ for all $v \in V(G)$ and
$\theta \in \{\Even, \Odd\}$ at the end of the algorithm, and runs in $O(m\log n)$ time.
\end{theorem}

\begin{proof}
To show the correctness, it suffices to show that everything is precisely fixed at all round $r$. We prove
it by the induction on $r$. (Basis) For $r = 0$, we only need to care the vertex $f$. It is obvious from line 1 
of the algorithm. (Induction step) Suppose as the induction hypothesis that everything is precisely fixed at 
round $r - 1$, and consider the case of round $r$. By Lemmas~\ref{lma:orthodxfix} and \ref{lma:unorthodoxfix}, 
one can easily see that everything is precisely fixed at round $r$. 

Finally we show $O(m \log n)$ running-time bound. 
The for loop of lines 6--14 when processing $(v, r, \theta)$ takes the (amortized) cost 
of $O(\text{(the degree of $v$}) \cdot \log n)$. Since that part is executed only when $\Vardist^{\theta}(v) = \infty$,
it is executed at most twice for each vertex. Hence the total running time of that part is $O(m \log n)$. It also
implies that the number of tuples added to $\Varchecklist$ at lines 13 and 14 is $O(m)$, and the total 
number of edges managed by $\Varnontree$ is $O(m)$. Then the total running time of the while loop of lines 20--22
is $O(m \log n)$. The number of tuples added to $\Varchecklist$ at line 23 is equal to the number of edge contractions,
which is obviously $O(n)$. Hence the total number of tuples added to $\Varchecklist$ is $O(m)$, and thus the outer
while loops takes $O(m)$ iterations. Putting all together, we obtain the running time bound of $O(m \log n)$.
\end{proof}

\subsection{Sequential Algorithm for Step 2}

The DDFS algorithm we utilize is the same as the one proposed in~\cite{MV1980,vazirani2024theory}.
We explain here its outline for self-containedness. Let $H$ be the input DAG, and $\phi\colon V(H) \to \mathbb{N}$
be any function such that $(v_1, v_2) \in E(H)$ implies $\phi(v_1) > \phi(v_2)$. In our use, the value of 
$\dist^{\Op}(\cdot)$ can be used as the function of $\phi$. 
The DDFS on $H$ starting from vertices $y$ and $z$ executes two parallel DFSs with root $y$ and $z$
in a coordinated manner. We denote by $T_y$ and $T_z$ the constructed DFS trees respectively rooted by $y$ and $z$, 
and $h_y$ and $h_z$ be their search heads. The coordination strategy is summarized as follows:
\begin{itemize}
\item At each step, the algorithm proceeds the tree growth of $T_x$ for $x \in \{y, z\}$ such that $x = 
\argmax_{x \in \{y, z\}} \phi(h_x)$ holds. The tree growth process follows the standard DFS procedure:
The algorithm finds a neighbor of $h_x$ not occupied by any of $T_y$ and $T_z$. 
If it is found, $T_x$ is expanded. Otherwise, it performs backtrack. 

\item Suppose that $\phi(h_z) < \phi(h_y)$ holds. Then the construction of $T_y$ continues.
If $h_y$ failed to find a path to $U$, $h_y$ finally goes back to $y$. %
At this timing, it is hasty to conclude that $\{y, z\}$ does not admit any double path, because an $x$--$U$ path is found by the backtrack of $h_z$ and renouncing the vertex $v$ at which $h_z$ stays might provide a $y$--$U$ path.
Hence if $h_y$ goes back to $y$, the tree $T_y$ steals the vertex $v$, with backtracking of $h_z$.
Under the condition that we appropriately process an omissible set in each of the previous iterations so that for each remaining vertex $v'$ there exists at least one $v'$--$U$ path in $H$ disjoint from any double paths already found, $T_y$ necessarily has a neighbor of $v$; thus adding $v$ to $T_y$ is seen as an expansion of $T_y$ following the standard DFS procedure.
The intuitive explanation of this fact is as follows:
Let $B$ be the set of the immediate predecessors of $v$ in $H$.
Obviously, $T_z$ contains a vertex in $B$.
If $T_y$ does not contain any vertex in $B$, $T_z$ cannot grow up for containing $v$, because $h_y$ searches only the predecessors of $B$ in $H$ and thus $\phi(h_y) > \phi(h_z)$ holds as long as $h_z$ is within $B$.

After $v$ is taken by $T_y$, $h_z$ tries to find a path to $U$ not containing $v$.
If it fails and $h_z$ is backtracked to $z$, one can conclude that both $T_y$ and $T_z$ need to contain $v$ for their further expansion.
Hence $v$ is the bottleneck, and we shrink the associated omissible set $W = (V(T_y) \cup V(T_z)) \setminus \{v\}$ into $v$.
Otherwise, the DDFS search continues. In the subsequent execution, $v$ behaves as the root of $y$, for avoiding that $h_y$ backtracks the predecessor of $v$.
When both $h_y$ and $h_z$ reach $U$, a double path $(P, Q, y, z)$ is found.
Due to the fundamental property of the standard DFS, any vertex $v'$ backtracked does not have a $v'$--$U$ path not intersecting vertices already visited.
Hence $V(T_y) \cup V(T_z)$ should be included in the associated omissible set $W$.
Furthermore, while there exists a vertex $v' \in V(H) \setminus (U \cup W)$ such that for every edge $(v', v'')$ the head $v''$ is in $W$ or has been removed in the previous iterations, we add such a $v'$ to $W$; this procedure guarantees that each remaining vertex has at least one path to $U$ in $H$ after removing $W$ (cf.~\textsc{RECURSIVE REMOVE} in \cite{vazirani2024theory}).
After that, we remove $W$ from $H$.
\end{itemize}

The precise correctness of the DDFS algorithm above and its running time bound follow 
the prior literature~\cite{vazirani2024theory}. %
By the argument explained in Section~\ref{subsec:DDFS}, we obtain the following lemma.

\begin{lemma} \label{lma:constDoublePaths}
There exists an algorithm of outputting a maximal set of disjoint double paths for $(H, \Ccset)$, taking 
any alternating base DAG $H$ and the set of all critical edges of volume $2\ell + 5$. The algorithm runs
in $O(m)$ time.
\end{lemma}

\begin{proof}
First, we add all the critical edges of volume $2\ell + 5$ into a queue. The algorithm dequeues each edge $\{y, z\}$ 
one by one, and tries to find a double path $(P, Q, y, z)$ using the DDFS. If it succeeds, all the vertices 
in the associated omissible set $W$ are removed.  Otherwise, all the vertices in the omissible set $W$ are shrunk to 
the bottleneck vertex. Since this algorithm runs in $O(|E_W|)$ time, the total running time is $O(m)$.
The correctness of the algorithm follows the properties of DDFS. 
\end{proof}

\section{Approximate MCM in CONGEST and Semi-Streaming Settings}

This section presents how our framework is utilized for leading an efficient algorithm to compute an approximate maximum matching in distributed and semi-streaming settings.
Our aim is to obtain a matching of size at least $(1 - \epsilon)\mu(G)$.
For simplicity, we assume that $\epsilon$ is represented as $\epsilon = 2^{-x}$ for some integer $x > 0$.

\subsection{Models}

\paragraph{CONGEST Model}
The CONGEST model is one of the standard models in
designing distributed graph algorithms. The system is modeled by a simple connected undirected graph 
$G$ of $n$ vertices and $m$ edges. Vertices and edges are uniquely identified by $O(\log n)$-bit integer values, and 
the computation follows round-based synchronous message passing. In one round, each vertex $v$ sends and receives 
$O(\log n)$-bit messages through each of its incident edges, and executes local computation following its internal 
state, local random bits, and all the received messages. A vertex can send different messages to different neighbors 
in one round, and all the messages sent in a round are necessarily delivered to their 
destinations within the same round. 

\paragraph{Semi-Streaming Model}
The semi-streaming model is a popular one in the graph-stream algorithm, where the algorithm can use the space of the vertex-size
order (i.e., $\tilde{O}(n)$ space), and the edge set arrives as a data stream. Intuitively, the model can be seen as the
one which is sufficiently rich to store the output of the algorithm, but not enough to store whole graph information. Our study
considers its multi-path variant, where the algorithm can read the input stream multiple times. In this settings, the efficiency
of the algorithm is measured by the number of \emph{passes} (i.e., the number of times the input stream is read). The model our 
algorithm assumes on input streams is the weakest, where the order of edges in the input stream can be arbitrary, and even can 
be different among different passes.

\subsection{Subroutines} \label{subsec:subroutine}

We first present the CONGEST/semi-streaming implementations of our framework.

\begin{lemma} \label{lma:subroutines}
Let $(G, M)$ be the matching system, and $2\ell + 1$ be the length of shortest augmenting paths in 
$(G, M)$. Suppose that the information of $M$ is given to each vertex, i.e., each vertex knows the incident edge 
in $M$ (if it exists). There exist the following deterministic CONGEST algorithms:
\begin{itemize}
    \item Algorithm $\textsc{CcomputeDist}_\ell$: Each vertex $v$ outputs the value of $\dist^{\theta}(v)$ for $\theta \in \{\Even, \Odd\}$ 
    if $\dist{^\theta}(v) \leq 2\ell + 5$ holds. The algorithm runs in $O(\ell^2)$ rounds.
    \item Algorithm $\textsc{CcomputeMIE}_\ell$: Initially, each vertex $v$ knows its parent in a given ABT $T$ of height 
    at most $\ell + 3$, and it outputs an MIE of $T(v)$, which runs in $O(\ell)$ rounds.
    \item Algorithm $\textsc{CdoubleToAug}_\ell$: Initially, each vertex knows
    the incident edges in $E(\Qcal)$ of a given set $\Qcal$ of disjoint double paths, and the output of $\textsc{CcomputeDist}_\ell$. Each vertex 
    outputs the incident edges contained in the corresponding augmenting path set $\Augt(\Qcal)$. The algorithm runs in $O(\ell)$ rounds.
\end{itemize}
\end{lemma}

\begin{proof}
The algorithm $\textsc{CcomputeMIE}_\ell$ is given in~\cite{IKY24} (see Appendix B). For $\textsc{CdoubleToAug}$,
one can also use the CONGEST implementation of $\textsc{PathConstruction}$ shown in~\cite{IKY24} (which is
referred to as the algorithm $\textsf{EXTPATH}$ there). It runs in $O(\ell)$ rounds, and when it transforms a path 
$P$ lying in a subtree $T(u)$ ($u \in U$) to an alternating path, only the vertices in $T(u)$ join the execution
of the algorithm.
Hence it can be run in parallel for all $T(u)$, yielding a CONGEST implementation of $\textsc{DoubleToAug}$. 

The CONGEST implementation of $\textsc{ComputeDist}_\ell$ is obtained naturally from Algorithm~\ref{alg:constABT}.
One can regard Algorithm~\ref{alg:constABT} as a round-by-round message-passing algorithm. 
In the algorithm, each vertex $v$ manages the variables
associated to it. More precisely, it manages $\Vardist^{\theta}[v]$ for $\theta \in \{\Odd, \Even\}$  
and $\Varparent[v]$. In addition, the information of the union-find date structure $\Varcontraction$ is simply implemented 
with vertex labels. That is, each vertex $v$ stores the value of $\Varcontraction.\textsc{find}(v)$ into the corresponding 
local variable $\Varcontraction[v]$. At every round, the information of $\Vardist^{\theta}[v]$ for both 
$\theta \in \{\Odd, \Even\}$ is exchanged with its neighbors. 
It allows each vertex to identify the levels of incident edges, as well as the information on the edge set added to $\Varnontree$.
The addition to $\Varnontree$ and merging two $\Varnontree$ priority queues are not explicitly implemented. Instead, to 
access $\Varnontree[v]$, $v$ checks all the edges incident to the vertices $v'$ such that $\Varcontraction[v'] = \Varcontraction[v]$ holds.
The operation $\Varnontree[v].\textsc{find-min}$ is implemented by the standard aggregation over the subtree
induced by the vertices $v'$ satisfying $v' = \Varcontraction[v]$, which is implemented by $O(\ell)$ rounds. The operation  $\Varnontree[v].\textsc{delete-min}$ is implicitly implemented by omitting
the edges $\{y, z\}$ such that $\Varcontraction[y] = \Varcontraction[z]$ holds in the search process of $\textsc{find-min}$.

In the simulation, one can regard the queue $\Varchecklist$ as the buffer of 
in-transmission messages. If $(v, r, \theta) \in \Varchecklist$ holds, it implies that a message 
$(v, r, \theta)$ to $v$ is in-transmission, which will be delivered exactly at round $r$. Note that when 
$(v, r, \theta)$ is added in processing another entry $(v', r', \theta')$, the component induced by the vertices shrunk into $v$
and that by $v'$ is connected by an edge. Hence one can route the message from $v'$ to $v$ within $O(\ell)$ rounds. 
To compute the distance up to $2\ell + 5$, running $\textsc{ComputeDist}$ up to round $2\ell + 5$ suffices. 
As argued above, simulating one round of $\textsc{ComputeDist}$ needs $O(\ell)$ CONGEST rounds. Hence the total running time of
the algorithm is $O(\ell^2)$ rounds.
\end{proof}

\begin{lemma} \label{lma:Csubroutines}
Let $(G, M)$ be the matching system, and $2\ell + 1$ be the length of shortest augmenting paths in 
$(G, M)$. Suppose that the vertex set $V(G)$ and the current matching $M$ are always stored in the memory. There 
exist the following deterministic semi-streaming algorithms of using $O(n \log n)$ bits:
\begin{itemize}
    \item Algorithm $\textsc{ScomputeDist}_\ell$: This algorithm
    outputs the value of $\dist^{\theta}(v)$ for $\theta \in \{\Even, \Odd\}$ such that $\dist{^\theta}(v) \leq 2\ell + 5$ holds. 
    The algorithm takes $O(\ell^2)$ passes.
     \item Algorithm $\textsc{ScomputeMIE}_\ell$: Initially, any ABT $T$ of height at most $\ell + 3$ and the output of 
    $\textsc{ScomputeDist}_\ell$ are stored in the memory. The algorithm outputs an MIE of $T(v)$ for all $v \in V(T)$ with one pass.
    \item Algorithm $\textsc{SdoubleToAug}_\ell$: Initially, a set $\Qcal$ of disjoint double paths and the output of 
    $\textsc{ScomputeDist}_\ell$ are stored in the memory. The algorithm outputs $\Augt(\Qcal))$ of the corresponding augmenting 
    paths with $O(1)$ passes.

\end{itemize}
\end{lemma}

\begin{proof}
We first consider the implementation of $\textsc{ScomputeMIE}$. The algorithm associates each vertex $v$ with a variable
$\textsf{MIE}_v$ storing the MIE information of $T(v)$. When an edge $e = \{y, z\}$ arrives, the algorithm checks if 
it is a non-tree edge or not. If it is a non-tree edge, for vertices $v$ such that $e$ becomes an incoming edge of 
$T(v)$, compare $\level(e)$ with $\level(\textsf{MIE}_v)$. Note that the $\level(e)$ is computed only from the stored information
(i.e., the output of $\textsc{ScomputeDist}_{\ell}$). If $\level(e)$ is smaller, $\textsf{MIE}_v$ is updated with $e$.
Obviously this algorithm computes MIEs for all $T(v)$, and takes only one pass. Next we consider the implementation of
$\textsc{SdoubleToAug}$. First, it computes an ABT $T$ respecting all double paths in $\Qcal$.
More precisely, when an edge $e = \{y, z\}$ arrives, the algorithm includes it into the constructed tree if $e \in E(\Qcal)$ holds or $e \in \paredgeset(z)$ and the parent of $z$ is not decided yet; note that if some $e' \in \paredgeset(z)$ has already been added and $e \in E(\Qcal) \cap \paredgeset(z)$ arrives later, we remove $e'$. 
This process takes one pass. After the construction of $T$, it identifies all MIEs for $T$ 
using $\textsc{ScomputeMIE}$. Since $\textsc{PathConstruction}$ is implemented only by the information of $T$
and its MIEs, the transformation from $\Qcal$ to $\Augt(\Qcal)$ is attained only by the stored information, i.e., it takes
no pass. In total, $O(1)$ passes suffice. 

Similar to the CONGEST case, the implementation of $\textsc{ScomputeDist}_\ell$ also follows Algorithm~\ref{alg:constABT}.
Basically, one round is implemented by taking one pass. The memory manages the information of $\Vardist^{\theta}[v]$ and for 
$\theta \in \{\Odd, \Even\}$, $\Varparent[v]$, and $\Varcontraction$. The operation of $\Varnontree[v].\textsc{find-min}$ 
(and $\Varnontree[v].\textsc{delete-min}$) is implemented by the brute-force check of all edges incident to vertices $v'$ satisfying
$\Varcontraction[v'] = v$. More precisely, each vertex $v$ invoking $\Varnontree[v].\textsc{find-min}$ manages its own output 
candidate. When $e = \{y, z\}$ arrives, the algorithm checks the ancestors $v$ 
of $y$ and $z$ currently invoking $\Varnontree[v].\textsc{find-min}$ in $T$ up to the lowest common ancestor of $y$ and $z$.
If $\level(e) < \level(e')$ holds for the current output candidate $e'$ of $v$, it is updated (recall that $\level(e)$ is
computed from the stored information). After the one-pass, the output candidates precisely hold the desired outputs.
The algorithm also saves the size of $\Varchecklist$ with omitting unnecessary entries. In the algorithm $\textsc{ComputeDist}$, 
once $\dist^{\theta}(v) < \infty$ holds, the remaining tuple $(v, r', \theta)$ becomes unnecessary by line 4 of Algorithm~\ref{alg:constABT}. 
Hence for each $v \in V(G)$ and $\theta \in \{\Odd, \Even\}$ it suffices to store only $(v, r, \theta) \in \Varnontree$ with the smallest 
$r$, which requires only $O(n \log n)$-bit space. 

The algorithm simulates one round of $\textsc{ComputeDist}$ with two passes.
The first pass is spent for processing lines 4--14,
particularly checking all neighbors at line 6.
The second pass is spent for processing $\textsc{find-min}$ operations at lines 19--22.
Hence the total number of passes is bounded by $O(\ell)$.
\end{proof}

\subsection{Algorithm $\textsc{Amplifier}_\alpha$}\label{subsec:Amplifier}

In this section, we prove our main result on approximation (Theorem~\ref{thm:mainApproximate}), assuming that there exists an algorithm $\textsc{AugToHit}$ satisfying Lemma~\ref{lma:AugAndHit}, which is proved in Section~\ref{subsec:AugAndHit}.

\begin{lemma} \label{lma:AugAndHit}
Let $(G, M)$ be any matching system. There exists an algorithm $\textsc{AugAndHit}_{\ell}(G, M)$ which outputs a set $\Qcal$
of disjoint shortest augmenting paths in $(G, M)$ and a hitting set $B$ of size at most $16|\Qcal|(\ell + 1)^2 + \frac{|M|}{4(\ell + 1)}$ for $(G, M)$ if 
the length of shortest augmenting paths in $(G, M)$ is not larger than $2\ell + 1$ (otherwise, it outputs $\Qcal = \emptyset$). 
It is implemented in both CONGEST and semi-streaming models, which respectively attains the following properties:
\begin{itemize}
    \item In the CONGEST model, each vertex initially knows the edge of $M$ incident to itself (if any), and outputs the incident edges in $E(\Qcal)$ and the one-bit flag of indicating $v \in B$ or not. The running time is $O(\ell^2 \mbox{\rm\textsf{MM}}(n))$ rounds, where $\mbox{\rm\textsf{MM}}(n)$ is the running time
    of computing a maximal matching of graphs on $n$ vertices. The algorithm is deterministic
    except for the part of computing maximal matchings in CONGEST.
    \item In the semi-streaming model, the algorithm is deterministic, uses $O(n\log n)$-bit space, and run with $O(\ell^2)$ passes. Initially the information on $V(G)$ and $M$ is stored in the memory, and outputs the set $\Qcal$ and $B$. 
\end{itemize} 
\end{lemma}

The key building block of our algorithm is $\textsc{Amplifier}_{\alpha}$.
We recall that it takes a parameter $\alpha > 0$ (represented as a power of two) and any matching $M$ such that $|M| = (1 - \alpha')\Mmax(G)$ holds for some $\alpha' \leq \alpha$, and outputs a matching $M'$ of $G$ whose approximation factor is at least $\min\{(1 - \alpha' + \Theta(\alpha^2)), (1- \alpha/2)\}$ (Lemma~\ref{lma:amplifier}).
Iteratively applying $\textsc{Amplifier}_{\alpha}$ $O(\alpha^{-1})$ times, we obtain a $(1 - \alpha/2)$-approximate matching.
Starting with any maximal matching (which is a $1/2$-approximate matching) and calling that iteration process $O(\log \epsilon^{-1})$ times with $\alpha = 1/2, 1/4, 1/8, \dots 1/\epsilon^{-1}$, we finally obtain a $(1 - \epsilon)$-approximate matching.

The execution of $\textsc{Amplifier}_\alpha$ consists of repetition of $K = 4\alpha^{-1}$ \emph{phases} with modification of the input graph. Let $(G_i, M_i)$ be the matching system at the beginning of 
the $i$-th phase. The behavior of the $i$-th phase consists of the following three steps:
\begin{enumerate}
    \item \emph{Path finding}: Call $\textsc{AugAndHit}_K(G_i, M_i)$ and find a set of shortest augmenting paths $\Qcal_i$ and a hitting set $B_i$. 
    \item \emph{Length stretch}: For each vertex $v \in B_i$, if $v$ is not free, the algorithm subdivides the matching edge incident to $v$ into a length-three alternating path such that the middle edge is a non-matching edge. Otherwise, we rename $v$ into $v'$, and add a length-two alternating path to $v'$ where the edge incident to $v'$ is a matching edge and the endpoint of the path other than $v'$ becomes the new $v$. The resultant graph and matching is denoted by $(G_{i+1}, M'_i)$. We refer to this transformation as $\sigma_i$, i.e.,
    $\sigma_i(G_i, M_i) = (G_{i+1}, M'_i)$. In addition, we abuse $\sigma_i$ for transforming any alternating path $P$ in $(G_i, M_i)$. That is, if $P$ does not contain any vertex in $B_i$, $\sigma_i(P)$ returns $P$ itself. Otherwise, $\sigma_i(P)$ returns the transformed path in $(G_{i+1}, M'_i)$ obtained by appropriately subdividing matching edges in $P$. 
    \item \emph{Augmentation}: Augment $\sigma_i(Q)$ for all $Q \in \Qcal_i$. The matching after the augmentation becomes $M_{i+1}$. 
\end{enumerate}
After the $K$ phases, the algorithm finally executes the \emph{recovery} step. %
We first define the reverse transformation $\hat{\sigma}_i$. The transformation $\hat{\sigma}_i(G_{i+1}, M')$ returns the matching system 
$(G_i, M)$ constructively defined as follows: any path $v_1, e_1, v_2, e_2, v_3, e_3, v_4$ of length three in $G_{i+1}$ subdivided by 
$\sigma_i$ is restored back to a single edge $e$ in $G_i$. If $e_1$ and $e_3$ are non-matching edge in $M'$, $e$ becomes a non-matching edge in 
$(G_i, M)$. Otherwise (i.e., $e_1$ and $e_3$ are matching edges), $e$ becomes a matching edge in $(G_i, M)$. Note that 
it never happens that only one of $e_1$ and $e_3$ is a matching edge because the degrees of $v_2$ and $v_3$ are two and 
thus any augmenting path contains both $e_1$ and $e_3$ or neither $e_1$ nor $e_3$. In addition, any path $v_1, e_1, v_2, e_2, v_3$ of length two
added to a free vertex in $G_{i}$ with $\sigma_i$ is deleted.
Let us define $(\hat{G}_{K+1}, \hat{M}_{K+1}) = (G_{K+1}, M_{K+1})$ and $(\hat{G}_i, \hat{M}_i) = \hat{\sigma}_i(\hat{G}_{i+1}, \hat{M}_{i+1})$. 
Finally, the procedure $\textsc{Amplifier}_{\alpha}$ outputs $\hat{M}_1$ as a matching of $\hat{G}_1 = G_1$.

The following lemma claims the correctness of Algorithm $\textsc{Amplifier}_\alpha$.

\begin{lemma} \label{lma:amplifier}
There exists a constant $c > 0$ such that if $|M_1| = (1 - \alpha')\Mmax(G_1)$ holds for some $\alpha' \leq \alpha$,
then $|\hat{M}_1| \geq \min\{(1 - \alpha' + c\alpha^2), (1 - \alpha/2)\}\Mmax(G_1)$ holds.
\end{lemma}

\begin{proof}
Let $c = 1/648$.
Since the output matching size is not smaller than the input size, the lemma obviously holds if $\alpha' \leq \alpha/2$. Hence in the following argument we assume $\alpha/2 < \alpha' \leq \alpha$. 
Let $\InvMmax(G, M) = \Mmax(G) - |M|$. Given a matching system $(G, M)$ and any maximum matching $M^{\max}$ of $G$, the symmetric difference of $M^{\max}$ and $M$ induces a subgraph containing $\InvMmax(G, M)$ 
disjoint augmenting paths in $(G, M)$.
Since $\sigma_i$ (as a mapping over all paths) is a bijective mapping between the set of all augmenting paths in $(G_i, M_i)$ and that in $(G_{i+1}, M'_{i})$, we have 
$\InvMmax(G_{i+1}, M'_i) = \InvMmax(G_{i}, M_{i})$, which implies $\InvMmax(G_i, M_i) - |\Qcal_i| = \InvMmax(G_{i+1}, M_{i+1})$.
Similarly, we also have $\InvMmax(\hat{G}_i, \hat{M}_i) = \InvMmax(\hat{G}_{i+1}, 
\hat{M}_{i+1})$. Let $q = \sum_{1 \leq i \leq K} |\Qcal_i|$. By the definition, we obtain $\InvMmax(G_1, \hat{M}_1) = \InvMmax(G_1, M_1) 
- q$. If $q \geq c\alpha^{2}\Mmax(G_1)$ holds, we have 
\begin{align*}
\InvMmax(G_1, \hat{M}_1) = \InvMmax(G_1, M_1) - q \leq \left(\alpha' - c\alpha^{2} \right)\Mmax(G_1).
\end{align*}
It implies $\hat{M}_1$ is a $(1 - \alpha' + c\alpha^2)$-approximate matching of $G_1$.
Hence we assume $q < c\alpha^2\Mmax(G_1)$.

Due to the hitting set property, any shortest augmenting path $P$ in $(G_i, M_i)$ contains a vertex 
$v$ in $B_i$. If $v$ is not free, $P$ must contain the matching edge incident to $v$, which is subdivided into a length-three alternating path in $G_{i+1}$. Otherwise, $\sigma_i(P)$ must contain a length-two alternating path added to $v$. %
In any case, the length of $\sigma_i(P)$ increases by at least two. Then $(G_{K+1}, M_{K+1})$ does not admit an augmenting path of length at most $2K - 1 = 8\alpha^{-1} - 1$, and thus $M_{K+1}$ is a $(1 - \alpha/4)$-approximate matching of $G_{K+1}$ by the Hopcroft--Karp analysis.
That is, $\InvMmax(G_{K+1}, M_{K+1}) \leq \alpha\Mmax(G_{K+1}) / 4$. If $\Mmax(G_{K+1}) \leq 2\Mmax(G_1)$ holds, we have $\InvMmax(\hat{G}_{1}, \hat{M}_1) = \InvMmax(G_{K+1}, M_{K+1}) \leq 
\alpha\Mmax(G_1) / 2$, and thus $\hat{M}_1$ is a $(1 - \alpha/2)$-approximate matching of $\hat{G}_1 = G_1$.
Hence the remaining issue is to show $\Mmax(G_{K+1}) \leq 2\Mmax(G_1)$.

We prove $\Mmax(G_{j}) \leq 2\Mmax(G_1)$ for any $1 \leq j \leq K + 1$
by the induction on $j$. Since the base case is obvious, we focus on the inductive step. Suppose as the 
induction hypothsis that the statement above holds for all $1 \leq i \leq j \leq K$ and consider the case of $j+1$.
We have
\begin{align*}
    \Mmax(G_{j+1}) &= \Mmax(G_1) + \sum_{1 \leq i \leq j} |B_i| \\
                 &\leq \Mmax(G_1) +  \sum_{1 \leq i \leq j} \left( 16 (K+1)^2 |\Qcal_i| +  \frac{|M_i|}{4(K+1)} \right) & (\text{Lemma~\ref{lma:AugAndHit}})\\
                 &\leq \Mmax(G_1) + 16 (K+1)^2 \cdot q + j\cdot\frac{\Mmax(G_1)}{2(K+1)} & (|M_i| \leq \Mmax(G_i) \leq 2\Mmax(G_1)) \\
                 &< \Mmax(G_1) + 16 (4\alpha^{-1} + 1)^2 \cdot c\alpha^2\Mmax(G_1) + \frac{\Mmax(G_1)}{2} & (K = 4\alpha^{-1},\ j \le K) \\
                 &= \Mmax(G_1) + 256c\left(1 + \frac{\alpha}{4}\right)^2\Mmax(G_1) + \frac{\Mmax(G_1)}{2} \\
                 &\leq \Mmax(G_1) + 324c\Mmax(G_1) + \frac{\Mmax(G_1)}{2} & \left(\alpha \le \frac{1}{2}\right)\\
                 &= \Mmax(G_1) + \frac{\Mmax(G_1)}{2} + \frac{\Mmax(G_1)}{2} & \left(c = \frac{1}{648}\right) \\
                 &\leq 2\Mmax(G_1).
\end{align*}
The Lemma is proved.
\end{proof}

We are now ready to prove Theorem~\ref{thm:mainApproximate}.

\let\temp\thetheorem
\renewcommand{\thetheorem}{\ref*{thm:mainApproximate}}
\begin{theorem}
There exists:
\begin{itemize}
\item a $(1 - \epsilon)$-approximation maximum matching algorithm for a given graph $G$ which runs in the CONGEST model
with $O(\epsilon^{-4} \mbox{\rm\textsf{MM}}(\epsilon^{-2}n))$ rounds, where $\mbox{\rm\textsf{MM}}(N)$ means the time complexity of computing
a maximal matching in graphs on $N$ vertices;
the algorithm is deterministic except for the part of computing maximal matchings.
\item a deterministic $(1 - \epsilon)$-approximation maximum matching algorithm for a given graph $G$ which runs in the semi-streaming model
with $O(\epsilon^{-4})$ passes and $O(n \log n)$-bit memory. 
\end{itemize}
\end{theorem}
\let\thetheorem\temp
\addtocounter{theorem}{-1}

\begin{proof}
We first focus on the CONGEST implementation. Obviously, the length-stretch step and the augmentation step are locally processed. 
Note that the simulation of $G_i$ is straightforward: 
each subdivided path is managed by an arbitrary one of its endpoints without overhead. For the first
invocation of $\textsc{Amplifier}_\alpha$ with $\alpha = 1/2$, we need give any $1/2$-approximate matching as its input. Since any maximal matching is a $1/2$-approximate solution, it can be constructed with $\textsf{MM}(n)$ rounds. Hence the total running time of the algorithm
is dominated by the cost for running $\textsc{AugAndHit}_K$. Since $K = O(\alpha^{-1}) = O(\epsilon^{-1})$ holds and the total number 
of augmenting paths found in the execution of each phase never exceeds $n$, for any $(G_i, M_i)$ appearing in the run of 
$\textsc{Amplifier}_{\alpha}$, $|V(G_i)| \leq \epsilon^{-2}n$ by Lemma~\ref{lma:AugAndHit}. It implies that any invocation of $\textsc{AugAndHit}_K(G_i, M_i)$ takes 
$O(K^2 \textsf{MM}(\epsilon^{-2}n))$ rounds by Lemma~\ref{lma:AugAndHit}, which is also the cost of one phase in the invocation of 
$\textsc{Amplifier}_{\alpha}$. Hence the running time of $\textsc{Amplifier}{\alpha}$ is 
$O(K^3 \textsf{MM}(\epsilon^{-2}n)) = O(\alpha^{-3} \textsf{MM}(\epsilon^{-2}n))$  rounds. The invocation of 
$\textsc{Amplifier}_{\alpha}$ with $\alpha = 1/2^j$ is repeated $O(\alpha^{-1})$ times.
Summing up the cost for all invocations, the total running time is bounded as follows:
\begin{align*}
    \sum_{1 \leq j \leq \log \epsilon^{-1}} O\left(\left(\frac{1}{2^j}\right)^{-4} \mbox{\rm\textsf{MM}}(\epsilon^{-2}n)\right) 
    = O(\epsilon^{-4} \mbox{\rm\textsf{MM}}(\epsilon^{-2}n)).
\end{align*}

Next, we consider the semi-streaming implementation. The algorithm keeps the information of $(V(G_i), M_i)$ for the current matching system $(G_i, M_i)$. 
The pass complexity follows the round complexity of the CONGEST implementation, and can obtain the same bound of $O(\epsilon^{-4})$ passes. 
On the space complexity, the algorithm 
$\textsc{Amplifier}_{\alpha}$ modifies the input graph, and thus we might need $\omega(n \log n)$-bit space if $(V(G_i), M_i)$ is 
explicitly stored in the memory. However, if $q_j = \sum_{1 \leq i \leq j} |\Qcal_i| < \alpha^2 \Mmax(G_1)/128$ holds,
the argument in the proof of Lemma~\ref{lma:amplifier} guarantees $\Mmax(G_{j+1}) = O(\Mmax(G_1))$, which implies $|V(G_{j+1})| = O(n)$ and
thus $O(n \log n)$-bit space suffices. To care the case of $q_j > c\alpha^2 \Mmax(G_1)$, we additionally introduce the early-stopping mechanism
of omitting subsequent phases when $q_j > c\alpha^2 \Mmax(G_1)$ happens at the end of the phase $j$. The proof of Lemma~\ref{lma:amplifier}
ensures that the output of the algorithm attains the desired approximation factor even if such a mechanism is
installed. Consequently, $O(n \log n)$-bit memory is enough to implement whole algorithm.
\end{proof}

\subsection{Algorithm $\textsc{AugAndHit}$} \label{subsec:AugAndHit}
The goal of this section is to prove Lemma~\ref{lma:AugAndHit}.

Consider the matching system $(G, M)$ where the length of shortest augmenting paths 
is bounded by $2\ell + 1$ (and thus any double path contains at most $2(\ell + 1)$ vertices).
We denote the ABD for $(G, M)$ up to depth $\ell$ by $H$ and the set of all crossable and critical edges in $H$ by $\Ccset$. 

\paragraph{Parallel DFS}
The implementation of $\textsc{AugAndHit}(H, \Ccset)$ utilizes the parallel DFS technique by Goldberg, Plotkin, and Vaidya\cite{GPV93}.
We perform the parallel DFS algorithm rooted by $U$ (the set of all free vertices) in the reverse graph of $H$, i.e., it manages a disjoint set $\Tcal$ of (DFS) in-trees rooted by each vertex in $U$.
Each tree in $\Tcal$ rooted by $u \in U$ is referred to as $T_u$, and we denote the search head of $T_u$ by $h_u$.
Each vertex is assigned with one of the three states: $\Idle$, $\Active$, and $\Dead$. 
Initially, all vertices in $V(H) \setminus U$ have state $\Idle$, and vertices in 
$U$ are assigned with $\Active$. When the search head $h_u$ visits a vertex $v$, the state of $v$ is 
changed to $\Active$, and when $h_u$ performs backtrack at $v$, it is changed to $\Dead$. In one step of the algorithm, all the heads try to expand their own trees in parallel: Letting $X$ be the set of the locations of all search heads and $Y$ be all $\Idle$ vertices, construct the auxiliary bipartite graph of vertex set $X \cup Y$ where $v \in X$ and $v' \in Y$ are connected by an undirected edge if and only if $(v', v) \in E(H)$ holds.
Then the algorithm computes a maximal matching in it. If $h_u \in X$ finds a matching partner $y \in Y$, $h_u$ moves to $y$. Otherwise it performs backtrack. When $h_u$ performs backtrack at $u$, the search from $u$ terminates with changing its state to $\Dead$. The algorithm iterates this step $L = 16(\ell + 1)^2$ times. 

The key technical properties of this algorithm is stated as follows:

\begin{lemma} \label{lma:parallelDFS}
Let $D$, $A$, and $I$ be the sets of vertices with states $\Dead$, $\Active$, and $\Idle$, respectively, at the end of the algorithm. The following three
properties hold:
\begin{itemize}
\item[(P1)] Any path $P$ in $H$ ending at a vertex in $U$ satisfies either $V(P) \subseteq D$ or $V(P) \cap A \neq \emptyset$.
\item[(P2)] The size of $A$ is at most $\frac{|V(H) \setminus U|}{8(\ell+1)}$. 
\item[(P3)] For each search tree $T_u$, the size of $V(T_u)$ is at most $L = 16(\ell + 1)^2$. 
\end{itemize}
\end{lemma}

\begin{proof}
\textbf{Proof of (P1)}: The DFS guarantees that there exists no edge $(v, v') \in E(H)$ such that $v$ is $\Idle$ and $v'$ is $\Dead$, because if such an edge exists the search head visiting $v'$ also visits $v$ before the backtrack at $v'$.
Since all the vertices in $U$ can be $\Dead$ or $\Active$, this implies $V(P) \subseteq D$ or $V(P) \cap A \neq \emptyset$.

\textbf{Proof of (P2)}: We say that the search from $u$ is active at the $i$-th iteration if it does not terminate at the end of the 
$i$-th iteration. Any search active at the $i$-th iteration changes the state of one vertex in $V(H) \setminus U$ during the $i$-th iteration.
Since the states of vertices in $V(H) \setminus U$ can change at most $2|V(H) \setminus U|$ times in total, at most $2|V(H) \setminus U|/L$ searches are still active at the end of the algorithm (after $L = 16(\ell + 1)^2$ iterations).
In each search tree $T_u$, only the vertices in the path from $u$ to $h_u$ are $\Active$, whose number is bounded by $\ell + 1$, and thus
the total number of active vertices at the end of the algorithm is at most $\frac{|V(H) \setminus U|}{8(\ell + 1)}$.

\textbf{Proof of (P3)}: It is obvious because at most one node is added to each tree in one iteration.
\end{proof}

\paragraph{Main Algorithm}

We present the outline of the behavior of Algorithm $\textsc{AugAndHit}(G, M)$ more precisely:
\begin{itemize}
\item First, the algorithm constructs $H$ and $\Ccset$ and runs the above-stated parallel DFS from $U$ in the reverse of $H$.
It provides a set of in-trees $T_u$ of $H$ rooted by each vertex $u \in U$. We refer to the vertex set $R_u$ spanned by each in-tree $T_u$ as 
the \emph{region} of $u$. All $\Active$ vertices at the end of the parallel DFS are added to a hitting set $B$. 
\item A vertex $v \in R_u$ is called \emph{double reachable} if there exists a partial double path $(P, Q, y, z)$ to $v$ such that both $P$ and $Q$ lie in $R_u$, i.e., $P$ is a $y$--$u$ path, $Q$ is a $z$--$v$ path, $V(P) \cap V(Q) = \emptyset$, and $V(P) \cup V(Q) \subseteq V(R_u)$.
For each region $R_u$, the algorithm identifies the set $D_u \subseteq R_u$ of all double reachable vertices. 

\item We say that two different regions $R_{u_1}$ and $R_{u_2}$ are \emph{mergeable} if there exists an edge $(v, v') \in E(H)$ from $R_{u_1}$ and $R_{u_2}$ such that $v$ 
is double reachable, or an edge $\{v, v'\} \in \Ccset$ bridging $R_{u_1}$ and $R_{u_2}$.
Note that the merged region $R_{u_1} \cup R_{u_2}$ necessarily contains a double path. If two mergeable pair of regions $R_{u_1}$ and $R_{u_2}$ 
are actually merged, the algorithm finds a double path in the merged region, and stores it into $\Qcal'$. Then 
all vertices in the merged region is added to $B$. The pairs of regions to be merged are decided by parallel maximal matching: First, 
the algorithm constructs the auxiliary graph where the vertex set is the set of all regions and two vertices are connected by an edge if and only 
if they are mergeable. Computing a maximal matching in that graph, all matched pairs are merged. Finally, the algorithm outputs $\Qcal = 
\{\Augt(Q) \mid Q \in \Qcal'\}$ and $B$.
\end{itemize}

The first step is easily implemented by both CONGEST and semi-streaming settings.
The construction of $H$ and $\Ccset$ follows the algorithms of Lemmas~\ref{lma:subroutines} and \ref{lma:Csubroutines}. In the parallel DFS, 
each vertex $v \in V(H)$ manages its state and the parent of the DFS tree it joins.
The update of the search head locations is implemented by a single invocation of any maximal matching algorithm.
Thus, the time complexity of the first step in CONGEST is $O(\ell^2 \textsc{MM}(n))$ rounds, and the pass complexity is $O(\ell^2)$ in the semi-streaming model.
In the second step, our algorithm first executes $\textsc{ComputeDist}_\ell$ in the subgraph $G[R_u]$ of $G$ 
induced by each region $R_u$ (with adding the super vertex $f$).
Let $\Ical = (G, M)$, $M[R_u] = M \cap E(G[R_u])$, and $\Ical[R_u] = (G[R_u], M[R_u])$. The following lemma holds:

\begin{lemma} \label{lma:IdentificationDouble}
Let $v \in R_u$ be any vertex with an outgoing edge $(v, v') \in E(H)$ reaching to another region. 
Then $v$ is double reachable if and only if $\dist^{\Even}_{\Ical}(v) = \dist^{\Even}_{\Ical[R_u]}(v)$ holds.
\end{lemma}

\begin{proof}
\textbf{Proof of direction $\Rightarrow$}: Since $v$ admits a partial double path within $R_u$, there exists a 
shortest $\Even$-alternating path $X$ from $f$ to $v$ for $\Ical$ which lying in $V(T_u)$ by Lemma~\ref{lma:partialreconstruction}. 
It is also an $\Even$-alternating path for $\Ical[R_u]$, and thus we have $\dist^{\Even}_{\Ical}(v) = \dist^{\Even}_{\Ical[R_u]}(v)$. 

\textbf{Proof of direction $\Leftarrow$}: 
Let $X$ be any shortest $\Even$-alternating path from $u$ to $v$ in $G[R_u]$. By Lemma~\ref{lma:uniqueCrossableEdge3}, that path is transformed into a partial double path lying in $R_u$. That is, $v$ is double reachable.
\end{proof}

The lemma above suffices to implement the second step. The main task of the third step is to find a 
shortest augmenting path in the merged region $R_{u_1} \cup R_{u_2}$. This is regarded as the task of 
finding a single augmenting path between $u_1$ and $u_2$ in the subgraph induced by that region. 
It can be done with $O(\ell^2)$ rounds in CONGEST, and $O(\ell^2)$ passes in the
semi-streaming model using our framework and the subroutines in Section~\ref{subsec:subroutine}. 
Hence the total running time is 
dominated by the first step, i.e., it takes $O(\ell^2 \textsf{MM}(n))$ rounds in CONGEST, and takes 
$O(\ell^2)$ passes in the semi-streaming model.

The correctness of the algorithm (Lemma~\ref{lma:AugAndHit}) is shown by the following key technical lemma.

\begin{lemma} \label{lma:correctnessAugAndHit}
The set $B \subseteq V(H)$ at the end of the algorithm is a hitting set for all shortest augmenting paths in $H$, and its size is at most 
$16|\Qcal|(\ell + 1)^2 + \frac{|M|}{4(\ell + 1)}$.
\end{lemma}

\begin{proof}
By Lemma~\ref{lma:parallelDFS} (P2) and (P3), the size of $B$ is bounded by $16|\Qcal|(\ell + 1)^2 + \frac{|V(H) \setminus U|}{8(\ell+1)}$. Since $2|M| = |V(H) \setminus U|$ holds, we obtain the stated bound.
Hence it suffices to show that $B$ is certainly a hitting set for all shortest augmenting paths.

We first show that $B$ is a hitting set for all double paths.
Suppose for contradiction that a double path $(P, Q, y, z)$ in $H$ does not intersect $B$.
By Lemma~\ref{lma:parallelDFS} (P1), both $P$ and $Q$ lie in the subgraph induced by $\Dead$ vertices.
If $y$ and $z$ respectively belong to different regions 
$R_{u_1}$ and $R_{u_2}$, at least one of $R_{u_1}$ and $R_{u_2}$ is merged with some region (otherwise, $R_{u_1}$ and $R_{u_2}$ must be
merged into one with edge $\{y, z\}$), which implies $y \in B$ or $z \in B$. If both $y$ and $z$ belong to a common region $R_u$, we 
consider the following two cases:

\textbf{(Case 1)} Either $P$ or $Q$ is contained in $R_u$: Without loss of generality, we assume $V(P) \subseteq R_u$. 
Then $P$ terminates with $u$, and $Q$ necessarily goes out from $R_u$. Let $(v_q, v'_q)$ be the first edge in $Q$ which goes out from 
$R_u$. Due to the existence of paths $P$ (from $y$ to $u$) and $Q[z, v_q]$ (from $z$ to $v_q$), one can conclude
that $v_q$ is double reachable.
Hence at least one of $R_u$ and the region containing $v'_q$ must be merged, which implies $Q$ intersects $B$.
It is a contradiction.

\textbf{(Case 2)} None of $P$ and $Q$ is contained in $R_u$: Let $(v_p, v'_p)$ (resp.\ $(v_q, v'_q)$) be the
first edge in $P$ (resp.\ $Q$) going out from $R_u$.
If either $v_p$ or $v_q$ is double reachable, one can lead a contradiction
with the same argument as Case 1. Otherwise, any path $P'$ from $v_p$ to $u$ intersects $Q[z, v_q]$, 
because if it does not intersect, $v_q$ must be double reachable due to the existence of disjoint paths $P[y, v_p] \circ P'$ and 
$Q[z, v_q]$. Similarly, any path $Q'$ from $v_q$ to $u$ intersects $P[y, v_p]$. Then, for any vertex
$x_1 \in V(P') \cap V(Q[z, v_q])$ and $x_2 \in V(Q') \cap V(P[y, v_p])$, $P'[v_p, x_1] \circ Q[x_1, v_q] \circ Q'[v_q, x_2] \circ P[x_2, v_p]$
forms a cycle, contradicting that $H$ is a DAG.

Finally, we show that $B$ is a hitting set of all shortest augmenting paths.
Suppose for contradiction that there exists a shortest augmenting path $X$ 
from $u_1$ to $u_2$ with edge $\{y, z\} \in \Ccset$.
From the construction, $B$ consists of in-trees rooted by vertices in $U$ (some of them may be paths formed by active vertices), which are all mutually disjoint.
By Lemma~\ref{lma:hittingset}, there exists two disjoint paths in $H$ respectively from $y$ to $u_1$
and from $z$ to $u_2$ which do not intersect $B$. It contradicts the fact that $B$ is a hitting set for all double paths. 
The lemma is proved.
\end{proof}

\section{Concluding Remarks}

In this paper, we presented a new structure theorem on shortest alternating paths in general graphs, which bypasses
most of the complication caused by the blossom argument. A key technical ingredient is the notion of alternating
base trees~\cite{KI22,IKY24}. Our new structure theorem extends it so that the framework can address 
the construction of a maximal set of disjoint shortest augmenting paths, whereas the original theorem focuses on
finding a single (shortest) augmenting path in general graphs to obtain efficient distributed message-passing 
algorithms for MCM in general graphs. It includes a short and concise proof of the original theorem in~\cite{IKY24}. Following our new structure theorem, we presented a new framework of finding a maximal set of disjoint shortest augmenting paths, which yields a new MCM algorithm under the standard sequential computation model.
The proposed algorithm is slightly slower but more implementable and much easier to confirm 
its correctness than the currently known fastest algorithm (i.e., the MV algorithm). We also presented two
$(1 - \epsilon)$-approximation matching algorithms in the CONGEST and semi-streaming models, which substantially
improve the best known upper bounds for those models.

The authors believe that our framework can yield new improved MCM algorithms in more various 
settings such as PRAM, MPC, dynamic algorithms, weighted graphs, and so on. The application and/or extension of
our framework for those scenarios is a natural open question to be addressed. Another intriguing open 
question is to refine the framework for bringing more shorter and concise algorithms and proofs. 
Obviously, it is also a highly challenging open problem if one can beat the $O(m\sqrt{n})$ time bound 
for sequential MCM computation in general graphs resorting to our framework.

\subsection*{Acknowledgment} 
Taisuke Izumi was supported by JSPS KAKENHI Grant Number 23K24825.
Naoki Kitamura was supported by JSPS KAKENHI Grant Number 23K16838.
Yutaro Yamaguchi was supported by JSPS KAKENHI Grant Number 25H01114.
This work was also supported by JST CRONOS Japan Grant Number JPMJCS24K2 and by JST ASPIRE Japan Grant Number JPMJAP2302.

\bibliographystyle{plain}
\bibliography{reference}

\begin{thebibliography}{10}

\bibitem{AK20}
Mohamad Ahmadi and Fabian Kuhn.
\newblock Distributed maximum matching verification in {CONGEST}.
\newblock In {\em Proceedings of the 34th International Symposium on
  Distributed Computing (DISC)}, pages 37:1--37:18, 2020.

\bibitem{AKO18}
Mohamad Ahmadi, Fabian Kuhn, and Rotem Oshman.
\newblock Distributed approximate maximum matching in the {CONGEST} model.
\newblock In {\em Proceedings of the 32nd International Symposium on
  Distributed Computing (DISC)}, pages 6:1--6:17, 2018.

\bibitem{AS11}
Kook~Jin Ahn and Sudipto Guha.
\newblock Laminar families and metric embeddings: Non-bipartite maximum
  matching problem in the semi-streaming model.
\newblock {\em arXiv preprint arXiv:1104.4058}, 2011.

\bibitem{AG13}
Kook~Jin Ahn and Sudipto Guha.
\newblock Linear programming in the semi-streaming model with application to
  the maximum matching problem.
\newblock {\em Information and Computation}, 222:59--79, 2013.

\bibitem{alman2025more}
Josh Alman, Ran Duan, Virginia~Vassilevska Williams, Yinzhan Xu, Zixuan Xu, and
  Renfei Zhou.
\newblock More asymmetry yields faster matrix multiplication.
\newblock In {\em Proceedings of the 2025 Annual ACM-SIAM Symposium on Discrete
  Algorithms (SODA)}, pages 2005--2039. SIAM, 2025.

\bibitem{ABI86}
Noga Alon, L\'{a}szl\'{o} Babai, and Alon Itai.
\newblock A fast and simple randomized parallel algorithm for the maximal
  independent set problem.
\newblock {\em Journal of Algorithms}, 7(4):567--583, 1986.

\bibitem{Assadi21}
Sepehr Assadi.
\newblock A two-pass (conditional) lower bound for semi-streaming maximum
  matching.
\newblock In {\em the 2022 Annual ACM-SIAM Symposium on Discrete Algorithms
  (SODA)}, pages 708--742, 2021.

\bibitem{Assadi25}
Sepehr Assadi.
\newblock A simple $(1-\epsilon)$-approximation semi-streaming algorithm for
  maximum (weighted) matching.
\newblock {\em TheoretiCS}, Volume 4, 2025.

\bibitem{ALT21}
Sepehr Assadi, S.~Cliff Liu, and Robert~E. Tarjan.
\newblock An auction algorithm for bipartite matching in streaming and
  massively parallel computation models.
\newblock In {\em 4th Symposium on Simplicity on Algorithms (SOSA)}, pages
  165--171, 2021.

\bibitem{AS23}
Sepehr Assadi and Janani Sundaresan.
\newblock Hidden permutations to the rescue: Multi-pass streaming lower bounds
  for approximate matchings.
\newblock In {\em IEEE 64th Annual Symposium on Foundations of Computer Science
  (FOCS)}, pages 909--932, 2023.

\bibitem{BCDELP19}
Nir Bacrach, Keren Censor-Hillel, Michal Dory, Yuval Efron, Dean Leitersdorf,
  and Ami Paz.
\newblock Hardness of distributed optimization.
\newblock In {\em Proceedings of the 2019 ACM Symposium on Principles of
  Distributed Computing (PODC)}, pages 238--247, 2019.

\bibitem{balinski1967labelling}
Michel~L. Balinski.
\newblock Labelling to obtain a maximum matching.
\newblock In {\em Combinatorial Mathematics and Its Applications (Proceedings
  Conference Chapel Hill, North Carolina)}, pages 585--602, 1967.

\bibitem{BKS18}
Ran Ben-Basat, Ken-ichi Kawarabayashi, and Gregory Schwartzman.
\newblock Parameterized distributed algorithms.
\newblock In {\em Proceedings of the 33rd International Symposium on
  Distributed Computing (DISC)}, pages 6:1--6:16, 2018.

\bibitem{Blum1990}
Norbert Blum.
\newblock A new approach to maximum matching in general graphs.
\newblock In {\em Proceedings of the 17th International Colloquium on Automata,
  Languages, and Programming (ICALP)}, pages 586--597, 1990.

\bibitem{camerini1992random}
Paolo~M. Camerini, Giulia Galbiati, and Francesco Maffioli.
\newblock Random pseudo-polynomial algorithms for exact matroid problems.
\newblock {\em Journal of Algorithms}, 13(2):258--273, 1992.

\bibitem{CS22}
Yi-Jun Chang and Hsin-Hao Su.
\newblock Narrowing the {LOCAL-CONGEST} gaps in sparse networks via expander
  decompositions.
\newblock In {\em Proceedings of the 2022 ACM Symposium on Principles of
  Distributed Computing (PODC)}, pages 301--312, 2022.

\bibitem{chen2022maximum}
Li~Chen, Rasmus Kyng, Yang~P Liu, Richard Peng, Maximilian~Probst Gutenberg,
  and Sushant Sachdeva.
\newblock Maximum flow and minimum-cost flow in almost-linear time.
\newblock In {\em Proceedings of the 2022 IEEE 63rd Annual Symposium on
  Foundations of Computer Science (FOCS)}, pages 612--623. IEEE, 2022.

\bibitem{CKPSRSY21}
Lijie Chen, Gillat Kol, Dmitry Paramonov, Raghuvansh~R. Saxena, Zhao Song, and
  Huacheng Yu.
\newblock Almost optimal super-constant-pass streaming lower bounds for
  reachability.
\newblock In {\em the 53rd Annual ACM SIGACT Symposium on Theory of Computing
  (STOC)}, page 570–583, 2021.

\bibitem{cheung2014algebraic}
Ho~Yee Cheung, Lap~Chi Lau, and Kai~Man Leung.
\newblock Algebraic algorithms for linear matroid parity problems.
\newblock {\em ACM Transactions on Algorithms (TALG)}, 10(3):1--26, 2014.

\bibitem{chuzhoy2024faster}
Julia Chuzhoy and Sanjeev Khanna.
\newblock A faster combinatorial algorithm for maximum bipartite matching.
\newblock In {\em Proceedings of the 2024 Annual ACM-SIAM Symposium on Discrete
  Algorithms (SODA)}, pages 2185--2235. SIAM, 2024.

\bibitem{chuzhoy2024maximum}
Julia Chuzhoy and Sanjeev Khanna.
\newblock Maximum bipartite matching in $n^{2+o(1)}$ time via a combinatorial
  algorithm.
\newblock In {\em Proceedings of the 56th Annual ACM Symposium on Theory of
  Computing (STOC)}, pages 83--94, 2024.

\bibitem{Edmonds}
Jack Edmonds.
\newblock Paths, trees, and flowers.
\newblock {\em Canadian Journal of mathematics}, pages 449--467, 1965.

\bibitem{EKMS12}
Sebastian Eggert, Lasse Kliemann, Peter Munstermann, and Anand Srivastav.
\newblock Bipartite matching in the semi-streaming model.
\newblock {\em Algorithmica}, 63(1-2):490--508, 2012.

\bibitem{EKS09}
Sebastian Eggert, Lasse Kliemann, and Anand Srivastav.
\newblock Bipartite graph matchings in the semi-streaming model.
\newblock In {\em 17th Annual European Symposium on Algorithms (ESA)}, pages
  492--503, 2009.

\bibitem{even1975n2}
Shimon Even and Oded Kariv.
\newblock An {$O(n^{2.5})$} algorithm for maximum matching in general graphs.
\newblock In {\em Proceedings of the 16th Annual Symposium on Foundations of
  Computer Science (FOCS)}, pages 100--112. IEEE, 1975.

\bibitem{FFK21}
Salwa Faour, Marc Fuchs, and Fabian Kuhn.
\newblock {Distributed CONGEST Approximation of Weighted Vertex Covers and
  Matchings}.
\newblock In {\em 25th International Conference on Principles of Distributed
  Systems (OPODIS)}, pages 17:1--17:20, 2022.

\bibitem{FMJ22}
Manuela Fischer, Slobodan Mitrovi{\'c}, and Jara Uitto.
\newblock Deterministic $(1 + \epsilon)$-approximate maximum matching with
  $\mathrm{poly}(1/\epsilon)$ passes in the semi-streaming model and beyond.
\newblock In {\em Proceedings of the 54th Annual ACM SIGACT Symposium on Theory
  of Computing (STOC)}, pages 248--260, 2022.

\bibitem{gabow1976efficient}
Harold~N. Gabow.
\newblock An efficient implementation of {E}dmonds' algorithm for maximum
  matching on graphs.
\newblock {\em Journal of the ACM (JACM)}, 23(2):221--234, 1976.

\bibitem{gabow2017weighted}
Harold~N Gabow.
\newblock The weighted matching approach to maximum cardinality matching.
\newblock {\em Fundamenta Informaticae}, 154(1-4):109--130, 2017.

\bibitem{GT91}
Harold~N. Gabow and Robert~E. Tarjan.
\newblock Faster scaling algorithms for general graph matching problems.
\newblock {\em Journal of the ACM (JACM)}, pages 815--853, 1991.

\bibitem{GG23}
Mohsen Ghaffari and Christoph Grunau.
\newblock Faster deterministic distributed mis and approximate matching.
\newblock In {\em the 55th Annual ACM Symposium on Theory of Computing (STOC)},
  page 1777–1790, 2023.

\bibitem{goldberg2004maximum}
Andrew~V. Goldberg and Alexander~V. Karzanov.
\newblock Maximum skew-symmetric flows and matchings.
\newblock {\em Mathematical Programming}, 100:537--568, 2004.

\bibitem{GPV93}
A.V. Goldberg, S.A. Plotkin, and P.M. Vaidya.
\newblock Sublinear-time parallel algorithms for matching and related problems.
\newblock {\em Journal of Algorithms}, 14(2):180--213, 1993.

\bibitem{Harris19}
David~G. Harris.
\newblock Distributed local approximation algorithms for maximum matching in
  graphs and hypergraphs.
\newblock In {\em Proceedings of the 60th {IEEE} Annual Symposium on
  Foundations of Computer Science (FOCS)}, pages 700--724, 2019.

\bibitem{harvey2009algebraic}
Nicholas J.~A. Harvey.
\newblock Algebraic algorithms for matching and matroid problems.
\newblock {\em SIAM Journal on Computing}, 39(2):679--702, 2009.

\bibitem{HK73}
John~E. Hopcroft and Richard~M. Karp.
\newblock An $n^{5/2}$ algorithm for maximum matchings in bipartite graphs.
\newblock {\em SIAM Journal on Computing}, pages 225--231, 1973.

\bibitem{HS23}
Shang-En Huang and Hsin-Hao Su.
\newblock {$(1-\epsilon)$}-approximate maximum weighted matching in
  poly($1/\epsilon$, $\log n$) time in the distributed and parallel settings.
\newblock In {\em the 2023 ACM Symposium on Principles of Distributed Computing
  (PODC)}, page 44–54, 2023.

\bibitem{II86}
Amos Israeli and Alon Itai.
\newblock A fast and simple randomized parallel algorithm for maximal matching.
\newblock {\em Information Processing Letters}, pages 77--80, 1986.

\bibitem{IKY24}
Taisuke Izumi, Naoki Kitamura, and Yutaro Yamaguchi.
\newblock A nearly linear-time distributed algorithm for exact maximum
  matching.
\newblock In {\em Proceedings of the 2024 Annual ACM-SIAM Symposium on Discrete
  Algorithms (SODA)}, pages 4062--4082, 2024.

\bibitem{KI22}
Naoki Kitamura and Taisuke Izumi.
\newblock A subquadratic-time distributed algorithm for exact maximum matching.
\newblock {\em IEICE Transactions on Information and Systems}, 105(3):634--645,
  2022.

\bibitem{FTR06}
Fabian Kuhn, Thomas Moscibroda, and Roger Wattenhofer.
\newblock The price of being near-sighted.
\newblock In {\em Proceedings of the 17th Annual ACM-SIAM Symposium on Discrete
  Algorithms (SODA)}, pages 980--989, 2006.

\bibitem{KMW16}
Fabian Kuhn, Thomas Moscibroda, and Roger Wattenhofer.
\newblock Local computation: Lower and upper bounds.
\newblock {\em Journal of the ACM (JACM)}, pages 1--44, 2016.

\bibitem{LPP08}
Zvi Lotker, Boaz Patt-Shamir, and Seth Pettie.
\newblock Improved distributed approximate matching.
\newblock {\em Journal of the ACM (JACM)}, pages 1--17, 2015.

\bibitem{McGregor05}
Andrew McGregor.
\newblock Finding graph matchings in data streams.
\newblock In {\em Approximation, Randomization and Combinatorial Optimization.
  Algorithms and Techniques (APPROX-RANDOM)}, pages 170--181, 2005.

\bibitem{MV1980}
Silvio Micali and Vijay~V. Vazirani.
\newblock An ${O}(\sqrt{V}{E})$ algorithm for finding maximum matching in
  general graphs.
\newblock In {\em Proceedings of the 21st Annual Symposium on Foundations of
  Computer Science (FOCS)}, pages 17--27, 1980.

\bibitem{MMSS25}
Slobodan Mitrovi\'{c}, Anish Mukherjee, Piotr Sankowski, and Wen-Horng Sheu.
\newblock Faster semi-streaming matchings via alternating trees.
\newblock In {\em 52nd International Colloquium on Automata, Languages, and
  Programming (ICALP)}, pages 119:1--119:19, 2025.

\bibitem{MS25}
Slobodan Mitrovi\'{c} and Wen-Horng Sheu.
\newblock A framework for boosting matching approximation: parallel,
  distributed, and dynamic.
\newblock In {\em the 37th ACM Symposium on Parallelism in Algorithms and
  Architectures (SPAA)}, page 443–457, 2025.

\bibitem{mucha2004maximum}
Marcin Mucha and Piotr Sankowski.
\newblock Maximum matchings via {G}aussian elimination.
\newblock In {\em Proceedings of the 45th Annual IEEE Symposium on Foundations
  of Computer Science (FOCS)}, pages 248--255. IEEE, 2004.

\bibitem{mulmuley1987matching}
Ketan Mulmuley, Umesh~V. Vazirani, and Vijay~V. Vazirani.
\newblock Matching is as easy as matrix inversion.
\newblock {\em Combinatorica}, 7:105--113, 1987.

\bibitem{rabin1989maximum}
Michael~O. Rabin and Vijay~V. Vazirani.
\newblock Maximum matchings in general graphs through randomization.
\newblock {\em Journal of Algorithms}, 10(4):557--567, 1989.

\bibitem{sato2025exact}
Ryotaro Sato and Yutaro Yamaguchi.
\newblock Exact matching in matrix multiplication time.
\newblock {\em arXiv preprint arXiv:2508.04081}, 2025.

\bibitem{Tirodkar18}
Sumedh Tirodkar.
\newblock Deterministic algorithms for maximum matching on general graphs in
  the semi-streaming model.
\newblock In {\em 38th IARCS Annual Conference on Foundations of Software
  Technology and Theoretical Computer Science (FSTTCS)}, volume 122, pages
  39:1--39:16, 2018.

\bibitem{van2023deterministic}
Jan Van Den~Brand, Li~Chen, Rasmus Kyng, Yang~P Liu, Richard Peng,
  Maximilian~Probst Gutenberg, Sushant Sachdeva, and Aaron Sidford.
\newblock A deterministic almost-linear time algorithm for minimum-cost flow.
\newblock In {\em Proceedings of the 2023 IEEE 64th Annual Symposium on
  Foundations of Computer Science (FOCS)}, pages 503--514. IEEE, 2023.

\bibitem{vazirani1994theory}
Vijay~V. Vazirani.
\newblock A theory of alternating paths and blossoms for proving correctness of
  the {$O(\sqrt{V}E)$} general graph maximum matching algorithm.
\newblock {\em Combinatorica}, 14(1):71--109, 1994.

\bibitem{Vazirani12}
Vijay~V. Vazirani.
\newblock An improved definition of blossoms and a simpler proof of the {MV}
  matching algorithm.
\newblock {\em arXiv preprint arXiv:1210.4594}, 2012.

\bibitem{vazirani2024theory}
Vijay~V. Vazirani.
\newblock A theory of alternating paths and blossoms from the perspective of
  minimum length.
\newblock {\em Mathematics of Operations Research}, 49(3):2009--2047, 2024.

\bibitem{WW05}
Mirjam Wattenhofer and Roger Wattenhofer.
\newblock Distributed weighted matching.
\newblock In {\em International Symposium on Distributed Computing (DISC)},
  pages 335--348, 2004.

\end{thebibliography}

\end{document}